%% file: ms.tex
\documentclass{article}

\usepackage{fullpage}

\usepackage{amsmath}
\usepackage{amssymb}
\usepackage{amsthm}
\usepackage[linesnumbered, lined, boxed]{algorithm2e}
\usepackage{float}
\usepackage{verbatim}
\usepackage{natbib}
\newcommand{\R}{\mathbb{R}}
\usepackage{xcolor}
\usepackage[T1]{fontenc}
\usepackage{thmtools}
\usepackage{thm-restate}
\usepackage{hyperref}
\usepackage{cleveref}
\crefname{subsection}{subsection}{subsections}

\allowdisplaybreaks

\newcommand{\Cells}{\mathcal{C}}
\newcommand{\Light}{\mathcal{L}}
\newcommand{\Heavy}{\mathcal{H}}
\newcommand{\Medium}{\mathcal{M}}
\newcommand{\Ancestors}{\mathrm{Anc}}
\newcommand{\Children}{\mathrm{Ch}}

\newcommand{\Ex}{\mathbb{E}}
\newcommand{\OPT}{\,\mathrm{OPT}\,}
\DeclareMathOperator{\poly}{poly}
\DeclareMathOperator{\Diam}{Diam}

\newcommand{\V}{\mathcal{V}}
\newcommand{\X}{\mathcal{X}}

\newcommand{\BTG}{\mathsf{Bitstogram}}
\newcommand{\HSO}{\mathsf{HeavySumsOracle}}
\newcommand{\HCM}{\mathsf{HeavyCellMarker}}
\newcommand{\CH}{\mathsf{CH}}
\newcommand{\CCH}{\mathsf{CCH}}

\newcommand{\BH}{\mathsf{BH}}
\newcommand{\BSO}{\mathsf{BSO}}
\newcommand{\SH}{\mathsf{SH}}
\newcommand{\SKM}{\mathsf{Standard}\;k-\mathsf{Means}}
\newcommand{\snap}[2]{G_{#2}(#1)}

\newcommand{\Count}{\mathrm{Count}}
\newcommand{\Sum}{\mathrm{Sum}}

\newcommand{\Z}{\mathbb{Z}}

\newcommand{\PH}{\mathsf{PH}}
\newcommand{\PSO}{\mathsf{PSO}}

\newcommand{\erralpha}{{err}}

\DeclareMathOperator*{\argmin}{arg\,min}

\usepackage{bbm}

\newtheorem{theorem}{Theorem}[section]
\newtheorem{lemma}[theorem]{Lemma}
\newtheorem{corollary}[theorem]{Corollary}

\theoremstyle{definition}
\newtheorem{definition}[theorem]{Definition}
\newtheorem{remark}[theorem]{Remark}

\makeatletter
\theoremstyle{plain}

\theoremstyle{plain}

\theoremstyle{definition}

\makeatother

\usepackage{babel}
\providecommand{\definitionname}{Definition}
\providecommand{\lemmaname}{Lemma}
\providecommand{\theoremname}{Theorem}

\makeatletter
\algocf@newcmdside@kobe{LetIn@let}{%
    \KwSty{let}%
    \ifArgumentEmpty{#1}\relax{ #1}%
    \algocf@group{#2}%
    \par
}
\algocf@newcmdside@kobe{LetIn@in}{%
    \KwSty{in}%
    \ifArgumentEmpty{#1}\relax{ #1}%
    \algocf@block{#2}{end}{#3}%
    \par
}
\newcommand\LetIn[2]{%
    \LetIn@let{#1}%
    \LetIn@in{#2}%
}
\makeatother

\title{Locally Private $k$-Means Clustering with Constant Multiplicative Approximation and Near-Optimal Additive Error}
\author {
        Anamay Chaturvedi\thanks{Khoury College of Computer Sciences, Northeastern University \dotfill \texttt{chaturvedi.a@northeastern.edu}} \and
        Matthew Jones\thanks{Khoury College of Computer Sciences, Northeastern University \dotfill \texttt{jones.m@northeastern.edu}} \and
        Huy L. Nguy\~{\^{e}}n\thanks{Khoury College of Computer Sciences, Northeastern University \dotfill \texttt{hu.nguyen@northeastern.edu}}
}

\usepackage{mathtools}

\newcommand{\DeclareAutoPairedDelimiter}[3]{%
  \expandafter\DeclarePairedDelimiter\csname Auto\string#1\endcsname{#2}{#3}%
  \begingroup\edef\x{\endgroup
    \noexpand\DeclareRobustCommand{\noexpand#1}{%
      \expandafter\noexpand\csname Auto\string#1\endcsname*}}%
  \x}

\DeclareAutoPairedDelimiter{\size}{\lvert}{\rvert}

\DeclarePairedDelimiter\abs{\lvert}{\rvert}%
\DeclarePairedDelimiter\norm{\lVert}{\rVert}%

\makeatletter
\let\oldabs\abs
\def\abs{\@ifstar{\oldabs}{\oldabs*}}

\let\oldnorm\norm
\def\norm{\@ifstar{\oldnorm}{\oldnorm*}}
\makeatother

\begin{document}
\maketitle

\begin{abstract}
    Given a data set of size $n$ in $d'$-dimensional Euclidean space, the $k$-means problem asks for a set of $k$ points (called centers) so that the sum of the $\ell_2^2$-distances between points of a given data set of size $n$ and the set of $k$ centers is minimized. Recent work on this problem in the locally private setting achieves constant multiplicative approximation with additive error $\tilde{O} (n^{1/2 + a} \cdot k \cdot \max \{\sqrt{d}, \sqrt{k} \})$ and proves a lower bound of $\Omega(\sqrt{n})$ on the additive error for any solution with a constant number of rounds. In this work we bridge the gap between the exponents of $n$ in the upper and lower bounds on the additive error with two new algorithms. Given any $\alpha>0$, our first algorithm achieves a multiplicative approximation guarantee which is at most a $(1+\alpha)$ factor greater than that of any non-private $k$-means clustering algorithm with $k^{\tilde{O}(1/\alpha^2)} \sqrt{d' n} \poly\log n$ additive error. Given any $c>\sqrt{2}$, our second algorithm achieves $O(k^{1 + \tilde{O}(1/(2c^2-1))} \sqrt{d' n} \poly\log n)$ additive error with constant multiplicative approximation. Both algorithms go beyond the $\Omega(n^{1/2 + a})$ factor that occurs in the additive error for arbitrarily small parameters $a$ in previous work, and the second algorithm in particular shows for the first time that it is possible to solve the locally private $k$-means problem in a constant number of rounds with constant factor multiplicative approximation and polynomial dependence on $k$ in the additive error arbitrarily close to linear.
\end{abstract}

\input{1Intro}

\input{2Preliminaries}

\input{3AlgorithmOneRound}

\input{4AlgorithmMultiRound}

\appendix

\input{5Appendix}

\bibliographystyle{plainnat}
\bibliography{ms.bib}
	
\end{document}

%% file: 1Intro.tex
\section{Introduction}
    
    Given $n$ points in a $d$-dimensional Euclidean space, the $k$-means clustering problem asks for a set of $k$ points $S$ such that the sum of $\ell_2^2$-distances from each data point to the closest respective point in $S$ is minimized. Although $k$-means clustering in the non-private setting is well-studied, over the past few years there have been several developments in the differentially private (DP) setting. Differential privacy \citep{dwork2006our} provides a framework to characterize the loss in privacy which occurs when sensitive data is processed and the output of this computation is revealed publicly. Although there are different ways to define and capture this loss in privacy, broadly speaking these characterizations tend to be either central or local in nature.
    
    Informally, differential privacy asks for a guarantee that the likelihood of any possible output does not change too much by adding to or dropping from our data set any possible private value from the data universe. In any private algorithm such a guarantee is fulfilled by adding carefully calibrated noise to quantities that are information-theoretically sensitive to the private data in the course of the computation, and under the constraints of being private the goal is to achieve relatively low error. Perfect answers to the algorithmic problem at hand typically violate privacy; as a consequence, the constraints of privacy usually enforce harsher lower bounds on accuracy or utility than those imposed by the limits of time or sample efficient computation.
    
    In local differentially privacy (LDP) the constraints are even more severe; the entity solving the algorithmic problem only gets access to noisy, privatized data. This constraint forces even stronger lower bounds on the accuracy of locally private protocols; for the $k$-means clustering problem a lower bound of $\Omega(\sqrt{n})$ is known for the additive error of any interactive constant factor multiplicative approximation algorithm \citep{stemmer2020locally}.
    
    \paragraph{Recent work on LDP $k$-means}{The first LDP algorithm for the $k$-means problem with provable guarantees was given by \cite{nissim2018clustering} wherein they achieved a multiplicative approximation of $O(k)$ and and an additive error term of $\tilde{O}(n^{2/3 + a} \cdot d^{1/3} \cdot \sqrt{k})$. They achieved this result by solving the related 1-cluster problem that asks the solver to privately allocate a small number of centers so that some center in that set covers all data points within a ball of minimal radius; by an observation of \cite{feldman2017coresets}, there is a general algorithm that given access to a private solution for the 1-cluster problem solves the private $k$-means problem. The exponent of $n$ in the additive error term holds for arbitrarily small $a$ at the cost of looser multiplicative approximation guarantees; this artefact is the consequence of using \emph{locality sensitive hashing} (LSH), something which appears in most later work as well.}
    
    \cite{kaplan2018differentially} gave the first constant factor multiplicative approximation algorithm for this problem within an additive error of $\tilde{O}(n^{2/3 + a} \cdot d^{1/3} \cdot k^2)$. They refine the approach of the previous work by specifically targeting the $k$-means problem but also use LSH functions to detect the accumulation of data. The additive error was further brought down by \cite{stemmer2020locally}, who achieved an additive error of $\tilde{O} (n^{1/2 + a} \cdot k \cdot \max \{ \sqrt{d}, \sqrt{k} \}$ and also proved a lower bound of $\Omega (\sqrt{n})$, as mentioned before. Given that all previous works exhibit some trade-off between the exponent of $n$ and the multiplicative approximation, the exponents of $1/2+a$ and $1/2$ in the upper and lower bounds of \cite{stemmer2020locally} is particularly provocative. It naturally leads to the question
    \begin{center}
        \textbf{Does there exist an LDP $k$-means clustering algorithm with constant multiplicative approximation and additive error with a $\sqrt{n}$ dependence on the size of the data set?}
    \end{center}
    
    In the non-private setting it has been seen that the performance of $k$-means clustering algorithms is usually not very sensitive to the multiplicative approximation guarantee, unless the data set is chosen in a pathological fashion. Experimental work \citep{balcan2017differentially, chaturvedi2020differentially} on $k$-means clustering in the related central model of DP shows that the performance of private clustering algorithms seems to be far more sensitive to the additive error, which as we have observed is bound to exist due to the constraints of being private. This highlights the importance of the question of determining the true dependence of the additive error term on the size of the data set.
    
    \begin{table}
    \centering
    \caption{Comparison with recent LDP algorithms for $k$-means}
    \begin{tabular}{c|c|c}\hline
         Work & Multiplicative Approximation & Additive Error \\ \hline
         \cite{nissim2018clustering} & $O(k)$ & $\tilde{O} (n^{2/3 + a} \cdot d'^{1/3} \cdot \sqrt{k})$ \\
         \cite{kaplan2018differentially} & $O(1)$ & $\tilde{O} (n^{2/3 + a} \cdot d'^{1/3} \cdot k^2 )$ \\
         \cite{stemmer2020locally} & $O(1)$ & $\tilde{O} (n^{1/2 + a} \cdot k \cdot \max \{ \sqrt{d'}, \sqrt{k} \})$ \\
         This work, \cref{alg:1Round} & $(1+\alpha) \eta$  & $\tilde{O} (n^{1/2} \cdot d'^{1/2} \cdot k^{\tilde{O}(1/\alpha^2)})$ \\
         This work, \cref{alg:main} & $O(c^2)$ & $\tilde{O} (n^{1/2} \cdot d'^{1/2} \cdot k^{1 + O(1/(2c^2-1))} )$ \\\hline
    \end{tabular}\\
    The additive error assumes a data set of size $n$ inside a ball with unit radius. The $\tilde{O}$ notation hides dependence on the privacy parameters, the failure probability, and $\log$ terms. The user-defined parameter $c$ can take any real value greater than $\sqrt{2}$.
    \end{table}
    
    \paragraph{Technical contributions:}{In this work we present two algorithms for the $k$-means problem in the LDP setting, wherein we go beyond the $n^{1/2 + a}$ barrier demonstrating that the trade-off can be independent of $n$ for some regimes of $k$ and $n$. Our first algorithm is a one-round protocol that achieves a $(1+\alpha)$-multiplicative approximation to the cost guarantee of any non-private clustering algorithm that it is given acess to as a subroutine. It achieves additive error $k^{\tilde{O}(1/\alpha^2)} \sqrt{d' n} \poly\log n$}; we see that the trade-off between the additive and multiplicative approximations in this algorithm has been shifted from $n$ to $k$. However, the $\tilde{O}$ term in the exponent of $k$ can hide large constants, which is an undesirable property in a setting where low additive error seems to dictate performance.
    
    \begin{restatable}{theorem}{oneRoundGuarantee}\label{thm:1RoundGuarantee}
        \Cref{alg:1Round} is an $(\epsilon,\delta)$-locally differentially private algorithm that after one round of interaction with a private distributed data set $D' \subset \R^{d'}$ of size $n$, outputs a set $S'$ of size $k$ such that for failure probability polynomially small in $n$,
        \begin{align*}
            f_{D'} (S') &\leq (1 + O(\alpha))\eta \OPT' + \frac{1}{\epsilon} k^{\tilde{O} (1/\alpha^2)} \sqrt{d' n \log 1/\delta} \poly\log n.
        \end{align*}
    \end{restatable}
    
    We address this deficit with our second algorithm, where we return to an LSH-based approach and drive down the exponent of $k$ to $1 + O(1/(2c^2-1))$. Again as this exponent approaches $1$ the multiplicative approximation factor blows up but this shows for the first time that it is possible to have constant factor multiplicative approximation $k$-means clustering algorithms in the LDP setting with additive error that has a truly square-root dependence on the data set size and the ambient dimension and arbitrarily close to linear dependence on the number of cluster centers.
    
     \begin{restatable}{theorem}{fourRoundGuarantee}\label{thm:fourRoundGuarantee}
        \Cref{alg:main} is an $(\epsilon,\delta)$-locally differentially private algorithm  such that given $c > \sqrt{2}$, after four rounds of interaction with a private distributed data set $D' \subset \R^{d'}$ of size $n$ outputs a set $S'$ of size $k$ such that with probability $1-\beta$,
        \begin{align*}
            f_{D'} (S') &= O(\OPT') + O \left(\frac{1}{\epsilon} \sqrt{d' n \ln(n/\delta)} \right) \left(\frac{k \poly\log n}{\beta }\right)^{1 + O(1/(2c^2-1))}.
        \end{align*}
    \end{restatable}
    
    It was observed in \cite{stemmer2020locally} that one of the main road-blocks in computing solutions with low additive error is figuring out how to generate a relatively small bi-criteria solution to the $k$-means problem as a first step. A bi-criteria solution relaxes two constraints of the $k$-means problem; we permit picking more than $k$ centers, and we relax the minimum cost requirement to a multiplicative approximation guarantee. Any such bi-criteria solution can be exploited to construct a proxy data set on which we can apply any non-private $k$-means clustering algorithm. The fact that the clustering cost of the original data set with respect to the candidate centers used to generate the proxy data set can be exploited to show that $k$-means solutions for proxy data sets work well for the original data set as well. In order to avoid an exponent of $1/2 + a$ on $n$, it is necessary to find a bi-criteria solution with $O(\poly k \poly \log n)$ many candidate centers such that the additive error to their respective multiplicative approximations is at most $O(\poly k \sqrt{n} \poly \log n)$ (omitting the dependence on dimension). Both our algorithms achieve their improvements by generating such a small size bi-criteria solutions for the $k$-means problem.
    
    \paragraph{LDP $k$-means with arbitrarily tight multiplicative approximation:} In our first algorithm, we appeal to recent advances in dimension reduction for $k$-means clustering \cite{DBLP:conf/stoc/MakarychevMR19} which show that Johnson-Lindenstrauss style dimension reduction to $\tilde{O}(\log k/\alpha^2)$ preserves the cost of every $k$-clustering of a data set within a multiplicative approximation of $(1\pm \alpha)$. Suppose we decompose the domain in concentric shells depending on their distance from some fixed $k$ cluster centers. By setting geometric thresholds of $1,1/2, 1/4$ units and so on, the $l$th ring has the property that every data point in that shell has a clustering cost of $O(1/(2^l)^2)$ units. To cluster the $l$th ring, we allocate a number of centers by appealing to a grid-based approach following \cite{chaturvedi2020differentially}; since we were able to reduce dimensions to $\tilde{O}(\log k/\alpha^2)$ we are able to show that allocating $k^{\tilde{O}(1/\alpha^2)} \poly\log n$ centers suffices to ensure that most points in the $l$th shell are within an $O(\alpha/(2^l)^2)$ distance of some candidate center.
    
    Extending this for every shell with appropriately scaled grids we get the promise that moving each data point to its closest candidate center would lead to net movement of $O(\alpha \OPT)$ where $\OPT$ is the optimal clustering cost. The rest of the argument follows essentially by applications of the triangle inequality to prove that the dimension reduced and proxy data sets have similar costs for any candidate $k$-means solutions.
    
    \paragraph{LDP $k$-means with low additive error:} We note that the constant absorbed by the $\tilde{O}$ term in the exponent of $k$ of our first algorithm could be large, which is an undesirable property in a setting where low additive error seems to dictate performance. We address this deficit with our second algorithm, where we return to an LSH-based approach and drive down the exponent of $k$ to $1 + O(1/(2c^2-1))$. This shows for the first time that it is possible to have constant factor multiplicative approximation $k$-means clustering algorithms in the LDP setting with additive error that has a square-root dependence on the data set size and the ambient dimension (up to $\log$ factors) and arbitrarily close to linear dependence on the number of cluster centers.
    
    We achieve this improvement by appealing to a construction of \cite{BFLSY17} who impose a randomly-shifted hierarchy of dyadic cells in a dimension reduced space. A tree structure is defined on subsets of the domain $[0,1)^d$; starting with $[0,1)^d$ as the root node, we bisect the hypercube along each axis to generate $2^d$ congruent octants. Each octant is itself a hypercube that we designate a child of the original cell, on proceeding recursively for $\log n$ levels the side length of each cell in the lowest level is $<1/n$.
    
    The crucial observation made by \cite{BFLSY17} was that after a uniformly random shift of the tiling there are $O(1)$ cells with side-length $t$ units within a distance of $t/d$ units of any point. By applying this observation to an optimal $k$-means solution, we are able to identify a small number of cells where the data accumulates per level. These cells serve as our domains for LSH functions. The number of data points that accumulate in these cells scales inversely with the side-length of the cells; this ensures that we only allocate centers when such an allocation is certain to be helpful. We are able to allocate a far smaller number of centers to generate our bi-criteria solution than in our first algorithm. Moving the ${1/2 + a}$-style exponent from $n$ to $k$ is technically involved and we give a more detailed explanation in \cref{sec:fourRound}.
    
    \paragraph{Challenges of the local setting:} We recall that in the locally private setting, each agent must add noise to any response they give under the assumption that it is public knowledge that all data lies in a domain of diameter $1$. This will require adding a noise vector with length proportional to $1/\epsilon$ to their private data if they were to $\epsilon$-privately release their point directly. The implications of the large noise needed to obfuscate information means that it is impossible to privately derive fine-grained information about where individual points lie.
    
    It follows from these considerations that we must try and get aggregate information about the geometry of the data set indirectly. One way of accomplishing this is to \emph{discretize} the agents' response. Although again the privatized individual responses are highly noisy, since the range of values taken by this discretized response is finite the slight bias towards values which are \emph{heavy-hitters} causes their counts to accumulate and be distinguishable from the counts of false positives. We will appeal to prior work on locally private succinct histogram recovery to recover such heavy hitting values with minimal loss in privacy.
    
    From this perspective, we see that in the first algorithm we achieve our discretization by dividing our space using proximity to grid points, and in the second algorithm we use a two different kinds of discretization; a cell based discretization which is philosophically similar to that of the first, and an LSH-based discretization which gives a geometrically meaningful response not in terms of the ambient space but instead in terms of the rest of the data set.
    
    \paragraph{Reducing round-complexity via $\HSO$:} In the course of our algorithms we often encounter a situation where we first identify some subset of the data domain that is advantageous for us to allocate a candidate center in and then we need to compute a vector average over points in that domain. Although naively performing such a computation would require two rounds in the LDP setting, we construct a subroutine that can be run in parallel with the succinct histograms used to identify such regions of the data domain, and can be queried to estimate the vector sums of all points mapping to such domains. Dividing these sums by the histogram counts yields the averages we need. Indeed, our construction is in fact a bit more general, and allows one to recover sums of arbitrary private vector values for all points that map to some heavy hitting value under a completely different value mapping. This construction allows us to compute vector averages over points mapping to heavy LSH buckets as well as vector averages in the original space over all points that map to a certain cluster in the dimension-reduced space; the two value mappings need not have anything to do with each other.
    
    \paragraph{Concurrent work:}{ In \cite{chang2021locally} a one-round protocol for LDP $k$-means with similar cost guarantees as \cref{alg:1Round} is introduced, also surpassing the $n^{1/2 + a}$ barrier mentioned above. They operate in the $\epsilon$-DP setting and get a multiplicative approximation of $\eta (1 + \alpha)$ where $\eta$ is the multiplicative approximation guarantee of any given non-private $k$-means algorithm and an additive error term of $k^{O_{\alpha} (1)} \cdot \sqrt{n d'} \cdot \poly\log (n) /\epsilon$. They also demonstrate that their protocol can be extended to the shuffle model~\citep{bittau2017prochlo, cheu2019distributed, erlingsson2019amplification} of differential privacy.
    }
    
    \paragraph{Outline of paper:}{In \cref{sec:prelims} we start by formalizing the problem statement and the definition of LDP that our algorithms must fulfill. We then summarize some notation that eases the description of our analysis and recall private subroutines from previous work. We also introduce the $\HSO$, a one-round protocol that can be run in parallel with a private succinct histogram and privately constructs a data structure that may be queries to recover sums of vector-function values taken by all agents that happen to map to a heavy-hitting value in the succinct histogram. We recall the LSH function definition and prove some fundamental properties of the construction we use in \cref{sec:fourRound}.
    
    In \cref{sec:1Round} we introduce our LDP $k$-means algorithm for arbitrarily tight multiplicative approximation, \cref{alg:1Round}. We start by establishing the pseudo-code and outlining the main steps, and then give a technical discussion explaining some of the algorithmic choices made as well as sketching why the cost analysis works out. We then give a formal proof of the cost and privacy guarantees. The main result of this section is \cref{thm:1RoundGuarantee}.
    
    In \cref{sec:fourRound}, we introduce our LSH-based LDP $k$-means algorithm, \cref{alg:main}. We start by giving a high level overview of the core ideas and advantages behind our algorithmic choices. We provide the pseudo-code in a modular fashion and analyse the cost guarantees of each subroutine in a separate subsection. The main result of this section is \cref{thm:fourRoundGuarantee}.}

%% file: 2Preliminaries.tex
\section{Preliminaries}\label{sec:prelims}

    \subsection{Problem Definition}
    
    We start by formally defining the $k$-means clustering problem. 
    
    \begin{definition}[Non-private $k$-means]
        For any Euclidean space $E$, let $z : E\times E  \rightarrow \R$ denote the square of the $\ell_2$ metric. Let $D'$ be a data set of $n$ points in $\R^{d'}$ such that $D' \subset B(0,1)$, the $d'$-dimensional unit ball of radius $1$ centered at the origin. The $k$-means clustering cost $f_{D'} (S)$ of the data set $D'$ for a set $S$ of $k$ points in $B(0,1)$ is defined by the expression
        \begin{align*}
            f_{D'} (S) = \sum_{p \in D'} z(p,S)
        \end{align*}
        where we let $z(p,S) = \min_{q \in S} z(p,q)$. The $k$-means clustering problems asks one to find a set of $k$ points in $B(0,1)$ such that $f_{D'} (\cdot)$ is minimized.
    \end{definition}
    
    \begin{remark}
        Both algorithms introduced in this work start with a dimension reduction so it will be convenient to let $d'$ denote the dimension of the given ambient space and $d$ denote the dimension of the space that the majority of the computation is done in. Similarly, $D'$ is used to denote the original data set and $D$ is used to denote the image of the data set in the dimension reduced space.
    \end{remark}
    
    We require that our $k$-means algorithm also satisfy the constraints of local differential privacy. In this framework, the dataset is distributed among $n$ agents each of whom has a single point of $D$, and the constraint of being locally differentially private requires that the transcript of any agent's responses is not too sensitive to their private data. This is formalized by appealing to the central model of differential privacy, which in turn is defined as follows:

    \begin{definition}[Differential privacy (DP), \cite{dwork2006our}]
        Two datasets $D_1, D_2 \in \mathcal{X}^n$ are \emph{neighbouring} if they differ in at most one member element, i.e. $\size{D_1 \triangle D_2} = 1$. An algorithm $A:\mathcal{X} \to \mathcal{Y}$ is said to be \emph{$(\epsilon,\delta)$-differentially private} (DP) if for any $S \subset \mathcal{Y}$ and any two neighbouring datasets $D_1,D_2 \in \mathcal{X}$,
        \begin{align*}
            P(A(D_1) \in S) \leq \exp(\epsilon) P(A(D_2) \in S) + \delta.
        \end{align*}
        If $\delta = 0$, we can say that $A$ is $\epsilon$-differentially private.
    \end{definition}

    Given the definition of the central model of differential privacy, local differential privacy is then defined as follows:
    
    \begin{definition}[Local differential privacy (LDP), \cite{kasiviswanathan2011can}]
        Consider a protocol which interacts with any one agent in some $r$ rounds, and let the response of the agent with private data $p$ be $A (p) = (A_1 (p), \dots, A_r (p))$, where $A_i(p)$ is the response of the agent in the $i$th round. We say that this protocol is \emph{$(\epsilon,\delta)$-locally differentially private} (LDP) if the algorithm that outputs privatized responses for any agent $p \mapsto A(p)$ is $(\epsilon,\delta)$-differentially private. Again, if $\delta = 0$, we can say that a protocol is $\epsilon$-locally differentially private.
    \end{definition}
    
    \begin{remark}[Notation]
        We use $\tilde{O} (\cdot)$ to denote that certain terms have been suppressed in the argument. Concretely, in this notation we omit terms that are logarithmic in the multiplicative approximation factor $\alpha$, the failure probability $\beta$ and $\log n$. We use the expression $\poly\log n$ to denote terms that are $O(\log^p n)$ for some constant power $p$.
    \end{remark}
    
    \subsection{Dimension reduction for \texorpdfstring{$k$}{k}-means clustering}
    
    In this subsection we recall some results about distance preserving dimension reduction maps that are fundamental to the construction of both algorithms described in this work. We follow the description in \cite{DBLP:conf/stoc/MakarychevMR19}, where the state of the art for the application of dimension reduction to $\ell_p$ clustering is stated and proved. We adopt the notation that for any $x,y, \alpha \in \R$, $x \simeq_{1 + \alpha} y$ if $\frac{x}{1+\alpha} \leq y \leq (1+\alpha)x$, note that for any $x,y$, for all sufficiently small $\alpha$, this is equivalent to $y = (1 \pm O(\alpha) x$.
    
    \begin{lemma}[Johnson-Lindenstrauss lemma, \cite{johnson1984extensions}]
        \label{lem:JL}
        There is a family of random linear maps $T_{d',d}: R^{d'} \to R^d$ with the property that for every $d'\geq 1$, $\alpha, \beta \in (0,1/2)$ and all $x \in \R^{d'}$,
        \begin{align*}
            P_{T \sim T_{d',d}} \left( \norm{Tx} \in \left[\frac{\norm{x}}{1+\alpha}, (1+\alpha)\norm{x}\right] \right) \geq 1-\beta,
        \end{align*}
        where $d = O\left( \frac{\log (1/\beta)}{\alpha^2} \right)$.
    \end{lemma}
    
    This result is often cited in the form that for a data set $D' \subset \R^{d'}$ of size $n$, in order to preserve all pair-wise distances with probability $1-\beta$ it suffices to reduce dimensions to $O\left( {\log (n/\beta)}/{\alpha^2} \right)$; this version follows directly by scaling the failure probability for the statement above by a factor of $1/n$ and applying the union bound. It is a well-known fact that the $k$-means clustering function can be written entirely in terms of pair-wise distances between the points in each cluster, i.e. for a $k$-means solution $S$ that induces a partition $(C_1,\dots, C_k)$,
     \begin{align*}
        f_{D'}(S) &= \sum_{p\in D'} z(p,S) = \sum_{i = 1}^k \sum_{p \in C_i} z(p,\mu_i) = \sum_{i=1}^k \frac{1}{2|C_i|}\sum_{p,q \in C_i}z(p,q).
    \end{align*}
    It follows that preserving $\ell_2$ distances within a $(1\pm \alpha)$ approximation guarantees that the $k$-means clustering cost for the same cluster sets is preserved within a $(1\pm O(\alpha))$ factor. This is the formulation that we appeal to for the multi-round clustering algorithm with low additive error.
    
    For the purpose of $k$-$\ell_p$ clustering it has been shown that one can reduce dimensions far more aggressively; this line of work culminates in the following near-optimal result of \cite{DBLP:conf/stoc/MakarychevMR19}:
    \begin{theorem}[Theorem 1.3 of \cite{DBLP:conf/stoc/MakarychevMR19}]
        \label{lem:Makarychev}
            Any family of linear maps $T_{d',d}: R^{d'} \to R^d$ that satisfies the conditions of the JL lemma and is sub-Gaussian tailed has the property that for any clustering $(C_1,\dots, C_k)$ of $D'$ with probability $1-\beta$ over the choice of $T \sim T_{d',d}$
            \begin{align*}
                 \sum_{i=1}^k \frac{1}{2|S_i|}\sum_{p,q \in S_i} z(p,q) =  \left(\sum_{i=1}^k \frac{1}{2|S_i|}\sum_{p,q \in S_i} z(Tp,Tq)\right) \left[\frac{1}{1+\alpha},(1+\alpha )\right].
            \end{align*}
    where $d = O(\log (k/\alpha \beta)/\alpha^2)$.
    \end{theorem}
    
    We recall that a family of linear maps $T_{d',d}$ is called sub-Gaussian-tailed if for every unit vector $x\in \R^{d'}$ and $t\geq 0$,
    \begin{align*}
        P_{T \sim T_{d',d}} ( \norm{Tx} \geq 1 + t ) \leq \exp \left( - \Omega(t^2 d) \right).
    \end{align*}
    
    For our purposes, we will also need a bound on the lengths of the vectors that holds with probability $1-\beta$ after map reducing dimensions to $\log (k/\alpha \beta)/\alpha^2$. We can use the fact that the dimension reduction maps are sub-Gaussian-tailed to get the following bound.
    
    \begin{lemma}
        \label{lem:JLstretch}
        For every point $p$ in a dataset $D'$ of size $n$, given a sub-Gaussian tailed dimension reducing family of maps $T_{d',d}$, we have that with probability $1-\beta$, $\norm{Tp} \leq O(\alpha \sqrt{\log n/\beta}) \norm{p}$.
    \end{lemma}
    
    \begin{proof}
        For any $p \in D'$ we have that
        \begin{align*}
            P_{T \sim T_{d',d}} (\norm{Tp} \geq (1+t) \norm{p}) \leq \exp \left(- \Omega\left(t^2 \frac{\log (k/\alpha\beta)}{\alpha^2} \right) \right).
        \end{align*}
        It follows that there is a choice of $t = O(\sqrt{\log n/\beta \cdot \frac{\alpha^2}{\log (k/\alpha\beta)}}) = O(\alpha \sqrt{\log n/\beta})$ such that the bound above is at most $\beta/n$. Applying the union bound, the desired inequality follows.
    \end{proof}

    \subsection{Fundamental privacy subroutines}
    
    We briefly recall a couple of standard results from the differential privacy literature that are used in the sequel. We rely on the following composition theorem which bounds the loss in privacy of the composition of multiple DP algorithms by appealing to their individual privacy guarantees in a modular fashion.
    
    \begin{theorem}[Basic Composition, \cite{dwork2006our}]
        \label{thm:basic_composition}
        A mechanism with $N$ adaptive interactions with $(\epsilon_i, \delta_i)$-DP mechanisms each for $i \in [N]$ and no other accesses to the database is $(\sum_{i\in[N]}\epsilon_i, \sum_{i \in [N]}\delta_i)$-DP.
    \end{theorem}
    
	We also use the \emph{Gaussian mechanism} and its privacy guarantee as formalized in the following lemma.
	
	\begin{lemma}[Gaussian mechanism, \cite{DBLP:journals/fttcs/DworkR14}]\label{lem:gauss}
	    Given a $d$-dimensional function $f: \mathcal{X} \to \R^d$ which has $\ell_2$-sensitivity $\max_{x,y \in \mathcal{X}} \lVert f(x) - f(y) \rVert_2 < \Delta_{f,2}$, randomized response via the Gaussian mechanism which for an agent with private data $p$ returns $f(p) + Y$ for $Y  \sim N(0,\frac{c_G^2 \Delta_{f,2}^2}{\epsilon^2} \mathbb{I}_{d \times d})$ is $(\epsilon,\delta)$-differentially private for any $c_G^2 > 2 \ln (1.25/\delta)$.
	\end{lemma}

	\subsection{Bitstogram and the Heavy Sums Oracle}
	
	The contents of this subsection are used in the cost analysis for both clustering algorithms. In the sequel we make extensive use of locally private frequency estimation. For private frequency estimation a lower bound of $\Omega_{\epsilon}(\sqrt{n})$ is known \citep{DBLP:conf/esa/ChanSS12}. A state of the art construction for this problem is the $\BTG$ algorithm \cite{DBLP:journals/jmlr/BassilyNST20}, which is an $\epsilon$-LDP algorithm for the heavy-hitters problem that achieves low error.
	
	\begin{restatable}[Algorithm $\BTG$, \cite{DBLP:journals/jmlr/BassilyNST20}]{lemma}{lembtg}\label{lem:btg}
	    Let $V$ be a finite domain of values, let $f:D' \to V$, and let $n(v)$ denote the frequency with which $v$ occurs in $f(D')$. Let $\epsilon \leq 1$. Algorithm $\BTG (f,\epsilon,\beta)$ interacts with the set of $n$ users in 1 round and satisfies $\epsilon$-LDP. Further, it returns a list $L = ((v_i,a_i))_i$ of value-frequency pairs with length $\tilde{O} (\sqrt{n})$ such that with probability $1-\beta$ the following statements hold:
	\begin{enumerate}
	    \item For every $(v,a) \in L$, $\lVert a - f(v) \rVert \leq E$ where $E = O\left( \frac{1}{\epsilon} \sqrt{n \log (n/\beta)} \right)$.
	    \item For every $v\in V$ such that $f(v) \geq M$, $v\in L$, where $M = O\left( \frac{1}{\epsilon} \sqrt{n \log |V|/\beta \log (1/\beta)}  \right)$.
	\end{enumerate}
    We overload notation to treat the list returned by $\BTG$ returns as either a set of (heavy-hitter, frequency) pairs or a function which may be queried on a value to return either the corresponding frequency if it is a heavy hitter or a value of $0$ otherwise. A subscript of $M$ will denote the upper bound on the maximum frequency omitted. We see that whenever $\size{V} = \Omega(n)$, $M = \Omega(E)$ and $\BTG$ promises a uniform error bound of $M$ when estimating the frequency of any element in the co-domain for an appropriate choice of constants.
	\end{restatable}
    
    We introduce an extension of the $\BTG$ algorithm called $\HSO$ that allows us to query the sums of some vector valued function over the set of elements that map to a queried heavy-hitter value. For a given value-mapping function $f:\X \to \mathcal{V}$ and a vector-valued function $g:\X \to \R^d$ the sum estimation oracle privately returns for every heavy hitter $v \in \mathcal{V}$ the sum of all agents that map to $x$, i.e. $\sum_{p : f(p) = x} p$. We recall that $\BTG$ is a modular algorithm with two subroutines; a frequency oracle that privately estimates the frequency of any value in the data universe, and a succinct histogram construction that constructs the heavy hitters in a bit-wise manner by making relatively few calls to the frequency oracle. The construction of $\HSO$ essentially mimics the frequency oracle construction called $\mathsf{Hashtogram}$ from \cite{DBLP:journals/jmlr/BassilyNST20} and can be run in parallel with $\BTG$, allowing us to reduce the round complexity of our protocols. The pseudo-code and proof of \cref{lem:HSO} may be found in \cref{subsec:BTGandHSO}.
    
    \begin{restatable}[$\HSO$]{lemma}{lemHSO}\label{lem:HSO}
        Let $f:\X \to \V$, $g:\X \to B(0,\Delta/2) \subset \R^{d'}$ be some functions where $g$ has bounded sensitivity $\Delta_{g,2}$ and let $D' \subset \X$ be a distributed dataset over $n$ users. With probability at least $1-\beta$, for every $v\in \mathcal{V}$ that occurs in $f(D')$, if $S(v)$ is the value returned by \Cref{alg:HS} then
	    \begin{align*}
	        \left\lVert S(v) - \sum_{f(y)  = v} g(y) \right\rVert &\leq 2 \Delta \sqrt{2n \log \frac{d' +1}{\beta}} + \frac{4 c_G \Delta_{g,2}}{\epsilon} \sqrt{2 d' n\log\frac{4}{\beta}}.
	    \end{align*}
	    Here $c_G$ is the constant derived from the Gaussian mechanism (\cref{lem:gauss}), and $\Delta_{g,2}$ is the $\ell_2$-sensitivity of $g$. Note that since $\Delta_{g,2} \leq \Delta$, this also implies (whenever $\epsilon < c_G = \sqrt{2 \ln (1.25/\delta)}$)
	    \begin{align*}
	        \left\lVert S(v) - \sum_{f(y) = v} g(y) \right\rVert &\leq O\left(\frac{c_G \Delta}{\epsilon} \sqrt{d' n \log\frac{1}{\beta}}\right).
	    \end{align*}
 	    Further, \Cref{alg:HS} is $(\epsilon,\delta)$-LDP.
    \end{restatable}
	
	\subsection{Locality Sensitive Hashing}
	
	The contents of this subsection are used only for the construction and analysis of the multi-round $k$-means algorithm with low additive error. We start by recalling the definition of an LSH family. Complete proofs may be found in the appendix.
   
    \begin{restatable}[Locality sensitive hashing (LSH)]{definition}{defLSH}
    \label{def:LSH}
   	    We say that a family of hash functions $H:\R^d \to B$ for a finite set of buckets $B$ is \emph{locality-sensitive} with parameters $(p,q,r,cr)$ if for every $x,y\in \R^d$ for some $1 \ge p > q \ge 0$, $r>0$ and $c>1$
   	    \begin{align*}
   	        P(H(x) = H(y)) \begin{cases}
   	        \geq p \mbox{ if } d(x,y) \leq r \\
   	        \leq q \mbox{ if } d(x,y) \geq cr.
   	        \end{cases}
   	    \end{align*}
    \end{restatable}

   	In this work we use an LSH family construction from \cite{AI06}.
   	
   	\begin{restatable}{theorem}{thmLSH}\label{thm:LSH}
		For every sufficiently large $d$ and $n$ there exists a family $\mathcal{H}$ of hash functions defined on $\R^d$ such that for a dataset of size $n$,
		\begin{enumerate}
			\item A function from this family can be sampled, stored and computed in time $t^{O(t)} \log n + O(dt)$, where $t$ is a free positive parameter of our choosing.
			\item The collision probability for two points $u,v\in \mathbb{R}^d$ depends only on the $\ell_2$ distance between them, which we henceforth denote by $p(\lVert u - v\rVert)$.
			\item The following inequalities hold:
			\begin{align*}
				p(1) &\geq \frac{A}{2 \sqrt{t}} \frac{1}{(1 + \epsilon + 8 \epsilon^2)^{t/2}}\\
				\forall c>1,\; p(c) &\leq \frac{2}{(1 + c^2 \epsilon)^{t/2}}
			\end{align*}
			where $A$ is an absolute constant $<1$, and $\epsilon = \Theta(t^{-1/2})$. One can choose $\epsilon = \frac{1}{4 \sqrt{t}}$.
			\item The number of buckets $N_B$ an LSH function with parameter $t$ uses is $t^{O(t)} \log n$.
		\end{enumerate}	
   	\end{restatable}
	
	Note that by scaling the input to the LSH function this gives us constructions for $(p,q,r,cr)$-sensitive LSH families for arbitrary values of $r>0$. Due to the occurrence of terms like $t^{O(t)}$ in the collision probabilities and the number of buckets, the performance of an LSH family is very sensitive to the choice of $t$. In the following lemma we show how to choose a value of $t$ for a desired ratio of $p^2(1)$ to $p(c)$.
	
	\begin{restatable}{lemma}{lemLSH}\label{lem:LSH}
	    Given a fixed $c>\sqrt{2}$, for any $B > 1$, there is a choice of $t = O\left( \log^2 B \right)$ for the LSH function described in \cref{thm:LSH} such that
	    \begin{align*}
	    \frac{p^2(1)}{p(c)} &= \Omega(B),    \\ 
	    p(1) &\geq \Omega( B^{-1/c'}/\log B ), \\
	    \log N_B &= O( \log^2 B \log\log B + \log \log n),
	    \end{align*}
	    where $c' = (c^2/8 - 1/4 )$. It will be convenient to note that $1/c' = O(1/(2c^2-1))$.
	\end{restatable}
	
	In the construction of the multi-round $k$-means algorithm with low additive error, we will need to estimate the average of all points that map to a given heavy bucket. Due to the pair-wise nature of the LSH guarantee, the analysis of this requires us to use an arbitrary point from the bucket as a filter to ensure that sufficiently many points close to it and not too many points far from it map to that bucket.
	
	\begin{restatable}{lemma}{lemLSHGuarantee}\label{lem:LSHGuarantee}
    	Let $C \subset D$ be a set of points with diameter $r$ and let the diameter of $D$ be $\Delta$. For any $x_0 \in C$, if $\hat{x_0}$ is the average over all points colliding with $x_0$ under a $(p(1),p(c),r,rc)$-sensitive LSH function $H$ applied to $D$, then with probability $p(1)/4$,
        	\begin{align*}
            \lVert x_0 - \hat{x}_0 \rVert \leq cr + \frac{8 p(c) |D|}{p^2(1) |C|} \Delta,
            \end{align*}
        and the number of points of $C$ that collide with $x_0$ is at least $\frac{p(1) C}{2}$.
	\end{restatable}

%% file: 3AlgorithmOneRound.tex
\section{LDP \texorpdfstring{$k$}{k}-means with arbitrarily tight multiplicative approximation}\label{sec:1Round}

    In this section we describe a one-round $k$-means clustering algorithm and formally analyse its cost and privacy guarantees. We start by describing our algorithm and provide the pseudo-code. We then give an informal description of our methods and a high-level justification for various algorithmic choices. In one line, what we will do is find a small collection of candidate centers for the bi-criteria relaxation to the $k$-means problem, derive cluster centers for a proxy data set derived by weighing the candidate centers by counts of points served, and use these cluster centers to cluster the original data set.
    
    \begin{table}
        \centering
        \begin{tabular}{c|c}
             Notation & Meaning \\ \hline
             $D' \subset \R^{d'}$ & Original data set \\
             $Q : \R^{d'} \to \R^{d}$ & mapping from high-dim. to low-dim. space \\
             $D \subset \R^d$ & $Q(D')$, dimension reduced data set \\
             $G_l$ & Rectangular grid in dimension reduced space \\
             $G_l (\cdot)$ & Mapping from $\R^d$ to coordinate-wise floor in $G_l$\\
             $t_l$ & Unit length of grid $G_l$ \\
             $\PH^l$ & Succinct histogram of number of points mapping to $g \in G_l$ for ``heavy" $g$ \\
             $\Count (\cdot)$ & Count of previously uncovered data points $g \in G_l^*$ serves\\
             $G_l^*$ & Candidate centers picked from grid points in $G_l$, $G_l^* \subset \PH^l$\\
             $N_G$ & A $k^{\tilde{O}(1/\alpha^2)}$ term used to greedily pick $G_l^*$\\
             $\PSO^l$ & Vector sums of points in original space mapping to $g \in G_l$ for ``heavy" $g$\\
             $\Sum (\cdot) $ & Sum of previously uncovered points in original space whose image served by $g \in G_l^*$\\
             $G_{l,i}^*$ & Points of $G_l^*$ for which $s_i^* \in S^*$ is closest center\\
             $M_l^*$ & Maximal grid points picked from $G_1,\dots , G_l$ \\
             $D^* \subset \R^d$ & Proxy data set generated by weighing points in $G_1^*, \dots, G_L^*$ by points served\\
             $S^*$ & $k$-means solution derived by clustering $D^*$ \\
             $S'$ & $k$-means solution for $D'$ output by \cref{alg:1Round}\\
        \end{tabular}
        \caption{Summary of notation used in \cref{alg:1Round}}
        \label{tab:1RoundNotation}
    \end{table}
    
    \begin{algorithm}
        \caption{1-Round $k$-means Clustering}
        \label{alg:1Round}
        \SetKwProg{DoParallelFor}{do in parallel for}{:}{end}
        \SetKwBlock{DoParallel}{do in parallel:}{end}
        \KwData{Data set $D' \subset \R^{d'}$ distributed over $n$ agents, privacy parameters $\epsilon, \delta$, accuracy parameter $\alpha$, failure probability $\beta$}
        \tcc{Step 1: Initialization and interaction}
        $T : \R^{d'} \to \R^d$ dimension reduction for $d = O(\log (k/\alpha\beta)/\alpha^2)$ \label{alg:1Round;line:s1b} \\
        $S : \R^d \to \R^d$ scaling by a factor $\Omega(1/(\alpha \sqrt{\log n/\beta}))$\\
        $P : \R^d \to B(0,1)$ projection to the unit ball\\
        $Q = P \circ S \circ T : \R^{d'} \to B(0,1) \subset \R^d$ \tcc*{Publicly available mapping}
        $L = \lg n$ number of grids in dimension reduced space\\
        $t_l = 2^{l-L+1}/\alpha\sqrt{d}$ for $l=1,\dots ,L$\\
        $G_l = t_l (\Z^d)$ grid with unit length $t_l$\\
        $G_l : \R^d \to G_l$ map flooring points coordinate-wise to the grid $G_l$ \tcc*{Overloaded notation}
        \DoParallelFor{$l\in [L+1]$}{
            $\PH^l \leftarrow \BTG ( G_l\circ Q ,\epsilon,\beta )$ \tcc*{Get frequency oracle for number of points snapping to grid point}
            $\PSO^l \leftarrow \HSO (G_l\circ M , p \mapsto p, \epsilon, \beta)$ \tcc*{Get sum oracle for points mapping to grid point}
        }\label{alg:1Round;line:s1e}
        \tcc{Step 2: Construction of proxy data set}
        $M_0^* \leftarrow \emptyset$ \tcc*{Keeps track of ``maximal" points in grid} \label{alg:1Round;line:s2b}
        \For{$l = 1 ,\dots, L$}{
            $(\Count : G_l \to \mathbb{R}) \leftarrow \PH^l (\cdot)$ \\
            $(\Sum : G_l \to \mathbb{R}^{d'}) \leftarrow \PSO^l (\cdot)$\\
             \For{$g \in M_{l-1}^*$}{
                $\Count (\snap{g}{l}) \leftarrow \Count (\snap{g}{l}) - \PH (g)$\\
                $\Sum (\snap{g}{l}) \leftarrow \Sum(\snap{g}{l}) - \PSO (g)$
             }
            $G^*_l \leftarrow \{ (g,\Count(g)) :g \in \lceil2N_G \log 1/\alpha\rceil \mbox{ points with largest values of }\Count\mbox{ in }G_l \}$\\
            $M_l^* \leftarrow M_{l-1}^*$\\
            \For{$g \in M_l^*$}{
                \If{$\snap{g}{l} \in G_l^*$}{
                    $M_l^*\leftarrow M_l^* \backslash \{g\}$
                }
            }
            $M_l^* \leftarrow M_l^* \cup G_l^*$\\
        }
        $D^* \leftarrow \{ g \mbox{ with multiplicity }\Count(g) \mbox{ for }(g, \Count(g)) \in G^*_1, \dots, G^*_L \}$ \tcc*{Proxy data set} \label{alg:1Round;line:s2e}
        \tcc{Step 3: Cluster center recovery}
        $S^* = \{s_1^*,\dots, s_k^*\} \leftarrow \SKM_\eta (D^*)$ \label{alg:1Round;line:s3b} \\
        $G_{l,i}^* \leftarrow \{g \in G_l^* : \argmin_{s \in S^*} z(g, c) = s_i^* \}$ for each level $l = 1 ,\dots, L$ and cluster center $s_i^* \in S$\\ 
        \For{$j = 1,\dots, k$}{
        $\Sum \leftarrow \sum_{l=1}^L \sum_{g: \in G_l^* (s_j^*) } \Sum (g)$\\
        $\Count \leftarrow \sum_{l=1}^L \sum_{g: \in G_l^* (s_j^*) } \Count (g)$\\
        $\hat{\mu}_j \leftarrow \frac{\Sum}{\Count} $
        }
        \KwRet{$S' = \{\hat{\mu}_1,\dots, \hat{\mu}_k \}$}\label{alg:1Round;line:s3e}
    \end{algorithm}
    
    \subsection{Pseudo-code and algorithm description}
    
    \paragraph{Step 1 - Initialization and interaction:}{ From \cref{alg:1Round;line:s1b} to \cref{alg:1Round;line:s1e}, we first formalize the publicly available dimension reduction, scaling and projection required to ensure that every point lies inside the domain $B(0,1)$; this is the map $Q$. We define $L = \lg n$ grids $G_1, \dots G_L$ where $G_l$ has unit length $t_l = 2^{l - L +1}/\alpha \sqrt{d}$. This definition ensures that the distance from any point in the space to its coordinate-wise floor is at most $\alpha 2^{l-L+1}$ units. The end of Step 1 occurs by $L$ calls to the $\BTG$ and $\HSO$ routines to privately generate succinct histograms $\PH^l$ and sum oracles $\PSO^l$ over points mapping to any given grid-point.}
    
    \paragraph{Step 2 - Construction of proxy dataset}{ From \cref{alg:1Round;line:s2b} to \cref{alg:1Round;line:s2e} we iteratively construct the proxy data set by going from low to high threshold and greedily picking some $2N_G \log 1/\alpha$ grid points $G_l^*$ that maximize the $\Count (\cdot)$ function. The $\Count(\cdot)$ function maintains estimate of the number of previously uncovered data points that would be covered by $g \in G_l$ if picked. We also keep track of the ``maximal" grid points in the sets $M_l^*$; at the beginning of round $l$, $M_{l-1}^*$ consists of all grid points that have been picked so far that have the property that no grid point which would cover them has yet been picked. This will ensure that when we update the $\Count ( \cdot )$ function to account for data points that have already been covered, we do not subtract for any one data points multiple times. Along the way we mimic the $\Count(\cdot)$ construction by generating a similar $\Sum(\cdot)$ mapping that estimates the vector sum of all points in the original space served by $g \in G_l$. This step ends with the construction of the proxy dataset $D^*$ where we repeat each grid point $g \in G_l^*$ with multiplicity equal to the number of data points it served, i.e. $\Count (g)$.}
    
    \paragraph{Step 3 - Cluster center recovery:}{From \cref{alg:1Round;line:s3b} to \cref{alg:1Round;line:s3e} we compute the final cluster centers $S'$ in the original space. We start by first using a non-private $k$-means algorithm $\SKM_\eta$ with an $\eta$-multiplicative approximation guarantee on the privately derived proxy data set $D^*$ to derive cluster centers $S^*$ in the low dimensional space. Then, we iterate over each cluster center $s_i^* \in S^*$ and for every fixed cluster center we use the $\Sum(\cdot)$ functions constructed in step 2 to compute the vector sums over all points snapping to grid points $G_{l,i}^*$ which are closer to $s_i^*$ than to any other center in $S^*$, as well as the count of all such data points (via $\Count$). Our estimate for the true average of this cluster in the original space is then simply $\hat{\mu}_i = \Sum / \Count$, and these $k$ estimates $\{ \hat{\mu}_1, \dots , \hat{\mu}_k\}$ form the final output of our algorithm.}
    
    \subsection{Technical discussion}
    
    We recall from the introduction that for the error we are targeting we need to find $O(\poly k \poly\log n)$ candidate centers with respect to which additive error in a $1 + \alpha$ approximation to $\OPT$ is $O(\poly k \sqrt{n} \poly\log n)$. We recall from that discussion that in the LDP setting one approach to get around the large amount of error added is to discretize the response of the agents. A natural way to achieve this is to the domain is via a rectangular $d'$-dimensional grid of points and ask agents to reveal their closest grid point; the question then becomes how best to exploit this privately derived information for $k$-means clustering. Previous work on $k$-means clustering in the central DP setting \citep{chaturvedi2020differentially} uses such an approach where in order to get an $O(1)$ multiplicative factor approximation to the optimal clustering cost, a sequence of grids is used where the unit length of the $l$th grid equals $2^{-l}/\sqrt{d}$. 
    
    To analyse this approach, one fixes an arbitrary optimal solution to the $k$-means problem $S_{\OPT}$ and partitions the data set based on how far a point lies from the optimal solution via geometrically increasing thresholds $2^{-l}$. Then for any point $p$ which lies at a distance between $2^{-l}$ and $2^{-l+1}$, the closest grid point to $p$ in the grid with unit length $2^{-l}/\sqrt{d}$ is at a distance of at most $2^{-l}$, i.e. closer than the optimal solution. One then reverses the argument to observe that if a point lies within a distance of $2^{-l+1}$ units of $S_{\OPT}$, then by the triangle inequality its closest grid point must lie within a distance of $O(2^{-l+1})$ units of $S_{\OPT}$. The authors then bound the total number of grid points that lie within any collection of $k$ centers to derive the promise that there is a small set of grid points which serve almost all data points which lie at a distance between $2^{-l}$ and $2^{-l+1}$ of the candidate centers.
    
    \paragraph{Choice of grid construction:}{As in this work we are targeting a $(1+\alpha)$ multiplicative approximation, we scale the grid unit lengths by a factor of $\alpha$ to get the promise that if a point lies within a distance of $2^{-l}$ of $S_{\OPT}$, it lies within a distance of $O(\alpha 2^{-l})$ of some grid point. Since we can only identify grid points whose counts are at least $\sqrt{n}$, we can afford to miss at most $O(\poly k \poly\log n)$ many such grid points serving the dataset across all levels for an additive error of $O(\poly k \sqrt{n} \poly\log n )$. It follows that we will need an $O(\poly k \poly\log n)$ bound on the number of grid points that lie close to the optimal centers. We will address this point further ahead in this discussion.
    
    One technicality suppressed so far is that we must have a finite (in fact $O(\poly k \poly\log n)$)  sequence of grids and thresholds for the set of candidate centers accrued across grids to be finite. We observe that if the smallest threshold is $1/n$, then the discretization error for all points which lie within $1/n$ of $S_{\OPT}$ is absorbed by an additive $O(1)$ term instead of the $O(1)$ multiplicative approximation factor; this in conjunction with the fact that the diameter of the domain is $O(1)$ shows that a set of $O(\log n)$-many thresholds suffices.
    
    Returning to the identification of grid points close to $S_{\OPT}$ in the grid, we observe that there is an issue with this approach; the choice of $S_{\OPT}$ was arbitrary and different choices can possibly lead to very different sets of grid points close to $S_{\OPT}$. It is not immediately clear what is a good way to pick grid points when we are oblivious of any choice of $S_{\OPT}$ using only the grid points histogram data.}
    
    \paragraph{Greedy maximum coverage:}{Reasoning along the lines of \cite{jones2020differentially} for the $k$-medians problem shows that a choice of grid points that greedily maximizes how many data points are covered by including these grid points among the candidate centers ensures that the clustering cost of the data set with respect to this set of grid points is at most $O(\OPT)$. Since the number of grid points is larger than $k$, and the cost is a constant factor multiplicative approximation to $\OPT$, this set of grid-points chosen across grids is a solution to the bi-criteria relaxation of the $k$-means solution (modulo some additive error).
    
    A closer look at the argument in \cite{jones2020differentially} shows that the greedy picks must maximize coverage only over yet-uncovered points, when proceeding from low to high thresholds. In the centrally private setting one can dynamically update the coverage of candidate grid points by directly accessing the data set and marking points off as they are covered, but this is not possible in the local setting. We get around this hurdle by two tools; one, ensuring a \emph{consistency} across grids in the sense that if two points map to the same grid point in a low-level grid then they also map to the same grid point in all higher-level grids; and two; keeping track of all grid points picked so far such that they are \emph{maximal} in the sense that no grid point that they themselves snap to in a coarser grid has been picked. We will then be able to evaluate the count of yet uncovered data points covered by any candidate grid point by simply subtracting the histogram counts of all maximal grid points picked so far snapping to that candidate grid points from the histogram count of that candidate grid point.
    
    We will ensure consistency by mapping each point to its \emph{coordinate-wise floor} in the $d$-dimensional grid instead of its closest point; this makes no significant different in the arguments made so far as the floor always lies within a distance of $2^{-l}$ in a grid with unit length $2^{-l}/\sqrt{d}$.
    }
    
    \paragraph{Dimension reduction for bounded candidate centers:}{We now discuss how to get the $O(\poly k \allowbreak \poly\log n)$ bound on the number of grid points within the aforementioned threshold distance of $S_{\OPT}$. For reasons of time and space efficiency, in \cite{chaturvedi2020differentially} the authors needed to bound the number of grid points close to any choice of $S_{\OPT}$ by $O(\poly(n))$. They showed that by dimension reduction to $O(\log n/\alpha^2)$ many dimensions, there are at most $O(n^{1/\alpha^2})$ many grid points within a distance of $r$ of any optimal center for a grid with unit length $\alpha r/\sqrt{d}$. They then appeal to the well-known Johnson-Lindenstrauss lemma that shows that there is a choice of $O(\log n/\alpha^2)$ many dimensions so that the $\ell_2$ distance between all pairs of data points in a data set of size $n$ is preserved within a multiplicative factor of $(1\pm \alpha)$. It is relatively easy to show that the $k$-means clustering cost for any choice of clusters is also preserved within a factor of $(1 \pm \alpha)$. 
    
    A recent work by \cite{DBLP:conf/stoc/MakarychevMR19} generalized the Johnson-Lindenstrauss guarantee for $k$-means clustering by showing that in fact performing dimension reduction to $\log (k/(\alpha\beta))/\alpha^2$-dimensions ensures that with probability $1-\beta$ the cost of every clustering solution is preserved within a multiplicative cost of $(1 \pm \alpha)$. By tracing the argument of \cite{chaturvedi2020differentially} for upper bounding the number of grid points close to any optimal center with this stronger bound on the dimensionality of the dimension-reduced space leads to a $k^{\tilde{O}(1/\alpha^2)}$ bound on the number of grid points close to $S_{\OPT}$. For any fixed approximation factor $(1+\alpha)$, this immediately gives us the desired $O(\poly k \poly\log n)$ bound on the number of grid points close to $S_{\OPT}$ as well as the $O(\poly k \poly\log n)$ bound on the number of candidate centers picked.
    }
    
    \paragraph{Proxy data set construction:}{To recap, the set of candidate centers derived to construct the proxy data set has the property that for all but $O(\poly k \poly\log n)$ many data points, there is a candidate center at a distance of $O(\alpha)$ times the distance between a data point and the optimal centers. We construct the proxy data set by repeating each candidate center with an estimate of the number of points it covers in this manner. This can be seen as essentially \emph{moving} each data point to the candidate center that covers it; in sum what we have shown is that the net movement is $O(\alpha \OPT)$. We can then show by the triangle inequality that the $k$-means clustering functions of the original and the proxy data set are within a $(1+O(\alpha))$ multiplicative approximation factor and $O(\poly k \poly \log n)$ additive error. It will follow that the optimal clustering cost for the proxy data set is a $(1+\alpha)$ factor more than the optimal cost for the original data set (modulo additive error), and therefore that any clustering solution derived by a non-private $k$-means clustering algorithm with multiplicative approximation factor $\eta$ has net clustering cost at most $(1+\alpha)\eta$. Using the relation between the $k$-means clustering functions this time in reverse, we get that the privately derived cluster centers for the proxy data set serve as cluster centers for the original data set with cost $(1+O(\alpha))\eta$.
    }
    
    \paragraph{Undoing the dimension reduction:}{We have privately derived $k$ cluster centers in the dimension reduced space that serve the data set $D$ with clustering cost $(1+O(\alpha))\eta \OPT$ and additive error $O(\poly k \poly\log n)$. This implicitly clusters the original data set with a similar error guarantee by mapping each data point to a cluster corresponding to the center in the low-dimensional solution that its image is closest to. To compute the centers of these clusters in the original space, we use the \emph{sum oracles} derived from calls to $\HSO$ to recover the vector sums of all data points in the original space that lie in these implicitly defined clusters. Dividing these sums by the counts derived during our proxy data set construction gives us good estimates to the cluster-wise centers.
    }
    
    \subsection{Formal cost and privacy analysis}
    
    \paragraph{Proof outline:}{We begin by relating the optimal clustering cost in the original space $\OPT'$, and the clustering cost in the dimension reduced space $\OPT$ (\cref{lem:costApproximation}). We then formally derive some properties of the grids $G_l$, the maximal points identified at the end of round $l$ i.e. $M^*_l$, and the accuracy of the $\Count$ map used in step 2 to choose grid points as candidate cluster centers (\cref{lem:snapBound} to \cref{lem:count_acc}). Since the error bound for the $\Sum$ map is practically identical to that of the $\Count$ map, we prove that in immediate succession.
    
    The core of our cost analysis for the bi-criteria solution is showing that the clustering cost of the data set with respect to many greedy choices of candidate centers is competitive with the optimal clustering (\cref{def:AlOl} and \cref{lem:al_Versus_ol}). These results allow us to show in that the $k$-means clustering functions for the dimension reduced data set $D$ and the proxy data set $D^*$ are close in $\ell_1$ norm (\cref{lem:1round_ProxyVersusActual}).  \Cref{lem:1round_ProxyVersusActual} is then exploited in turn to show that the output of the non-private clustering algorithm works well for the original dimension reduced data set (\cref{cor:proxyToOriginal}).
    
    Finally, starting from \cref{def:actualClustering}, we start the work of recovering cluster center in the original space. We begin in the definition by formalizing the actual clustering of the dimension reduced data set that results from identifying each data point with the first grid point that serves it in some grid. Then we show that the output of the algorithm works well for this actual clustering and leads to a $(1+O(\alpha)) \eta$ factor multiplicative approximation ( \cref{lem:actualDimRedCost} to \cref{lem:originalSpaceCost}). This section culminates in the main result \cref{thm:1RoundGuarantee} which accounts for all privacy loss which occurs across all calls to $\BTG$ and $\HSO$ and after scaling the privacy parameters in the calls to these routines formalizes the final cost guarantee of \cref{alg:1Round}.
    
    }
    
    \begin{definition}
        We recall some notation used in the algorithm description and introduce some definitions that help with the cost analysis for this algorithm.
        \begin{itemize}
            \item There is a sensitive dataset $D' \subset B(0,1)$ distributed among $n$ users, exactly one point per user. We denote the cost of the optimal $k$-means solution by $\OPT'$.
            \item Let $Q:\R^{d'} \to \R^d$ be a publicly available function that maps the data domain to $B(0,1)$ in the dimension reduced space $\R^d$. It is computed by first computing the output of the dimension reduction map $T$, followed by a scaling $S$ by $1/\alpha\sqrt{\log n}$ units (which ensures that with high probability all points lie inside the unit ball in the dimension reduced space, followed by a projection $P$ to the unit ball to deal with any outliers.
            \item We denote the dimension-reduced data set $Q(D')$ by $D$. We denote its optimal clustering cost by $\OPT$. We fix any optimal $k$-means solution $S_{\OPT}$ for $D$ with clustering cost $\OPT$.
            \item Let $L = \lceil \lg n\rceil$ denote the number of levels.
            \item Let $r_l = 2^l/2^{L-1}$ for $l = 1,\dots, L$ denote the $\ell_2^2$ distances which we use as thresholds to partition $D$ depending on how far points lie from $S_{\OPT}$. Note that $r_1 < 1/n$ and $r_L = 2$. Further, we set $r_0 = 0$. With this notation we see that $D \subset B(0,1) \subset B(p, r_L)$ for any $p \in B(0,1)$. 
            \item Let $o_l := \{ p\in D: z(p,S_{\OPT}) \in [r_{l}, r_{l+1}) \}$ for $l = 1 , \dots, L$ denote the thresholded partitions of $D$. Note that with our choice of $r_l$ this definition implies $D = \sqcup_{i=l}^L o_l$.
            \item Let $t_l = \alpha r_l/\sqrt{d}$ denote the unit length of the grid $G_l$  for $l = 1,\dots, L$. Let $G_l$ be the axis aligned grid of unit length $t_l$ units centered at the origin in $B(0,1)$, i.e. $G_l := B(0,1) \cap (t_l \Z^d)$. We overload notation and let $\snap{\cdot}{l} : \R^d \to G_l$ map each point to its floor in $G_{l}$, i.e. $p$ maps to $t_l (\lfloor p_1/t_l\rfloor, \dots, \lfloor p_d/t_l \rfloor)$. Note that these multidimensional floor maps are consistent in the sense that for any $m>i$, for the $j$th coordinate we have
            \begin{align*}
                G_{m} \circ G_l (p)_j &= t_{m} \lfloor  t_l \lfloor p_j/t_l \rfloor / t_m \rfloor \\ 
                &= t_{m} \lfloor 2^{l-m} \lfloor p_j/t_l \rfloor \rfloor \\
                &= t_{m} \lfloor 2^{l-m} p_j/t_l \rfloor \\
                &= t_{m} \lfloor p_j/t_{m} \rfloor \\
                &= G_{m} (p)_j
            \end{align*}
            so putting all coordinates together $G_{m} \circ G_i (p)_j = G_{m}(p)$. Note that $t_{m} \lfloor 2^{l-m} \lfloor p_j/t_l \rfloor \rfloor = t_{m} \lfloor 2^{l-m} p_j/t_l \rfloor$ because $1/2^{l-m} \in \Z$.
            \item We assume each grid point is implicitly tagged with the index of its parent grid point. We will abuse notation and drop indices for the succinct histograms $\PH^l$ and $\PSO^l$ where they may be deduced from the grid point for which the frequency or sum is being queried.
        \end{itemize}
    \end{definition}
    
    \begin{lemma}[Accounting for dimension reduction]
        \label{lem:costApproximation}
        With probability $1-\beta$, we have that for every clustering $(D'_1, \dots, D'_k)$ of $D'$,
        \begin{align*}
            \sum_{i \in k} \sum_{p \in D'_i} s\left( p , \frac{\sum_{q\in D'_i} q}{\size{D'_i}} \right) \simeq_{1 + \alpha} (\alpha \log n/\beta ) \sum_{i \in k} \sum_{p \in Q (D'_i) } s\left( Q (p) , \frac{\sum_{q\in D'_i} M(q)}{\size{D'_i}} \right).
        \end{align*}
        As a direct corollary $\OPT' \simeq_{1 + \alpha} (\alpha \log n/\beta) \OPT$.
    \end{lemma}
    
    \begin{proof}
        We write $Q = P \circ S\circ T$, where $T$ is the dimension reduction to $O(\log (k/\alpha\beta)/\alpha^2)$, $S$ is the scaling by a factor of $\Omega(1/\alpha\sqrt{n/\beta})$, and $P$ is projection to the unit ball. Given any clustering $(D'_1, \dots, D'_k)$ of $D'$, by \cref{lem:Makarychev} we have that
        \begin{align*}
            \sum_{i \in k} \sum_{p \in D'_i} s\left( p , \frac{\sum_{q\in D'_i} q}{\size{D'_i}} \right) \simeq_{1 + \alpha} \sum_{i \in k} \sum_{p \in D'_i} s\left( T(p) , \frac{\sum_{q\in D'_i} T(q)}{\size{D_i}} \right).
        \end{align*}
        The scaling map changes all $\ell_2$-distances by precisely the scaling factor, so we also have that
        \begin{align*}
            \sum_{i \in k} \sum_{p \in D'_i} s\left( T(p) , \frac{\sum_{q\in D'_i} T(q)}{\size{D'_i}} \right) = (\alpha \log n/\beta ) \sum_{i \in k} \sum_{p \in S \circ T D'_i} s\left( S \circ T (p) , \frac{\sum_{q\in D'_i} S \circ T(q)}{\size{D'_i}} \right).
        \end{align*}
        Finally, since with probability $1-\beta$ all points lie in the unit ball after scaling by a factor of $1/\alpha\sqrt{\log n/\beta}$, the projection map does not move any point and hence the same clustering cost holds for $P \circ S \circ T (D') = Q(D')$.
    \end{proof}
    
    In the following lemma we derive a bound on the discretization error and use that to derive a promise that in every level $l$ if we snap $o_l$ to the grid then we get at most $k^{\tilde{O}(1/\alpha^2)}$ many grid points.
    
    \begin{lemma}[Properties of grids $G_l$]
        \label{lem:snapBound}
        \label{lem:NGBound}
        For all $l = 1,\dots, L$, the following bound statements hold for each grid $G_l$:
        \begin{enumerate}
            \item For any $p\in B(0,1)$, $\norm{p - \snap{p}{l}} \leq \alpha 2^{-l} = t_l \sqrt{d} = \alpha r_l$.
            \item $\size{G_l (\cup_{j=1}^l o_j)} = k^{O(1/\alpha^2)}$.    
        \end{enumerate}
    \end{lemma}
    
    \begin{proof}
        \begin{enumerate}
            \item By definition $G_l(p) = t_l (\lfloor p_1/t_l \rfloor, \dots, \lfloor p_d/t_l \rfloor)$. Since $\size{p_j/t_l - \lfloor p_j/t_l \rfloor } \leq 1$, it follows that $\size{p_j - \snap{p}{l}_j} \leq t_l \Rightarrow \norm{p - \snap{p}{l}} \leq t_l\sqrt{d} = \alpha r_l$.
            
            \item Let $p\in (\cup_{j=1}^l o_j)$. By definition of $z(\cdot,\cdot)$, $z(G_l (p), S_{\OPT}) \leq z(G_l(p), \argmin_{c \in S_{\OPT}} z(p,c))$. Then, since $z(p,G_l (p) ) \leq \alpha^2 r_l^2 = O(\alpha^2 r_l)$ and $r_l = O(z(p,\argmin_{c \in S_{\OPT}} z(p,c)))$, by the weak triangle inequality $z(G_l (p) , S_{\OPT}) \leq (1 + O(\alpha))r_l$.
        
            Since we have shown $G_l (\cup_{j=1}^l o_j) \subset \{ g \in G_l: z(g, S_{\OPT}) < (1 + O(\alpha))r_l \}$, it will suffice to bound the size of the latter set. Fix any $s \in S_{\OPT}$. If $g \in G_l$ is such that $z(g,s) \leq (1+ O(\alpha)) r_l$ then by the weak triangle inequality $z(\snap{s}{l}, g) \leq (1+ O(\alpha))r_l$. By translating $G_l(s)$ and $G_l$ so that $G_l(s)$ lies at the origin and scaling the space up by a factor of $1/t_l$ so that $G_l$ maps onto $\Z^d$, we see that there is an injection from $\{g \in G_l : z(g,s) \leq (1 + O(\alpha))r_l) \}$ into $V = \{j \in \Z^d : z(j,0) \leq d/\alpha^2 + O(1)\}$.
            
            Expanding the definition of $V$, we get
            \begin{align*}
                v \in V \Rightarrow \sum_{i\in [d]} v_i^2 \leq d/\alpha^2 + O(1).
            \end{align*}
            It follows that the number of unsigned $v \in V$ is at most the number of ways of partitioning $d/\alpha^2 + O(1)$ balls into $d+1$ distinguishable bins. Then,
            \begin{align*}
                \size{V} &= \binom{d/\alpha^2 + O(1)}{d + 1} \\
                &< \left( \frac{e \cdot (d/\alpha^2 + O(1))}{d+1} \right)^{d +1} \\
                &= k^{\tilde{O}(1/\alpha^2)}
            \end{align*}
            where we use that $d = O\left( \log (k/(\alpha\beta))/\alpha^2\right)$. Hence, $\size{\{ g \in G_l: z(g, S_{\OPT}) < (1 + O(\alpha))r_l \}} = k\cdot 2^d \cdot k^{\tilde{O}(1/\alpha^2)} =  k^{\tilde{O}(1/\alpha^2)}$.
        \end{enumerate}    
    \end{proof}
    
    \begin{definition}
        We make a couple of definitions to ease our analysis from this point.
        \begin{enumerate}
            \item Let $N_G$ be a uniform upper bound on the the sizes of the sets $G_l (\cup_{j=1}^l o_j)$. It follows from \cref{lem:snapBound} that we can choose a value of $N_G = k^{\tilde{O}(1/\alpha^2)}$.
            \item We define a sequence of subsets $a_l$ inductively. Let $a_1 = \{p \in D : \snap{p}{1} \in G_1^* \}$ and let $a_l = \{ p \in D : \snap{p}{l} \in G_l^* \} \backslash \left( \cup_{j=1}^{l-1} a_j \right)$. Informally, $a_l$ consists of those points which were not explicitly covered at a distance of $\alpha r_j$ for any $j<l$ but are successfully covered by some $g\in G_l^*$ at an $\ell_2$ distance of $\alpha r_l$, since its floor in the grid was added to $G_l^*$.
            \item $M_l^*$ is the set of grid points constructed iteratively by adding all grid points picked in round $l$ from $G_l$ to grid points picked in previous rounds and then removing all grid points picked in previous rounds which snap to any grid point picked in round $l$. Intuitively, we can think of this set as the set of ``maximal" grid points that have been picked so far. Keeping track of this set will allow us to avoid over-counting data points being covered at different levels and keep private estimation error terms small.
        \end{enumerate}
        
    \end{definition}
    
    \begin{lemma}[Properties of maximal grid point sets $M_l^*$]
    \label{lem:M*Props}
        The following properties hold for the sets $M_l^*$ for $l = 1,\dots, L$.
        \begin{enumerate}
            \item If $p\in a_j$ for some $j\leq l$ then $\exists ! k \in \{j,\dots, l\}$ such that $\snap{p}{k} \in M_l^*$.
            \item $\size{M_l^*} = l N_G$
        \end{enumerate}
    \end{lemma}
    
    \begin{proof}
    \begin{enumerate}
        \item Given that $p\in a_j$, by construction of $M_j^*$, $\snap{p}{j} \in M_j^*$. If for some $j' > j$ there is some $g \in G_{j'}^*$ such that $\snap{\snap{p}{j}}{j'} = g$ and $\snap{p}{j}$ is removed from $M_{j'}^*$ then since $\snap{\snap{p}{j}}{j'} = \snap{p}{j'}$ and $\snap{p}{j'} = g$ is included in $M_{j'}^*$ proceeding inductively it follows that $\snap{p}{k} \in M$ for some $k \in \{j, \dots, l \}$.
        
        To see that this value of $k$ is unique suppose to the contrary that $\snap{p}{k_1}$ and $\snap{p}{k_2}$ both lie in $M_l^*$. Assume without loss of generality that $k_1 < k_2$. Then since $\snap{\snap{p}{k_1}}{k_2} = \snap{p}{k_2}$ we see that $\snap{p}{k_1} \not\in M_{k_2}^*$ and therefore $\snap{p}{k_1} \not\in M_l^*$.
        
        \item We see that by construction in every loop $\size{M_l^*}\leq \size{M_{l-1}^*} + N_G$. The bound follows directly. 
    \end{enumerate}
    \end{proof}
    
    In order to proceed with the cost analysis, we derive bounds on the estimation error for the point histograms $\PH^l$ and point sum oracles $\PSO^l$. We will avoid substituting for these error terms until we have reached the end of this analysis but it will be useful to keep in mind that, as the lemma shows, they are roughly $O(\frac{1}{\epsilon} \sqrt{n} \log n)$. These bounds are essentially corollaries of the $\BTG$ and $\HSO$ error bounds.
    
    \begin{lemma}[Private estimation error bounds]
        \label{lem:histBounds}
        With probability $1-\beta$, for all $l = 1,\dots, L$, suppressing terms logarithmic in $1/\alpha$, $1/\beta$ and $\log n$, the following guarantees hold.
        \begin{enumerate}
            \item For every $g \in G_l$, $\size{\PH^i(g) - G_i^{-1} (g)} \leq \PH_E :=  \tilde{O} \left( \frac{1}{\epsilon \alpha} \sqrt{n \log^3 n} \right)$.
            \item For every $g \in \PH^l$, $\PSO_E \leq  \tilde{O}\left(\frac{c_G}{\epsilon} \sqrt{d' n}\right)$.
        \end{enumerate}
    \end{lemma}
    
    \begin{proof}
        \begin{enumerate}
            \item We simplify the $\BTG$ guarantee and use $\PH_M$ as a uniform upper bound for both $\PH_E$ and $\PH_M$. In other words, since $\PH_M$ is larger than $\PH_E$, which we show below using the bound on $\log(|G_l|)$, every heavy hitter in $\PH$ is already estimated within an error of $\PH_M$. If a value does not occur in $\PH$, it must be the case that its frequency is less than $\PH_M$, so we estimate the frequency of any omitted element by $0$ and use the upper bound for $\PH_M$ as a uniform bound for the frequency estimates of $g\in G_l$. Similarly, we bound $\PSO_E$ by $\PSO_M$.
            
            To derive the expression for the bound we need to bound from above the sizes of the grids $G_l$. The domain $B(0,1)$ is contained inside the unit cube with side-length $2$ units centered at the origin. The length of each axis that lies within this unit cube is $2$. For every $g \in G_l$, $g_j$ for every coordinate $j$ can take at most $2/t_l  = 2^{L-l}\sqrt{d}/\alpha$ many values. Since the number of dimension is $\tilde{O}((\log k)/\alpha^2)$, it follows that $\size{G_l} = (2^{L-l}\sqrt{d}/\alpha)^{\tilde{O}((\log k)/\alpha^2)} \Rightarrow \log (\size{G_l} \cdot 2L/\beta) < \tilde{O}(((\log k)/\alpha^2)\log (n/\alpha)) + \log (2L/\beta)$. Substituting, for any $l = 1,\dots, L$, $\size{\PH^l (g) - G_l^{-1} (g)} \leq \tilde{O}\left( \frac{1}{\epsilon} \sqrt{n ((\log k)/\alpha^2)\log (n) \log (2L/\beta)} \right) = \tilde{O} \left( \frac{1}{\epsilon \alpha} \sqrt{n \log^3 n} \right)$ with probability $1-\beta/2L$.
            \item We recall that the diameter of the data domain is $O(1)$. Scaling the failure probability by $1/2L$ so that we may apply the union bound and absorbing the resulting $\sqrt{\log\log n/\beta}$ term in the $\tilde{O}$ notation, the $\HSO$ guarantee gives us that $\norm{\PSO_E} = \tilde{O}\left(\frac{c_G}{\epsilon} \sqrt{d' n }\right)$.
        \end{enumerate} 
    \end{proof}
    
    The significance of the following lemma is that $\Count (g)$ is a good estimate of the number of previously uncovered data points covered by a grid point $g \in G_l$ for any $l = 1,\dots, L$ within a distance of $\alpha r_l$.
    
    \begin{lemma}
        \label{lem:count_acc}
        For $l = 1,\dots, L$ and any $g \in \PH^l$,
        \begin{align*}
            \size{ \Count (g) - \size{\{p \in D: \snap{p}{i} = g \} \backslash \left(\cup_{j \in [i]} a_i\right)}  } \leq (l \cdot N_G ) \PH_E
        \end{align*}
    \end{lemma}
    
    \begin{proof}
        By construction of $\Count$, we can write
        \begin{align}
            \Count (g) &= \PH^{l} (g) - \sum_{\substack{g' \in M_{i-1}^* \\ \snap{g}{l} = g}} \PH (g') \label{eq:count1}.
        \end{align}
        Similarly, by definition we can write 
        \begin{align*}
            a_l = \{ p \in D : \snap{p}{l} = g \} \backslash \left( \cup_{j=1}^{l-1} a_j \right).
        \end{align*}
        By \cref{lem:M*Props} we see that the sets $\{ p \in D: \snap{p}{j} = g', j<l \}$ for $g' \in M_{l-1}^*$ form a partition of $\cup_{j=1}^{l-1} a_j$. Similarly, the sets $\{ p \in D: \snap{p}{j} = g' \mbox{ for some }j\}$ for $g' \in M_l^*$ such that $\snap{g'}{l} = g$ form a partition of $\{p \in D: \snap{p}{l} = g \} \cap \left(\cup_{j = 1}^{l-1} a_i\right)$. This implies
        \begin{align}
            \size{\{p \in D: \snap{p}{l} = g \} \cap \left(\cup_{j=1}^{l-1} a_j \right)} &= \sum_{\substack{g' \in M_{l-1}^*,\\ \snap{g'}{l} = g}} \size{\{ p \in D: \snap{p}{j} = g'\mbox{ for some }j \}} \nonumber \\
            \Rightarrow \size{\{p \in D: \snap{p}{l} = g \} \backslash \left(\cup_{j=1}^{l-1} a_i\right)}& \nonumber \\
            = \size{\{p \in D : \snap{p}{l} = g \}} &- \sum_{\substack{g' \in M_{l-1}^*,\\ \snap{g'}{l} = g}} \size{\{ p \in D: \snap{p}{j} = g'\mbox{ for some }j \}}.  \label{eq:count2}
        \end{align}
        By the $\BTG$ guarantee we have that for every $g'\in \PH^j$, $\size{\PH^j (g') - \size{\{ p \in D: \snap{p}{j} = g' \}}} \leq \PH_M$. Subtracting \cref{eq:count2} from \cref{eq:count1} and using the error bound derived from the $\BTG$ guarantee (\cref{lem:histBounds}) we get
        \begin{align*}
            \Count (g) - \size{\{p \in D: \snap{p}{l} = g \} \backslash \left(\cup_{j=1}^{l-1} a_j \right)} &\leq \PH_M + \sum_{g' \in M_{l-1}^*, \snap{g}{l} = g} \PH_M \\
            &\leq O\left(\size{M_{l-1}^*} \PH_M \right) \\
            &\leq l N_G \PH_M.
        \end{align*}
    \end{proof}

    The following lemma is used only for the cluster center recovery in the original space but we state and prove it here due to its similarity to \cref{lem:count_acc}.
    
    \begin{lemma}
        \label{lem:sum_acc}
        For any $g \in G_l^*$, 
        \begin{align*}
            \bigg\|\Sum (g) - \sum_{\substack{G_l(Q(p)) \in G_l^{-1} (g) \\ Q(p)\not\in (\cup_{j=1}^{l-1} a_j )}  } p \bigg\|    \leq l N_G \PSO_M.
        \end{align*}
    \end{lemma}
    
    \begin{proof}
        The proof is essentially identical to that of \cref{lem:count_acc}, but we reproduce the calculations for completeness. In the level $l$ we can write by construction that
        \begin{align*}
            \Sum (g) &= \PSO^{i} (g) - \sum_{g' \in M_{l-1}^*, \snap{g'}{i} = g} \PSO (g').
        \end{align*}
        By \cref{lem:M*Props} we see that $\{p: Q(p) \in D: \snap{Q(p)}{j} = g', j<l \}$ for ${g' \in M_{l-1}^*}$ is a partition of $\cup_{j=1}^{l-1} a_j$. We can write these sets more succinctly as $\{p: Q(p) \in G^{-1} (g') \}$ where we can drop the index of the $G^{-1}$ as it depends upon and can be inferred by the argument $g'$. Continuing, we also have that the sets $\{p: T(p) \in G^{-1} (g') \}$ for ${g' \in M_{i-1}^*, \snap{g'}{l} = g}$ form a partition of $\{p : T(p) \in G_l^{-1} (g) \} \cap \left(\cup_{j=1}^{l-1} a_j \right)$. This implies
        \begin{align*}
            \sum_{\substack{Q(p) \in G_l^{-1} (g) \\ \cap \left(\cup_{j=1}^{l-1} a_j \right)}} p &= \sum_{\substack{g' \in M_{l-1}^*\\ \snap{g'}{l} = g}} \sum_{Q(p) \in G^{-1} ( g')} p \\
            \Rightarrow \sum_{\substack{Q(p) \in G^{-1} (g) \\ \backslash \left(\cup_{j=1}^{l-1} a_i\right)}} p &= \sum_{Q(p) \in G_l^{-1}( g) } p -  \sum_{\substack{g' \in M_{l-1}^* \\ \snap{g'}{l} = g}} \sum_{\substack{Q(p)\in G^{-1} (g') }} p 
        \end{align*}
        By the $\HSO$ guarantee we have that $\size{\PSO (g') - \sum_{Q(p) \in G^{-1} (g')} p } \leq \PSO_E \leq \PSO_M $. Subtracting the second equation from the first and using the error bound derived from the $\HSO$ guarantee (\cref{lem:histBounds}) we get
        \begin{align*}
            \bigg\|\Sum (g) - \sum_{\substack{Q(p) \in G^{-1} (g) \\ \backslash \left(\cup_{j=1}^{l-1} a_j \right)}} p\bigg\|  &\leq \PSO_M + \sum_{\substack{g' \in M_{l-1}^* \\ \snap{g'}{l} = g}} \PSO_M \\
            &\leq l N_G \PSO_M.
        \end{align*}
    \end{proof}
    
    \begin{definition}\label{def:AlOl}
        We let $O_l = \sum_{j=l}^L \size{o_j}$ and $A_l = \sum_{j=l}^L \size{a_j}$. Note that with these definitions,
            \begin{align}
                \sum_{l=1}^L A_l (r_l - r_{l-1}) &= \sum_{l=1}^{L-1} (A_{l} - A_{l+1})r_l - r_0 A_1 \nonumber \\
                &= \sum_{l=1}^{L}\size{a_l}r_{l}. \label{eq:aInTermsOfA}
            \end{align}
        This relation will be useful to us in the cost analysis. Further, we observe that $A_l - O_{l+1} = (n - O_{l+1}) - (n - A_{l}) = \sum_{j=1}^l \size{o_j} - \sum_{j=1}^{l-1} \size{a_j}$. Since the $o_j$ are disjoint, it follows that $A_l - O_{l-1}$ is a lower bound for $\size{\cup_{j=1}^l o_j \backslash \cup_{j=1}^{l-1} a_j}$, the size of the set of points covered by $S_{\OPT}$ within a distance of $r_l$ but still uncovered at the beginning of the $l$th round, i.e. before $a_l$ is picked.
    \end{definition}
    
    \begin{lemma}
    \label{lem:al_Versus_ol}
        For $l = 1,\dots, L$ and for $\erralpha = 4 l N_G^2 \PH_M$,
        \begin{align*}
            \size{a_l} \geq (1 - \alpha) (A_{l} - O_{l + 1})  - \erralpha.
        \end{align*}
    \end{lemma}
    
    \newcommand{\C}[1]{\mbox{Cover}(#1)}
    
    \begin{proof}
        From \cref{lem:NGBound} we know that $\size{G_l(\cup_{j=1}^l o_j))} \leq N_G$. For all $g\in G_l$, let $\C{g} = \{ p \in D : G_l (p) = g \} \backslash \cup_{j=1}^{l-1} a_j$, i.e. the set of yet uncovered data points that would be served by $g$ if $g$ were picked. We note that the sets $\C{g}$ as defined are disjoint for distinct $g \in G_l$. Let $G^\dagger_l = (g^\dagger_1,\dots,g^\dagger_{N_G})$ be the $N_G$ many grid points $g$ with the greatest values of $\size{\C{g}}$ sorted in decreasing order. Then by the observations in \cref{def:AlOl} it follows that
        \begin{align*}
            \size{ \{p \in D : \snap{p}{l} \in G^\dagger_l \} \backslash \cup_{j=1}^{l-1} a_j } &\geq \size{ o_l \backslash \cup_{j=1}^{l-1} a_j } \\
            \Rightarrow \sum_{j \in N_G} \size{\C{g_j^\dagger}} &\geq A_l - O_{l+1}.
        \end{align*}
        In \cref{alg:1Round}, we pick grid points greedily via the privatized counts $\Count (g)$. By \cref{lem:count_acc} we know that for all $g \in \PH^l$,
        \begin{align*}
            \size{\Count (g) - \size{\C{g}}} \leq l N_G \PH_M.
        \end{align*}
        It follows that if $g_j^*$ is our $j$th greedy pick maximizing $\Count(g)$ and the maximum value of $\size{\C{g}}$ over all unpicked grid points is $4l N_G \PH_M$, then $\size{\C{g_j^*} } \geq 4l N_G \PH_M \geq (1/2) \max_{g\in G_i} \size{\C{g}}$. 
        
        Let $m$ be the largest index such that $\size{\C{g_m^\dagger}} \geq 4l N_G \PH_M$. In the context of \cref{lem:greedycover}, we let $Z = \{ \C{g_j^\dagger} : j\leq m \}$, a family of sets guaranteed to cover $U = \cup_{z \in Z} z$, and $\mathcal{S}$ the family of all possible sets we can pick from $\{ \C{g} : g \in G_i \}$. Then, since in the $j$th round each greedy pick $g_j^*$ covers at least half of the maximum that any pick could cover, we see that $(2 \size{N_G} \log(1/\alpha) + 1)$ greedy picks cover $(1-\alpha)\size{\cup_{j = 1}^m \C{g_j^*}} \geq (1-\alpha) \left( \left[\sum_{j \in N_G} \size{\C{g_j^\dagger}} \right] - 4 l N_G^2 \PH_M \right) \geq (1-\alpha) \size{o_l} - 4 l N_G^2 \PH_M$ points, which is what we wanted to show. We use $\erralpha = 4 l N_G^2 \PH_M$ as shorthand going forward.
    \end{proof}
    
    \begin{lemma}
        We can relate the optimal clustering cost $\OPT$ to the sizes of the sets $a_l$ and $o_l$ via the following bounds.
        \begin{enumerate}
            \item $\sum_{l = 0}^{L} \size{o_l}r_{l} = O(\OPT) + O(1)$.
            \item $\sum_{l=1}^{L}\size{a_l}r_{l} \le (1 + O(\alpha) )\sum_{l=1}^{L}\size{o_l}r_{l} + O(\erralpha )$.
            \item $\sum_{l=1}^{L}\size{a_l}r_{l} \le O(\OPT)
            + O(\erralpha )$.
        \end{enumerate}
    \end{lemma}
    
    \begin{proof}
        \begin{enumerate}
            \item Since $r_{l+1} = 2r_{l}$ for $l > 0$, and $r_1 \le 1/n$,
                \begin{align*}
                    \sum_{l=1}^{L}\size{o_l}r_{l} &= \sum_{l = 1}^{L}\size{o_l}4r_{l-1} + |o_0|r_1 \\
                    &\le 2\sum_{l=1}^{L} \size{o_l}r_{l-1} + 1 \\
                    &\le 2\sum_{l=1}^{L}\sum_{p \in o_l}z(p,S_{\OPT}) + 1 \\
                    &\leq 2 f_D(S_{\OPT}) + 1.
                \end{align*}
            \item By \cref{lem:al_Versus_ol},
                \begin{align*}
                    \size{a_l} &\ge (1-\alpha)(A_l - O_{l+1}) - \erralpha\\
                    &\ge (1-\alpha)(\size{a_l} + A_{l+1} - O_{l + 1}) - \erralpha\\
                    &\ge \left(\frac{1-\alpha}{\alpha}\right)(A_{l+1} - O_{l + 1}) - \frac{\erralpha}{\alpha}\\
                    \Rightarrow A_{l+1} &\le \frac{\alpha\size{a_l}}{1-\alpha} + O_{l + 1} + \frac{\erralpha}{1-\alpha} \\
                    &= O(\alpha) \size{a_l} + O_{l+1} + O(\erralpha).
                \end{align*}
                Continuing from \cref{eq:aInTermsOfA} and using the convention that $a_0$ begin undefined is empty,
                \begin{align*}
                    \sum_{l=1}^{L}\size{a_l}r_{l} &= \sum_{l=1}^{L}A_l(r_{l} - r_{l-1})\\
                    &\leq \sum_{l=1}^L \left( O(\alpha) a_{l-1} + O_l + O(\erralpha)\right) (r_l - r_{l-1}) \\
                    &\le \sum_{l=1}^{L} O(\alpha)\size{a_{l-1}}(r_{l} - r_{l-1}) + \sum_{l=1}^{L}O_{l}(r_{l} - r_{l-1}) + O(\erralpha)(r_{L + 1} - r_0)\\
                    &\le O(\alpha)\sum_{l=1}^{L}\size{a_{l-1}}r_{l-1} + \sum_{l=1}^{L}\size{o_l}r_{l} + O(\erralpha )\\
                    \Rightarrow \left(1 - O(\alpha )\right)\sum_{l=1}^{L}\size{a_l}r_{l} &\le \sum_{l=1}^{L}\size{o_l}r_{l} + O(\erralpha )\\
                    \Rightarrow \sum_{l=1}^{L}\size{a_l}r_{l} &\le (1 + O(\alpha) )\sum_{l=1}^{L}\size{o_l}r_{l} + O(\erralpha ).
                \end{align*}
            \item This is a direct consequence of the first and second results of this lemma.
        \end{enumerate}
    \end{proof}
    
    \begin{lemma}
        \label{lem:1round_ProxyVersusActual}
        The $k$-means clustering functions for the dimension reduced dataset $D$ and the proxy dataset $D^*$ are close in $\ell_1$ norm. Concretely, for any finite set $C$,
        \begin{align*}
            f_{D^*} (C) &\le (1 + O(\alpha)) f_D (C) + O(\alpha \OPT) + O(\alpha \erralpha) + O(L^2 N_G^2 \PH_M),\\
            f_{D}(C) &\le (1 + O(\alpha)) f_{D^*} (C) + O(\alpha \OPT) + O(\alpha \erralpha) + O(L^2 N_G^2 \PH_M).
        \end{align*}
    \end{lemma}
    
    \begin{proof}
        We can write $D^* = \sqcup_{l = 1}^L \sqcup_{g \in G_l^*} \{g \mbox{ with multiplicity } \Count(g)\}$. Then it follows that
        \begin{align*}
            f_{D^*} (C) &= \sum_{l=1}^L \sum_{g \in G_l^*} \Count(g) f_{g} (C) \\
            &= \sum_{l=1}^L \sum_{g \in G_l^*} \Count(g) z(g,C) \\
            &\le \sum_{l=1}^L \sum_{g \in G_l^*} \left( \size{a_l \cap G_l^{-1} (g)} + l N_G \PH_M\right) z(g,C) \\
            &\le \left(\sum_{l=1}^L \sum_{p \in a_l} z(\snap{p}{l},C) \right) + O(L^2 N_G^2 \PH_M) \\
            &\leq \left(\sum_{l=1}^L \sum_{p \in a_l} z(\snap{p}{l},\argmin_{s \in C} z(p,s)) \right) + O(L^2 N_G^2 \PH_M)
        \end{align*}
        To bound $z(\snap{p}{l},\argmin_{s \in C} z(p,s))$, we use the AM-GM inequality in conjunction with the triangle inequality for the $\ell_2$ norm as follows:
        \begin{align*}
                z(\snap{p}{l},\argmin_{s \in C} z(p,s)) &\leq (\sqrt{z(\snap{p}{l}, p)} + \sqrt{z(p,C)} )^2 \\
                &\leq \alpha^2 r_l^2 + 2 \alpha r_l \sqrt{z(p,C)} + z(p,C)\\
                &\leq \alpha^2 r_l^2 + 2 r_l (\sqrt{\alpha} \cdot \sqrt{\alpha z(p,C)}) + z(p,C)\\
                &\leq \alpha^2 r_l^2 + \alpha r_l + \alpha r_l z(p,C) + z(p,C) \\
                &\leq O(\alpha r_l) + (1 + O(\alpha)) z(p,C)
        \end{align*}
        Applying this bound for every point $p\in a_l$ for all $l = 1,\dots, L$ we get
        \begin{align*}
            f_{D'} (C) &= (1 + O(\alpha)) f_D (C) + O(\alpha) \sum_{l=1}^L \size{a_l} r_l  +  O(L^2 N_G^2 \PH_M) \\
            &= (1 + O(\alpha)) f_D (C) + O(\alpha \OPT) + O(\alpha \erralpha) + O(L^2 N_G^2 \PH_M).
        \end{align*}
        Similarly,
        \begin{align*}
            f_D(C) &= \sum_{p \in D} z(p,C)\\
            &= \sum_{l=1}^L \sum_{p \in a_l} z(p,\argmin_{s\in C} z(\snap{p}{l}, s))\\
            &= \sum_{l =1}^L \sum_{p \in a_l} (1 + O(\alpha)) z(p,C) + O(\alpha)r_l\\
            &\leq \left[\sum_{l = 1}^L \sum_{p \in a_l} (1 + O(\alpha)) z(G_l(p),C)\right] + \left[\sum_{l = 1}^L O(\alpha) \size{a_l} r_l\right] \\
            &\leq \left[\sum_{l =1}^L \sum_{g\in G_l^*} (\Count(g) + l N_G \PH_M) (1 + O(\alpha)) z(G_l(p),C)\right] + O(\alpha\OPT)  + O(\alpha \erralpha) \\
            &\leq \left[\sum_{l =1}^L \sum_{g\in G_l^*} \Count(g) z(G_l(p),C)\right] + O(L^2 N_G^2 \PH_M)  + O(\alpha\OPT)  + O(\alpha \erralpha) \\
            &\leq f_{D'} (C) + O(L^2 N_G^2 \PH_M)  + O(\alpha\OPT)  + O(\alpha \erralpha).
        \end{align*}
    \end{proof}
    
    \begin{corollary}
        \label{cor:proxyToOriginal}
        As a direct consequence of \cref{lem:1round_ProxyVersusActual}, it follows that
        \begin{align*}
            f_{D} (S^*) \leq (1 + O(\alpha))\eta \OPT + O(\alpha \erralpha) + O(L^2 N_G^2 \PH_M),
        \end{align*}
        where we absorb the $\eta$ factor in the additive error terms in the big-Oh notation and an $O(\alpha \OPT)$ term in the first term.
    \end{corollary}
    
    We now want to recover the cluster centers of the clusters derived from the low-dimension space by using the $\PSO$ derived from calls to $\HSO$. Since we identify points by their images in level-wise grids, we incur additional discretization error that must be accounted for. Concretely, the clustering we actually derive is not $p \mapsto \argmin_{s' \in S'} z(T(p),s')$ but instead given by the following definition.
    
    \begin{definition}
        \label{def:actualClustering}
        \begin{enumerate}
            \item Let $G'_* : \R^{d'} \to \left(\cup_{l=1}^L G_l\right)$ denote $G_l \circ Q (p')$ where $l$ is the minimum index such that $G_l \circ Q (p') \in G_l^*$. We then define a clustering of $D'$ via the solution $S^*$ by letting $D'(s_i^*) = \{ p' \in D' : \argmin_{s \in S^*} z(s, G'_* (p')) = s_i^* \}$. Alternatively, we can first define $G_{l,i}^* = \{g \in G_l^* : \argmin_{s \in S^*} z(g,s) = s_i^* \}$, then let $D'_l (s_i^*) = \{ p' \in D' : G_l \circ Q \in G_{l,i}^*  \mbox{ and }G_j \circ Q \not\in G_j^* \mbox{ for }j < l \}$ and then let $D'(s_i^*) = \cup_{l \in [L]} D'_l(s_i^*)$; these two formulations are equivalent.
            \item We see that with these definitions $a_l = Q(\sqcup_{s_i^* \in S^*} D'_l(s_i^*))$. Further, this also defines a clustering of $D$ by identifying each point in $D'$ with its dimension-reduced image in $D$, with clustering cost
            \begin{align*}
                \sum_{l =1}^L \sum_{p \in a_l} z(G_l(p), S^*) = \sum_{p'\in D'} z(G'_* (p'), S^*).
            \end{align*}
        \end{enumerate}
    \end{definition}

    \begin{lemma}
        \label{lem:actualDimRedCost}
        For the privately derived cluster centers $S^*$ in the dimension reduced space, we have the following bound for the clustering of $D$ as defined in \cref{def:actualClustering}.
        \begin{align*}
            \sum_{l = 1}^L \sum_{p \in a_l} z(G_l(p), S^*) = \eta (1 + O(\alpha))\OPT + O(\alpha \erralpha) + O(N_G^2 \PH_M \log^2 n).
        \end{align*}
        As a direct corollary,
        \begin{align*}
                \sum_{p' \in D'} z(G'_* (p), S^*) = \eta (1 + O(\alpha))\OPT + O(\alpha \erralpha) + O(N_G^2 \PH_M \log^2 n).
            \end{align*}
    \end{lemma}
    
    \begin{proof}
        We want to understand the increase in clustering cost due to discretization.
        \begin{align*}
            \sum_{l = 1}^L \sum_{p \in a_l} z(G_l(p), S^*) - z(p,S^*) &\leq \sum_{l = 1}^L \sum_{p\in a_l} z(G_l (p) ,\argmin_{s\in S^*} z(p,s)) -  z(p,S^*).
        \end{align*}
        Here we use the same trick of applying the $\ell_2$-triangle inequality in conjunction with the A.M.-G.M. inequality as in the proof of \cref{lem:1round_ProxyVersusActual} and bound $z(G_l(p) ,\argmin_{s\in S^*} z(p,s))$ from above by $O(\alpha r_i) + (1 + O(\alpha)) z(p,S^*)$. Continuing,
        \begin{align*}
            \sum_{l \in 1}^L \sum_{p \in a_l} z(G_l(p), S^*) - z(p,S^*) &\leq \sum_{i = 1}^L \sum_{p \in a_l} O(\alpha r_l) + \sum_{i = 1}^L \sum_{p \in a_l} O(\alpha) z(p,C^*)\\
            &\leq O(\alpha) \sum_{l=1}^L \size{a_l} r_l + \sum_{i = 1}^L \sum_{p \in a_l} O(\alpha) z(p,S^*) \\
            &\leq O(\alpha \OPT) + O(\alpha \erralpha) + O(\alpha) f_{D} (S^*).
        \end{align*}
        Since we use a non-private clustering algorithm with multiplicative approximation factor $\eta$, we can substitute for $f_D(S^*)$ by $\eta \OPT$, and rearranging terms we get the stated bound.
    \end{proof}
    
    We now use the error bounds for the sum oracle and the succinct histogram to recover the cluster centers of the cluster as defined in \cref{def:actualClustering}.
    
    \begin{lemma}
    \label{lem:recoveryError}
        For every cluster center $s_i^* \in S^*$, we have the following estimation error bound for the cluster centers of the clusters derived in the original space.
        \begin{align*}
            \norm{\frac{\sum_{l=1}^L \sum_{g \in G_{l,i}^*} \Sum (g)}{\sum_{l=1}^L \sum_{g \in G_{l,i}^*} \Count (g)} - \frac{\sum_{p' \in D'(s_i^*)} p' }{\size{D'(s_i^*)}}} \le \frac{2 L^2 N_G^2 }{\size{D'(s_i^*)} }  \left( \norm{\frac{\sum_{p' \in D'(s_i^*)} p' }{\size{D'(s_i^*)}}}\PH_M + \PSO_M \right) .
        \end{align*}
    \end{lemma}
    
    \begin{proof}
        The proof of this result is essentially the same as that of \cref{lem:noisyAvgGuarantee}, with the additional complication that we must account for the error accrued when summing over queries for multiple heavy values.
        \begin{align*}
            &\frac{\sum_{l=1}^L \sum_{g \in G_{l,i}^*} \Sum (g)}{\sum_{l=1}^L \sum_{g \in G_{l,i}^*} \Count (g)} - \frac{\sum_{p' \in D'(s_i^*)} p' }{\size{D'(s_i^*)}} \\
            = &\frac{\sum_{l=1}^L \sum_{g \in G_{l,i}^* } \Sum (g)}{\sum_{l=1}^L \sum_{g \in G_{l,i}^* } \Count (g)} - 
            \frac{\sum_{p' \in D'(s_i^*)} p' }{\sum_{l=1}^L \sum_{g \in G_{l,i}^*} \Count (g)} + \frac{\sum_{p' \in D'(s_i^*)} p' }{\sum_{l=1}^L \sum_{g \in G_{l,i}^*} \Count (g)} - \frac{\sum_{p' \in D'(s_i^*)} p' }{\size{D'(s_i^*)}} \\
            = &\frac{\sum_{l=1}^L \sum_{g \in G_{l,i}^*} \Sum (g) - \sum_{p' \in D'(s_i^*)} p' }{\sum_{l=1}^L \sum_{g \in G_{l,i}^*} \Count (g)} +  \frac{\sum_{p' \in D'(s_i^*)} p' }{\size{D'(s_i^*)}} \frac{\size{D'(s_i^*)} - \sum_{l=1}^L \sum_{g \in G_{l,i}^*} \Count (g)}{\sum_{l=1}^L \sum_{g \in G_{l,i}^*} \Count (g)} \\
            \Rightarrow &\norm{\frac{\sum_{l=1}^L \sum_{g \in G_{l,i}^*} \Sum (g)}{\sum_{l=1}^L \sum_{g \in G_{l,i}^*} \Count (g)} - \frac{\sum_{p' \in D'(s_i^*)} p' }{\size{D'(s_i^*)}}} \\
            \leq &\frac{ L \cdot N_G \cdot L N_G \PSO_M }{\size{D'(s_i^*)} - L \cdot N_G \cdot L N_G \PH_M } +  \frac{\sum_{p' \in D'(s_i^*)} p' }{\size{D'(s_i^*)}} \cdot \frac{L\cdot N_G \cdot L N_G \PH_M}{\size{D'(s_i^*)} - L \cdot N_G \cdot L N_G \PH_M }
        \end{align*}
        So for all clusters $D'(s_i^*)$ such that with at least $2 L^2 N_G^2 \PH_M $ many points, we can bound the $\ell_2$ estimation error by
        \begin{align*}
            &\frac{ 2 L^2 N_G^2 \PSO_M }{\size{D'(s_i^*)}} +  \norm{\frac{\sum_{p' \in D'(s_i^*)} p' }{\size{D'(s_i^*)}}} \cdot \frac{2 L^2 N_G^2 \PH_M }{\size{D'(s_i^*)} } \\
            &= \frac{2 L^2 N_G^2 }{\size{D'(s_i^*)} }  \left( \norm{\frac{\sum_{p' \in D'(s_i^*)} p' }{\size{D'(s_i^*)}}}\PH_M + \PSO_M \right) .
        \end{align*}
    \end{proof}
    
    Now we can derive the cost bound for the private clustering solution derived in the original space.
    
    \begin{lemma}
        \label{lem:originalSpaceCost}
        \begin{align*}
            f_{D'} (S') \leq (1 + O(\alpha))\eta \OPT' + O(\alpha^2 \erralpha \log n/\beta) + O(\alpha L^2 N_G^2 \PH_M \log n/\beta) + O(k L^2 N_G^2 \PSO_M).
        \end{align*}
    \end{lemma}
    
    \begin{proof}
        We are interested in bounding the clustering cost of $D'$ with respect to the clusters $(D'(s_1^*),\dots, D'(s_k^*))$. In \cref{lem:actualDimRedCost} we bounded the cost of the dimension reduced image of this clustering $(D(s_1), \dots, D(s_k)) = (Q(D' (s_1^*)), \dots , Q(D'(s_k^*)))$. From \cref{lem:costApproximation} we recall that for any clustering $(D'_1, \dots, D'_k)$ of $D$ we have that
        \begin{align*}
            \sum_{i \in k} \sum_{p \in D'_i} s\left( p , \frac{\sum_{q\in D'_i} q}{\size{D'_i}} \right) \simeq_{1 + \alpha} (\alpha \log n/\beta ) \sum_{i \in k} \sum_{p \in Q (D'_i) } s\left( Q (p) , \frac{\sum_{q\in D'_i} M(q)}{\size{D'_i}} \right).
        \end{align*}
        If we let $D'_i = D' (s_i^*)$, and denote the (unknown) true cluster centers of the clusters in the original space by $\mu_i = \frac{\sum_{p' \in D'(s_i^*)} p' }{\size{D'(s_i^*)}}$ for $i = 1, \dots, k$ then we get
        \begin{align*}
            f_{D'} (\{ \mu_1, \dots, \mu_k \} ) &\simeq_{1 + \alpha} (\alpha \log n/\beta) f_D (\{s^*_1, \dots, s^*_k \}) \\
            &\simeq_{1 + \alpha} \alpha \log n/\beta (1 + O(\alpha))\eta \OPT + O(\alpha^2 \erralpha \log n/\beta) + O(\alpha L^2 N_G^2 \PH_M \log n/\beta).
        \end{align*}
        Then, since $\OPT' \simeq_{1 + \alpha} \alpha \log n/\beta \OPT$, we can write
        \begin{align*}
            f_{D'} (\{ \mu_1, \dots, \mu_k \} ) &\simeq_{1 + O(\alpha)}  (1 + O(\alpha))\eta \OPT' + O(\alpha^2 \erralpha \log n/\beta) + O(\alpha L^2 N_G^2 \PH_M \log n/\beta).
        \end{align*}
        We have estimates 
        \begin{align*}
            \hat{\mu_i} = \frac{\sum_{l=1}^L \sum_{g \in G_i^* (s_i^*)} \Sum (g)}{\sum_{l=1}^L \sum_{g \in G_i^* (s_i^*)} \Count (g)}
        \end{align*} 
        for the true cluster centers $\mu_i$ for $i = 1, \dots, k$. From \cref{lem:approxClustering}, in order to bound the additive error incurred due to the estimation error, i.e. $f_{D'(s_i^*)}(\{\hat{\mu}_i\}) - f_{D'(s_i^*)} (\{\mu_i\})$, it will suffice to bound $\size{D'(s_i)^*} \norm{\mu_i - \hat{\mu}_i}^2$. \Cref{lem:recoveryError} bounds the estimation error $\norm{\hat{\mu} - \mu}$. Putting everything together, we get
        \begin{align*}
            f_{D'(s_i^*)}(\{\hat{\mu}_1, \dots, \hat{\mu}_k \}) - f_{D'(s_i^*)} (\{\mu_i\}) &\leq \size{D'(s_i^*)} \left(\frac{2 L^2 N_G^2 }{\size{D'(s_i^*)} }  \left( \norm{\frac{\sum_{p' \in D'(s_i^*)} p' }{\size{D'(s_i^*)}}}\PH_M + \PSO_M \right) \right)^2 \\
            &\leq \frac{8L^4 N_G^4 \PH_M^2}{\size{D'(s_i^*)}} + \frac{8 L^4 N_G^4 \PSO_M^2}{\size{D'(s_i^*)}} \\
            &\leq \frac{L^4 N_G^4 }{\size{D'(s_i^*)}} O(\PSO_M^2).
        \end{align*}
        For each $s_i^* \in S^*$, if $D'(s_i^*) \geq L^2 N_G^2 \PSO_M$, then the first factor is $O(L^2 N_G^2 \PSO_M)$. On the other hand, if $D(s_i^*) < L^2 N_G^2 \PSO_M$ then the clustering cost for $D'(s_i^*)$ i.e. $f_{D'(s_i)}(\{\hat{\mu}_1, \dots, \hat{\mu}_k \})$ is unconditionally less than $L^2 N_G^2 \PSO_M$ as the diameter of the data domain is $O(1)$. It follows that the additive error over all $k$ clusters is at most $O(k L^2 N_G^2 \PSO)$. Since $f_{D'(s_i^*)} (\{\mu_i\}) = f_{D'(s_i^*)} (\{\hat{\mu}_1, \dots, \hat{\mu}_k \})$, putting everything together we get that
        \begin{align*}
            f_{D'}(\{\hat{\mu}_1,\dots, \hat{\mu}_k\}) &\leq f_{D'} (\{\mu_1,\dots, \mu_k\}) + O(k L^2 N_G^2 \PSO_M) \\
            &\leq (1 + O(\alpha))\eta \OPT' + O(\alpha^2 \erralpha \log n/\beta) + O(\alpha L^2 N_G^2 \PH_M \log n/\beta) + O(k L^2 N_G^2 \PSO_M).
        \end{align*}
    \end{proof}
    
    \oneRoundGuarantee*
    
    \begin{proof}
       We make $2L = 2 \log n$ calls (in parallel) to $\BTG$ and $\HSO$. From their respective privacy guarantees, we know that each call is $(\epsilon, \delta)$-differentially private. By simple composition of privacy, it follows that the net privacy loss is $(2 (\log n) \epsilon, 2 (\log n) \delta)$. To ensure net $(\epsilon, \delta)$ privacy loss, we must scale the respective privacy parameters by a factor of $1/(2\log n)$; with this scaling we have $\PH_M = \tilde{O} \left( \frac{1}{\epsilon \alpha} \sqrt{n \log^5 n} \right)$ and $\PSO_M = \tilde{O} \left( \frac{c_G}{\epsilon} \sqrt{d' n \log^2 n} \right)$. We recall that $N_G = k^{\tilde{O} (1/\alpha^2)}$. Substituting all these bounds in the guarantee of \cref{lem:originalSpaceCost} we get
       \begin{align*}
            f_{D'}(\{\hat{\mu}_1,\dots, \hat{\mu}_k\}) &\leq (1 + O(\alpha))\eta \OPT' + O(\alpha^2 \erralpha \log n/\beta) + O(\alpha L^2 N_G^2 \PH_M \log n/\beta) + O(k L^2 N_G^2 \PSO_M) \\
            &\leq (1 + O(\alpha))\eta \OPT' + \frac{1}{\epsilon} k^{\tilde{O} (1/\alpha^2)} \sqrt{d' n \log 1/\delta} \poly\log n.
        \end{align*}
        To simplify the error term in the above expression we assume without loss that $k\geq 2$, as $k=1$ is a degenerate case i.e. mean estimation of vectors in $d'$ dimensional space. We then absorb all constants in the $\tilde{O}$ term in the exponent of the $k$ to state a simplified bound.
    \end{proof}

%% file: 4AlgorithmMultiRound.tex
\section{LDP \texorpdfstring{$k$}{k}-means with low additive error}\label{sec:fourRound}

    In this section we describe our second algorithm that, given a constant $c > \sqrt{2}$, can achieve a constant factor multiplicative approximation and $O(k^{O(1/(2c^2-1))}\sqrt{n d'} \poly\log n)$ additive error. Our algorithm is described in a modular fashion, and one may refer to the respective section for the pseudo-code and an informal walk-through of how the algorithm proceeds. We begin with a technical discussion to help explain some of the algorithmic choices made along the way.
    
    \begin{table}
        \centering
        \begin{tabular}{c|c}
             Notation & Meaning \\ \hline
             $D' \subset \R^{d'}$ & Original data set \\
             $Q : \R^{d'} \to \R^{d}$ & mapping from high-dim. to low-dim. space \\
             $D \subset \R^d$ & $Q(D')$, dimension reduced data set \\
             $L$ & Number of cell grid levels \\
             $\Cells_l$ & Grid of cells in dimension reduced space for $l\in [L]$ \\
             $t_l$ & Side-length of any cell in $\Cells_l$ (equals $2^{-l}$)\\
             $\Cells_l (\cdot)$ & Mapping from $\R^d$ to unique containing cell $\Cells_l$\\
             $\Ancestors^* : \Cells_l \to \Cells_{l- (3/2)\lg d}$ & Mapping from cells to the set of their ancestors $j$ with side-length $d^{3/2} t_l$\\
             $\CH^l$ & Succinct histogram of number of points mapping to $C \in \Cells_l$ for ``heavy" $C$ \\
             $F$ & Number of geometrically varying guesses for true optimal clusternig cost $\OPT$\\
             $\Heavy_l^f$ & Heavy cells identified $\CH^l$ where guess for $\OPT = k\sqrt{n} \cdot 2^f$, $f \in [F]$ \\
             $\Light_l^f$ & Cells which are not heavy \\
             $\Medium_l^f$ & Light children of heavy cells \\
             $M$ & Number of distance scales with which LSH functions applied \\
             $r_{l,1}, \dots, r_{l,M}$ & Scales at which LSH functions are used to allocate cluster centers for points in $\Medium_l$ \\
             $R$ & Number of repetitions of LSH subroutine to boost success probability \\
             $\Lambda_l^f$ & Synthetic space of heavy cells in $\Ancestors^*$ level\\
             $\Lambda_l^f (\cdot)$ & Mapping from $\R^{d}$ to synthetic space \\
             $H_{l,m,r,f} (\cdot)$ & $(p(1),p(c),r_{l,m},c r_{l,m})$-sensitive hash function with domain $\Lambda_l^f$ for points in $\Medium_l^f$\\
             $\BH_{l,m,r,f}$ & Histogram of number of points per hash bucket\\
             $\BSO_{l,m,r,f}$ & Vector sums of points in original space mapping to heavy buckets\\
             $\hat{b}$ & Average vector mapping to bucket $b \in \BH_{l,m,r,f}$ \\
             $\Pi_l(\hat{b})$ & projection of $\hat{b}$ to $\Lambda_l^f (\cup_{C \in \Heavy_l^f} C)$\\
             $S_l$ & Candidate centers allocated in one level for some guess of $\OPT$\\
             $S_{\Heavy}$ & Candidate centers allocated at the center of heavy cells for some guess of $\OPT$\\
             $S$ & $k$-means bi-criteria solution
        \end{tabular}
        \caption{Summary of notation used in \cref{alg:main}}
        \label{tab:4RoundNotation}
    \end{table}
    
    \subsection{Technical discussion}
    
    We recall from the introduction that any differentially private solution for $k$-means clustering in the local setting has to somehow indirectly access the aggregate geometry of the data set because of the high magnitude of the noise that is added to maintain privacy. We then discussed how discretizing the response function that is sensitive to the location of each point allows us to do precisely this and understand the geometry of the data set in sum. The one-round clustering algorithm uses a grid-based discretization of the domain to elicit a discrete response. For our four-round algorithm, we will use a combination of a cell-based discretization (which is similar in essence to the grid-based discretization used before) in combination with LSH functions.
    
    \paragraph{Dyadic hierarchy of cells:}{ In \cite{BFLSY17}, the authors describe a way of decomposing the data domain in a way that helps identify regions of the domain where data accumulates. Given a rectangular domain $[0,1)^d$, they construct a dyadic $2^d$-ary tree of cells, where each rectangular cell is sub-divided into $2^d$ child cells by bisecting the cell along each axis. The cell at the top of the hierarchy with side-length one unit is simply the whole domain, and it has $2^d$ children with side-length half units that collectively again cover the whole domain. Each child cell is recursively divided in the same manner, and in level $l$ the side length of each cell is $t_l = 2^{-l}$ units. The idea is that although each point in the domain is covered by each level of cells, the further down the hierarchy one goes the finer is the resolution at which the domain is discretized and the smaller is the diameter of the bounding box at a level. $L = \log n$ levels of the grid will suffice to discretize the domain to a sufficiently fine degree so as to capture clusters at all relevant scales; the cost of clustering cluster with radius smaller than $O(1/n)$ will be dominated by the additive error terms that any private $k$-means clustering algorithm must have. This entire construction can be done after an application of the JL transform for dimension reduction which ensures that $d = O(\log n)$. We will see during the course of our discussion why this is crucial for our cost analysis.

    The authors of \cite{BFLSY17} then observe that if we \emph{randomly shift} this hierarchy of cells, then one can show that with probability $1-\beta$, for any point in the domain there are at most $O(1/\beta)$ many cells with side-length $t_l$ within an $\ell_2$ distance of $t_l/d$ units of that point. Applying this on any choice of optimal centers $S_{\OPT}$ means that there are $O(k/\beta)$ many cells close to $S_{\OPT}$ at any level. How can we exploit this to capture the aggregate geometry of the data set?
    }
    
    \paragraph{Guessing the optimal cost:}{ Suppose that we knew what the optimal cost $\OPT$ were. If this were the case, then we can bound the number of cells further than $t_l/d$ units away from any choice of optimal centers $S_{\OPT}$ that carry significantly many data points. Concretely, all data points in cells further than $t_l/d$ units away from $S_{\OPT}$ must have a clustering cost of at least $t_l^2/d^2$. On the other hand, their clustering cost with respect to $S_{\OPT}$ cannot exceed the total clustering cost $\OPT$, which means there cannot be more than $\OPT d^2/t_l^2$ many such points. Tracing a similar argument with cells, we compute a threshold depending on the level's side length $t_l$ such that there cannot be more then $O(kL/\beta)$ many cells that have more than the number of points in the threshold and lie further than $t_l/d$ units away from $S_{\OPT}$. To see why the bound has changed from $O(k/\beta)$ to $O(kL/\beta)$, note that we scale the failure probability by a factor of $1/L$ so that it apply across all $L$ levels with probability $1-\beta$. Coupled with the guarantee that there cannot be more than $O(kL/\beta)$ many cells closer than $t_l/d$ to $S_{\OPT}$ we get that regardless of where they lie in the domain there are at most $O(kL/\beta)$ \emph{heavy} cells in any level, i.e. cells that beat the threshold $T_l$ for their level.
    
    In the top cell, this threshold is lower than $n$, so the top cell is always marked heavy. In the bottom level, this threshold exceeds $n$, so all bottom cells are marked \emph{light} (i.e. not heavy). Between these two extremes the threshold increases monotonically as $t_l$ decreases, which means that there is a unique level for every point where the cell it belongs to transitions from being heavy to light. There is a small technicality here that since we can only identify cell counts via noisy privatized responses we can inadvertently mark heavy cells light and light cells heavy. In practice we will appeal to the locally private histogram construction $\BTG$ of \cite{DBLP:journals/jmlr/BassilyNST20} to estimate the data point counts of cells. The issue of incorrect labelling of cells as heavy or light is readily resolved by requiring that heavy cells have only heavy ancestors, and using the accuracy guarantees of $\BTG$ to bound the consequences of such errors in our cost analysis.
    
    In sum, under the promise that $\OPT$ is known, we have identified regions of the domain at different scales where the data set accumulates beyond some thresholds. Since we are targeting an additive error of $\tilde{O}(k\sqrt{n})$, we let $\OPT$ vary in factors of $2$ from $k\sqrt{n}$ to $n$ and simply run the algorithm with varying values of $\OPT$ at different scales to ensure that the promise holds for at least some run. This leads to an inflation in our additive error on the order of $\log n$ as the number of candidate centers grows by this factor.
    
    We recall that in the introduction we mentioned that when finding a bi-criteria solution, to get $O(\poly k \sqrt{n}\allowbreak \poly\log n)$ error we would like to find $O(\poly k \poly\log n)$ many candidate centers with respect to which the data set has a clustering cost within a constant factor multiplicative approximation to $\OPT$ and additive error at most $O(\poly k \sqrt{n} \poly\log n)$. It is in fact the case that if we can limit the exponent of $k$ in both the number of centers allocated and the additive error incurred, then we will have at most that same exponent in the final error term. Keeping this in mind we observe that we have partitioned the data set across $O(k \log^2 n /\beta)$ many heavy cells. If we can allocate some $O(\poly\log n)$ centers in each cell such that the additive error with respect to these centers is $O(k  \sqrt{n}  \poly\log n)$, then we would achieve $O(k \sqrt{n} \poly\log n)$ error in sum. Although we do not achieve exactly this term, the reason we are able to get arbitrarily close to it is because of the relatively small number of cells within which we have partitioned the data set. We turn to using LSH functions to allocate candidate centers in a cell-wise fashion.
    }
    
    \paragraph{The $n^{1/2 + a}$ barrier:}{ We recall that a $(p,q,r,cr)$ LSH function has the property that if two points are within a distance of $r$ units, they must collide with probability at least $p$, and if two points are further than $cr$ units, then they cannot collide with probability more than $q$. By applying LSH functions on the data domain and appealing to locally private succinct histograms, we can recover all heavy LSH buckets; the idea then is that any sufficiently large cluster with radius less than $r$ units must populate one of these heavy buckets with a lot of points, possibly with some false positives. We estimate the point average over each heavy bucket to get a point that is no more than $cr$ units away from the cluster, and serves as a cluster center with a constant factor approximation to the true radius.
    
    We now describe why prior work taking this approach suffer an $O(n^{1/2 + a} \poly\log n)$ dependence on $n$ in the additive error. When dealing with LSH functions one technicality that has to be dealt with is that the LSH guarantee holds only in a pair-wise fashion, i.e. you only get bounds on the likelihood of points colliding a pair of points at a time. Fixing some cluster $C$ with radius $r$, this leads us to use some arbitrary fixed point from $C$ as a filter, using the LSH guarantee to argue that (1) ``most" points which collide with it under the LSH function with parameters $(p,q,r,cr)$ must lie at a distance of at most $cr$ units and (2) for every cluster, at least a $p$ fraction of points from that cluster must collide with it. What we would like to be the case is that the average over all points colliding with our filter lie at a distance of $O(cr)$ from the filter; since the filter itself lies in the cluster, by the triangle inequality the average can then serve as a candidate center for the cluster with an $O(c)$ constant factor approximation to the radius.
    
    Let $\Delta$ be the diameter of the data domain. The distance of the weighted mean of all points colliding with the filter under the LSH function, from that filter, can roughly be bounded from above by
    \begin{align*}
        cr \cdot |\{\mbox{points from }C\mbox{ colliding with filter}\}| + \Delta \cdot |\{ \mbox{points further than }cr\mbox{ units from filter}\}| 
    \end{align*}
    We are bounding the impact of points from outside the cluster by the diameter of the domain, and dealing with the arbitrarily many points that lie between a distance of $r$ and $cr$ units by simply inflating the distance considered ``close" to $cr$ units so they can be dropped from consideration without giving us an unfair advantage (notice that they can only pull this average towards $cr$ units). It is easy to see by linearity of expectation that the expected number of points from the cluster that collide with the filter is at least $p |C|$, and the number of points from further than $cr$ units that collide with the filter is at most $q |D|$ (again using the worst case as an upper bound).
    
    To get this weighted mean to be of the order of $cr$, we tune the LSH parameters to get the collision probability ratios to fulfill
    \begin{align*}
        cr \geq \frac{q|D|\Delta}{p|C|}.
    \end{align*}
    One can see this as a tug of war between false positives which in expectation increase with the side of the data set and whose impact is exacerbated by the diameter of the data domain and ``true" cluster points whose impact can be as low as $cr$ units and whose number scales with the size of the cluster $C$. Rearranging terms gives us
    \begin{align*}
        \frac{p}{q} \geq \frac{|D|}{|C|} \frac{\Delta}{cr}.
    \end{align*}
    It follows that if one needs this procedure to work for clusters $C$ with as few as $\sqrt{n}$ many points, as well as for cluster radii that are a $\poly (n)$ factor smaller than $\Delta$, then since $\size{D} = n$, one would need the ratio between the collision probabilities of near and far points to beat a $\poly(n)$ term.
    
    It is an intrinsic property of LSH functions that tuning parameters to increase the ratio between $p$ and $q$ causes both $p$ and $q$ to fall individually. This is an issue because we also need sufficiently many points from the cluster to accumulate in a bucket to ensure we can distinguish the heavy bucket from random noise; in expectation the number of true points accumulating in a bucket number drops with $p$. It is in fact the case that $p$ scales with $n^{-\Theta(1/c)}$ which leads us to try and boost the success probability with $n^{\Theta(1/c)}$ many independent runs. Since we cannot test which runs are successful and which are not, we are forced to include all bucket averages generated along the way as candidate centers; this $n^{\Theta(1/c)}$ factor in the number of centers is what leads to the greater than $1/2$ exponent of $n$ that is incurred in previous work applying LSH functions for clustering as discussed in the introduction. We can push this exponent arbitrarily close to $1/2$ by letting $c \to \infty$, but naturally this causes the multiplicative approximation guarantee to blow up.
    
    Even if we were to somehow reduce the number of possible false positives (i.e. the size of the data set $D$ that lies in the LSH domain) from $n$ to something that scales with the cluster size, there is still the issue that $\Delta/cr$ could again be $\poly (n)$. We must find a way to both limit the sizes of the subset of the data that participate in the LSH procedure as well as the diameter of the data domain within which that subset can lie. We describe how we achieve exactly this in the sequel.
    
    If we apply this LSH subroutine heavy cell by heavy cell, then the impact of any point from more than $O(cr)$ units can be at most the diameter of the cell, i.e. $2^{-l} \sqrt{d}$ in the $l$th level, which resolves the $n^{1/2 + a}$ issue for all LSH scales which are 
    $$\Omega\left(\frac{2^{-l}}{\poly\log n}\right).$$ 
    However, there are still two issues to be resolved. We have yet to bound the size of the data subset lying in the LSH domain, as we discussed is necessary. Further, if the lowest LSH scale is still $2^{-l}/\sqrt{n}$ (for example), then the ratio of collision probabilities still has a factor of $n$, which will lead to an exponent of $n$ greater than $1/2$, as described above. In order to get a truly $O(\sqrt{n} \poly\log n )$ term, we need to increase the smallest cut-off distance for the set of LSH scale parameters.
    }
    
    \paragraph{Limiting the sequence of LSH scale parameters:}{We first take a small detour and describe how a finite sequence of scale parameters is chosen for cluster radii when identifying a bi-criteria solution. The analysis fixes some arbitrary optimal clustering solution $S_{\OPT}$ and decomposes the data set using concentric rings around $S_{\OPT}$ at geometrically varying thresholds. More concretely, each partition of the data set consists of points which lie between $2^{-l}$ and $2^{-l+1}$ units for $l = 1,2,\dots$. The goal then is to allocate cluster centers so that for each partition we can derive the promise that most points are covered by some candidate center at a distance of $O(2^{-l})$. Since the optimal clustering distance was at least $2^{-l}$ units per partition, this would give us a bi-criteria solution with an $O(\OPT)$ cost.
    
    One typically tries to identify these partitions and allocate centers separately for each partition, but doing so requires that there be a finite (and in fact small) set of distance thresholds and partitions. One way of accomplishing this is to cut off the sequence of thresholds at $\log n$ and instead of promising a constant multiplicative approximation to the optimal clustering distance for points which lie closer than $2^{-\log n} = 1/n$ units to $\OPT$, one observes that as long as there is a candidate center at a distance of $O(1/n)$, the net clustering cost for the at most $n$ such points there could be is $O(n \cdot 1/n^2) = o(1)$. The cost of clustering such points is then treated as a small additive error term in the constant factor approximation guarantee.
    
    When using LSH at a sequence of geometrically varying scales, one runs into a similar issue of needing to identify a lower bound for the smallest distance at which we allocate candidate centers. If the smallest such scale is $t$ units, then as there could be as many as $n$ points within this distance we will need $t^2 n$ units to be dominated by $O(k\sqrt{n})$, which would require $t$ to scale with $O(1/\sqrt{n})$ in the case where $k$ is small. As discussed, we need to avoid a $1/\poly(n)$ scaling factor for the lowest threshold $t$ so as to avoid an exponent of $n$ greater than $1/2$; it follows that the only way to do this is to reduce the size of the data subset on which LSH being applied. Essentially, this issue has been reduced to other condition which we needed to fulfill; that of bounding the size of the data subset participating in the LSH subroutine.
    }
    
    \paragraph{Bounding the subset of $D$ participating in LSH subroutines:}{ We see that simply using LSH on heavy cells does not work as is since there could again be arbitrarily many points in a heavy cell; all we have is a lower bound on the number of points it contains. To derive an upper bound we instead focus on the \emph{light children of heavy cells}. By virtue of being light, they have fewer points than the threshold mentioned before; we will be able to show that in level $l$ where the side length $t_l = 2^{-l}$ the total number of points which lie in such cells is $O(d^2 \OPT/t_l^2)$. From the previous discussion, this will allow us to set the lower LSH scale parameter $t = O(t_l/(d \sqrt{L}))$ and incur only $O(\OPT/L)$ additional error per level, leading to an additional $O(\OPT)$ cost across all levels. Observe that the lowest LSH scale parameter is essentially $t_l/\poly\log n$, which implies a $\poly\log n$ ratio between the diameter of the cell to the scale parameter, which is exactly what we wanted. The additional $O(\OPT)$ additive error term is readily absorbed in our multiplicative approximation factor (as opposed to a small additive error as is usually the case). Since the dimension $d$ and the number of levels in our cell hierarchy are both $O(\log n)$, this means that we have successfully avoided an exponent greater than $1/2$ on the factor of $n$ in the additive error.
    
    However, there is a different sort of issue in the dyadic hierarchy approach that we have not yet addressed; for any level the collection of light children of heavy cells partitions the data in arbitrary ways. It need not be the case that a cluster will lie entirely inside the domain of a single LSH function when making calls to the LSH subroutine. How do we account for the division of clusters across data partitions and cells?
    }
    
    \paragraph{Clusters and cluster sections:}{ Let us denote the partition of the data set $D$ that lies in heavy cells in level $l-1$ but light cells in level $l$ by $D_l$. With this notation it follows from our observations regarding the existence of a unique level for each data point such that its containing cell is light for the first time when going down levels that $D_0,\dots, D_{L-1}$ form a partition of $D$. For any fixed optimal clustering solution, we see that each cluster too can be partitioned across all levels $D_l$. Based on the discussion above, we would ideally like to use LSH functions on $O(kL/\beta)$ many cells in level $l-1$ and elicit a response only from $D_l$ to ensure that the diameter of the bounding box is not too high and the number of points participating in the LSH subroutine is not too many. This implies that we only need to allocate cluster competitively with respect to the sections of the optimal clusters that lie in heavy cells. However, this could lead to $O(k^2 L/\beta)$ many cluster sections per level, which would lead to a candidate center set of size at least $\Omega(k^2 \poly\log n)$, leading to $\Omega(k^2 \sqrt{n} \poly\log n)$ error down the line. In order to try and reduce the exponent of $k$ in the number of cluster centers allocated, we make three technical algorithmic choices.
    
    Firstly, we allocate a candidate center at the center of every heavy cells (which would be at most $O(kL^2/\beta)$ many more candidate centers). This gives us the guarantee that every point in the data set partition $D_l$ has a candidate center at a distance of $2^{-l}\sqrt{d}$. Secondly, we go up a few levels and apply LSH functions to the ancestors of these heavy cells of interest which have side-length $d^{3/2} 2^{-l}$. The consequence of these two modifications is that we only need to allocate cluster centers within a distance of $2^{-l}\sqrt{d}$ units of any point of $D_l$, and that since there are only $O(L/\beta)$ many cells with side-length $d^{3/2} 2^{-l}$ within a distance of $2^{-l}\sqrt{d}$ units of an optimal center, there are only $O(kL/\beta)$ many cluster sections we must account for. 
    
    Thirdly, in order to avoid dealing with the worst case $O(k)$ many cluster sections for every heavy cell when calling the LSH subroutine heavy cell by heavy cell, we construct a synthetic space out of the union of all heavy cells in a level and apply the LSH subroutine on this entire space. We will be able to extend the $\ell_2$ metric in a natural way to work across cells, ensure that the cells are far enough apart in this distance measure so that bucket averages that land up ``between" cells end up in the correct cell after projection, and that the diameter of this synthetic space is still small enough to keep the improvements we have derived so far.
    
    There is one final technical point which must be addressed; we need to identify a lower bound on the cluster section size to ensure that the ratio of the participating subset of the data to the size of the cluster section does not grow to $\poly (n)$, which would lose us the advances we have made. Since there are at most $O(kL/\beta)$ many such cluster sections in a level, we simply set the threshold to be $\OPT \cdot \frac{\beta}{kL} \cdot \frac{1}{L} \cdot \frac{1}{dt_l^2}$. Why does this work? We recall that we allocated a cluster center at the center of every heavy cell, that ensures that any cluster section has a candidate center at a distance of $\sqrt{d} t_l$, so for a cluster section below the threshold the cluster cost can be at most $\OPT \cdot \frac{\beta}{kL} \cdot \frac{1}{L}$. Then, since there are at most $O(kL/\beta)$ many such cluster sections, the net clustering cost for any one level across all cluster sections is $\OPT \cdot \frac{1}{L}$. Summing this up over all $L$ levels leads to an additional $\OPT$ term which again we can absorb into our constant factor approximation.
    
    We can summarize this lower bound on the cluster section size as $O(\frac{\OPT}{k t_l^2 \poly\log n})$. We recall that the size of the participating data set $D_l$ was at most $O(\OPT/t_l^2)$, which implies a ratio of $k \poly\log n$. A dependence on $n$ in the ratio that the LSH collision probabilities have to beat has been replaced by a dependence on $k$, leading to a $O\left( k^{1 + O(1/2c^2-1)} \sqrt{n} \poly\log n\right)$ bound on the number of candidate centers allocated.
    }
    
    \paragraph{Constructing the proxy data set and undoing dimension reduction:}{ In the one-round algorithm, we constructed a set of candidate centers and undid the dimension reduction in essentially one round of interaction. However, doing everything in one round increases the exponent of $k$; this was not apparent in that analysis unless studies it carefully since the big-Oh term in the power of $k$ in the number of candidate centers dominated any similar order increases (such as being squared) in the big-Oh notation. Since our goal in this section is to keep the error as low as possible, we avoid reducing the round complexity and instead use two rounds of interaction; one to construct the proxy data set, and one to recover the cluster centers in the original space.
    
    The construction of the proxy data set is relatively straightforward; we release the collection of candidate centers found and invite agents to privately reveal which candidate center is closest to them. Again by an appeal to $\BTG$, we estimate the number of data points a candidate center serves and construct a proxy data set by repeating each candidate center with multiplicity equal to its respective estimate. We then apply the non-private clustering algorithm of our choice on the privately generated proxy data set to get cluster centers $S^* = \{s_1^*, \dots, s_k^*\}$ in the dimension reduced space.
    
    In the final round of interaction we reveal the set $S^*$, and we ask agents to privately reveal a $k$-tuple of $d'$-dimensional vector where the $i$th vector equals its true location if $s_i^*$ is its closest cluster center in the dimension reduced space, and is otherwise the $0$ vector. In the same round of interaction, we ask them to reveal which is the center closest to them. We then simply compute the noisy sum for each of the $k$ coordinates and divide that by the noisy count of the number of points mapping to the center corresponding to that coordinate; we will be able to show that this estimate for the cluster center in the original space works well in its place for a $k$-means clustering solution.
    }
    
    \paragraph{Outline:}{ We divide the description and technical analysis of this algorithm into 4 parts. In \cref{subsec:grid} we formally describe the dyadic hierarchy of cells needed to construct the algorithm. In \cref{subsec:candidates}, we use the identification of heavy and light cells in the previous subsection to partition the data set level-by-level. Fixing any level, we prove that for any fixed optimal clustering, with probability $1-\beta$ we allocate candidate centers for most points in the partition corresponding to that level at an $\ell_2^2$ distance at most $O(c^2)$ times their distance from the optimal centers. In \cref{subsec:cost} we use the guarantees of \cref{subsec:candidates} to show that the sum-of-squares cost of clustering the dimension reduced dataset via the candidate centers is $O(\OPT)$ modulo some additive error. We go on to show by applications of the weak triangle inequality that the clustering functions of the proxy dataset and the dimension reduced dataset are close in $\ell_2^2$ distance up to an $O(\OPT)$ additive error. Then we bound the cost of the original dataset with respect to cluster centers derived via the dimension reduced clustering and account for the privacy loss to derive our net cost guarantee. }

\subsection{The cell grids and their hierarchy}\label{subsec:grid}

    \begin{algorithm}
	    \caption{Heavy cell marker}
	    \label{alg:heavylight}
	    \KwData{For every level $l\in [L]$, a succinct histogram of heavy-hitter cells $\CH^l$ with error bound $\CH^l_E$ and maximum frequency omitted $\CH^l_M$.}
	    \For{$l \in [L]$}{ 
	        $\Heavy_l \leftarrow \emptyset$\\
	        $\Light_l \leftarrow \emptyset$
	   }
	    $\Heavy_0 \leftarrow \Cells_0 $ \label{alg:heavylight;line:topheavy}\\
	    \For{$i=l \in \{2,\dots,L-1\}$}{
            \For{$C \in \;\CH^l$}{
                \uIf {$\CH^l (C) \geq \frac{\beta d^2 \OPT}{t_l^2 k L d} + \CH^l_E$ and $\Ancestors_1(C_j) \in \Heavy_{l-1}$}{ 
                    $\Heavy_l \leftarrow \Heavy_l \cup \{ C_j \}$
                }
                \uElse{ 
                    $\Light_l\leftarrow \Light_l \cup \{ C_j \}$
                }
            }
            $\Light_l \leftarrow \Light_l \cup \Cells_l \backslash \Heavy_l $\\
		}
		$\Light_{L-1} \leftarrow \Cells_{L-1}$ \label{alg:heavylight;line:bottomlight} \\
		\For{$l \in [L]$}{
		$\Medium_l \leftarrow \{C \in \Light_l : \Ancestors_1 (C) \in \Heavy_{l-1} \}$
		}
		\KwRet{$\{\Heavy_l,\Light_l,\Medium_l : l \in [L]\}$}
	\end{algorithm}

    In this subsection we formally define the cell grid hierarchy used to allocate candidate centers in the next section and describe an algorithm that uses succinct histogram of cell counts to tag cells as being either \emph{heavy} or \emph{light}. Apart from the definitions made, the main results of this subsection are \cref{lem:cellProps} which guarantees lower and upper bounds for the number of data points that can lie in heavy and light cells respectively; and \cref{lem:additiveErrorByLayer}, which shows how we can use the identification of heavy and light cells to partition the data set $D$, one partition per level, to get the subsets $D_l$ for $l\in [L]$. 
    
    We work over the domain $[0,1)^d$. We start by dividing this domain recursively in a dyadic fashion, with $L = \lceil \lg n \rceil$ levels in all.
	
	\begin{definition} We formalize the construction of the dyadic hierarchy of cells.
	\begin{enumerate}
	    \item A \emph{cell} is a dyadic cube in $(0,1]^d$. Explicitly, if we let the set of cells at level $l$ be denoted $\Cells_l$; then 
		\begin{align*}
			\Cells_l := \left\{ \prod_{e=1}^d \left[\frac{j_e}{2^l},\frac{j_e+1}{2^l}\right) : j \in \{0,1,\dots,2^l-1\}^d \right\}.
		\end{align*}
		We also define the notation $\Cells := \cup_i \Cells_i$.
		\item We let $t_l = 2^{-l}$ for $l \in [L]$; with this notation, every $C \in \Cells_l$ has side-length $t_l$. Note that with these definitions the minimum side-length $t_L \leq \frac{1}{n}$. 
	    \item For all $l \in [L]$, $\Ancestors_i : 2^{\Cells} \to 2^{\Cells}$ and $\Children_i : 2^{\Cells} \to 2^{\Cells}$ are defined by the following expressions:
    	\begin{align*}
    		\mbox{for }\Cells' \subset \Cells_l,\;\Ancestors_i (\Cells') &= \{ C \in \Cells_{l-i} : C' \subset C \mbox{ for some }C'\in \Cells'  \},\\
    		\mbox{for }\Cells' \subset \Cells_l,\;\Children_i (\Cells') &= \{ C \in \Cells_{l+i} : C \subset C' \mbox{ for some }C'\in \Cells'  \}.
    	\end{align*}
    	We set $\Cells_i = \{[0,1)^d\}$ for $i<0$; with this definition seeking the ancestor at a level above $0$ always returns the entire domain. It will not be necessary to define cells below level $L$. We abuse notation so that any singleton set of cells is identified with the element in it; with this notation we also have $\Children_l: \Cells \to 2^\Cells$ and $\Ancestors_l: \Cells \to \Cells$.
    	\item $\Ancestors^* := \Ancestors_{1.5 \lg d}$. Note that for $C \in \Cells_l$, $\Ancestors^* (C)$ is the unique cell that contains $C$ and has side length $d^{3/2} t_l$.
	\end{enumerate}
	\end{definition}
	
	We recall the following lemma from \cite{BFLSY17}:
	
	\begin{lemma}[Lemma 2.2, \cite{BFLSY17}]
	    \label{lem:bflsy}
	    Let $S$ be a finite set of points and $\Cells_R$ be a $d$-dimensional rectangular grid of cells with unit length $R$ shifted by a uniformly random displacement in $[0,R]$ along each dimension. With probability at least $1 - \beta$, $\size{\{C \in \Cells_R : z(S,C) \leq R^2/d^2 \}} = O(|S|/\beta)$.
	\end{lemma}
	
	\begin{proof}
        The number of cells with side-length $R$ within an $\ell_2$ distance of $x<R/2$ from any $s \in S$ can be bounded from above by $N_1 N_2 \dots N_d$ where $N_i$ is $1$ if there is no cell wall within a distance of $s$ along the $i$th dimension and $2$ otherwise. Since the random shifts along each dimension are independent, it follows that
        \begin{align*}
        \Ex\left[\prod_{i=1}^d N_i\right] &= \prod_{i=1}^d \Ex[N_i].
        \end{align*}
        The probability of $s_i$ lying within a distance of $x$ units of one of the sides is $(2x/R)$. It follows that
        \begin{align*}
        \Ex[N_i] &= \left(1 - \frac{2x}{R}\right) \cdot 1 + \frac{2x}{R} \cdot 2 \\
        &= 1 + \frac{2x}{R}.
        \end{align*}
        Substituting in the first display, we get
        \begin{align*}
        \Rightarrow \Ex\left[\prod_{i=1}^d N_i\right] &= \left(  1 + \frac{2x}{R} \right)^d.
        \end{align*}
        It follows that for $x \leq \frac{R}{d}$, the expected number of cells within an $\ell_2$ distance of $x$ from $s$ is at most $O(1)$. By linearity of expectation, the expected number of cells within a distance of $R/d$ of $S$ is at most $O(|S|)$. By Markov's inequality, with probability $1- \beta$, the number of cells within a $\ell_2$ distance of $R/d$ from $S$ is $\leq O(|S|/\beta)$. The result follows directly as by definition $z(\cdot,\cdot)$ is the $\ell_2^2$ distance.
	\end{proof}
	
	\begin{remark}\label{rem:closeCellBound}
        We fix any arbitrary optimal $k$-means clustering solution $S_{\OPT}$ with clustering cost $\leq \OPT$ and condition on the event that the number of cells in any level $l$ within a distance of $t_l/d$ from $S_{\OPT}$ is at most $O(kL/\beta)$. By scaling the failure probability in \cref{lem:bflsy} by a factor of $\frac{1}{L}$ and applying the union bound over $L$ such events (one for each level) we see that this event holds with probability $1-\beta$. We will also assume that $\OPT \geq k \sqrt{n}$. Note that if $\OPT< k \sqrt{n}$ then for any choice of $S_{\OPT}$ $f_D(S_{\OPT}) \leq k\sqrt{n}$. Once we obtain an $O(1)$ multiplicative approximation to $k\sqrt{n}$, this term can be absorbed by the additive error term, so the guarantee as stated will hold unconditionally.
    \end{remark}
    	
	We partition the $\Cells_l$ into collections of \emph{heavy} and \emph{light} cells at every level depending on the number of data points within each cell. We perform this partitioning algorithmically via \cref{alg:heavylight}.
	
	\begin{definition}
	\label{def:heavylight}
	$C \in \Cells$ is called \emph{heavy} if $C \in \Heavy_l$ for some $l\in [L]$ where $\Heavy_l \subset \Cells_l$ is defined by the output of \cref{alg:heavylight}. Similarly, $C \in \Cells_l$ is called \emph{light} if $C \in \Light_l$ for some $l \in [L]$ for $\Light_l$ defined by the output of \cref{alg:heavylight}. We summarize the notation of \cref{alg:heavylight} for cells and collections of cells here:
    	\begin{enumerate}
    	    \item We denote the set of heavy cells at level $l$ by $\Heavy_l$, and the set of light cells by $\Light_l$.
    	    \item We denote the collection of all heavy cells by $\Heavy := \cup_l \Heavy_l$ and the collection of all light cells by $\Light := \cup_l \Light_l$.
    	    \item The center of any cell $C$ is denoted by $o(C)$ (this may be thought of as the \emph{origin} of $C$).
    	    \item The cell at level $l$ containing $p\in D$ is denoted by $\Cells_l(p)$.
    	\end{enumerate}
	\end{definition}
	
	We summarize some basic properties of heavy and light cells in \cref{lem:basicProps} as a sanity check.

	\begin{lemma}\label{lem:basicProps}
	    The following statements hold:
		\begin{enumerate}
		    \item $\forall l\in [L],\Cells_l = \Heavy_l \sqcup \Light_l $. 
			\item If $C \in \Light$, then $\Children_l (C) \subset \Light \; \forall l \geq 0$.
			\item $\Heavy_0 = \Cells_0 = \{[0,1)^d\}$
			\item $\Light_{L-1} = \Cells_{L-1}$.
		\end{enumerate}
	\end{lemma}
	\begin{proof}
	    \begin{enumerate}
	        \item This statement follows from the algorithm description - for every level $<L$ a cell recovered from $\BTG$ is marked either heavy or light, and all other cells in that level are marked light. For level $L$ all cells are marked light.
	        \item This statement holds because a cell is marked heavy only if its parent was marked heavy in the previous iteration.
	        \item This statement holds by \cref{alg:heavylight;line:topheavy} of \cref{alg:heavylight}.
	        \item This statement holds by \cref{alg:heavylight;line:bottomlight} of \cref{alg:heavylight}.
	    \end{enumerate}
	\end{proof}
	
	 The definition of the succinct cell count histograms $\CH^l$ occurs later in this section in \cref{alg:main}. Its properties follow entirely from the $\BTG$ guarantee and the definition of the value mapping. Since it is used in \cref{alg:heavylight} and is necessary in the analysis of \cref{alg:heavylight}, we will state and prove them here.
	
	\begin{corollary}
	    \label{cor:CH}
	    $\CH^l_E = O\left(\frac{1}{\epsilon_{\CH}} \sqrt{n \log n/\beta} \right)$ and $\CH^l_M =  O\left(\frac{1}{\epsilon_{\CH}} \sqrt{n \poly\log n/\beta} \right)$. 
	\end{corollary}
	
    \begin{proof}
        Fix any $l\in [L]$. Looking ahead, we see that $\CH^l$ is derived from a call to $\BTG$ on \cref{alg:main;line:CHdef} of \cref{alg:main} with mapping $f_l : p \mapsto C_l(p)$, privacy parameter $\epsilon_{\CH}$, and failure probability $\beta/L$. We note that the size of the co-domain for the mapping $f_l$ is at most $2^{dL}$. Since $d, L = O(\log n)$, substituting we get the stated bounds.
    \end{proof}
    
    Note that since $|V| = 2^{dL} = \Omega(n)$, we can bound $\CH^l_E = O(\CH^l_M)$.
    
    \begin{remark}
        We recall that the significance of \cref{cor:CH} is that the $\BTG$ guarantee gives us that with probability $1-\beta$, for every $C \in \Cells_l$ such that $\size{D \cap C} \geq \CH^l_M$, $C \in \CH^l$ and $\abs{\CH^l (C) - \size{D \cap C}} \leq \CH^l_M$.	
    \end{remark}
	
    In \cref{lem:cellProps} we characterize the accumulation of data in heavy and light cells across different levels.

    \begin{lemma}\label{lem:cellProps}
	    For all $l\in [L]$, the following properties hold:
	    \begin{enumerate}
	        \item If $C \in \Heavy_l$, $\size{D \cap C} \geq \max\left( \CH^l_M, \frac{\beta \OPT}{t_l^2 kLd}  \right)$.
	        \item If $C \in \Light_l$, $\size{D \cap C} < \min \left( \CH^l_M, \frac{\beta \OPT}{t_l^2 kLd} + 2 \CH^l_M  \right) = \CH^l_M$.
	        \item $\size{\Heavy_l} \leq O\left(\frac{kL}{\beta}\right)$.
	    \end{enumerate}
	\end{lemma}
	
	\begin{proof}
    	\begin{enumerate}
    	    \item If a cell $C \in \Cells_l$ is marked heavy, then it must have occurred in the histogram $\CH^l$ and so $\size{D \cap C} \geq \CH^l_M$, or it is the solitary top cell. In the former case, since the count estimate crossed the threshold to be considered heavy, $\size{D \cap C} \geq \frac{\beta \OPT}{t_l^2 kLd} + \CH^l_E - \CH^l_E = \frac{\beta \OPT}{t_l^2 kLd}$. In the latter case, substituting $l=0$ we see that the desired lower bound is $\frac{\OPT}{kLd}$. Since $\OPT$ can be at most $n$, and $\size{D \cap C} = n$, the bound holds.
    	    \item If a cell $C$ is marked light then either $\size{D \cap C} < \CH^l_M$, $\Ancestors_1 (C) \in \Light_{l-1}$, or it is a bottom level cell. In the first three cases, by induction down the levels $l$, since $\size{D \cap C} \subset \size{D \cap \Ancestors_1 (C)}$ and both $\frac{\beta d^2 \OPT}{t_l^2 kLd}$ and $\CH^l_M$ are monotonically increasing with $l$, the result follows by the induction hypothesis. Note that since the top cell is always marked heavy, the base case is vacuously true. In the last case, we substitute $l = \log n - 1$ to get $\CH^l_M \geq \frac{\beta  \OPT}{t_l^2 kLd} \geq \frac{\beta n^2 k \sqrt{n}}{kLd} \geq \beta n^{2.5}/\log^2 n$, which is asymptotically impossible for failure probability $\beta \geq \frac{1}{n^2}$ and certainly for $\beta = \Theta(1)$.
    	    \item We fix any optimal $k$-means solution $S_{\OPT}$. By \cref{rem:closeCellBound}, $\size{\{C \in \Cells_l: z(C,S_{\OPT}) \leq 1/(4^{l}d^2)\}} = O(kL/\beta)$. For any $C \in \Heavy_l$ such that $z(C,S_{\OPT}) > 1/(4^{l}d^2)$, from statement 1 we have that $|C \cap D| =\max\left(\frac{\beta \OPT}{t_l^2 kLd}, \CH^l_M \right) \ge \frac{\beta d^2 \OPT}{t_l^2 kL}$. It follows that $f_{C \cap D} (S_{\OPT}) > \frac{t_l^2}{d^2} \cdot \frac{\beta d^2 \OPT}{t_l^2 kL} = \frac{\beta \OPT}{kL}$. Since $\sum_{C \in \Heavy_l} f_{C \cap D} (S_{\OPT}) < f_{D} (S_{\OPT}) \leq  \OPT$ it follows that there cannot be more than $\frac{kL}{\beta}$ many such $C$. In sum, $\lvert \Heavy_l \rvert \leq O\left( \frac{kL}{\beta} \right)$.
    	\end{enumerate}
	\end{proof}
	
	We now define a decomposition of the data set $D$ using the definitions of the heavy and light cells.
	
	\begin{definition}
	    For $l\in [L]$, we define $D_l = \{ p \in D: \Cells_{l-1} (p) \in \Heavy, \Cells_l (p) \in \Light \}$.
	\end{definition}
	
	\begin{lemma}\label{lem:additiveErrorByLayer}
	    The following statements hold.
	    \begin{enumerate}
	        \item $D = \sqcup_{l \in [L]} D_l$.
	        \item $\forall l \in [L],\; \size{D_l} = O(d^2 \OPT/t_l^2) + O\left( \frac{kL \CH_M}{\beta} \right)$.
	        \item $\sum_{l \in [L]} \frac{1}{4^{l} d^2 L} \lvert D_l \rvert = O( \OPT ) + O\left( \frac{kL \CH_M}{\beta} \right)$.
	    \end{enumerate}
	\end{lemma}
	
	\begin{proof}
	    \begin{enumerate}
	        \item By \cref{lem:basicProps}, we see that the solitary top cell $[0,1)^d$ is heavy, and that $\Cells_L \subset \Light$. Further, for every $l \in [L]$, if $C \in \Light_l$ then $\Children_j (C) \in \Light$. It follow that for any point $p$, in the sequence of cells $\Cells_0 (p), \Cells_1 (p) , \dots, \Cells_L (p)$ there exists a unique index $l^*$ such that $\Cells_{l^*-1} (p) \in \Heavy$ and $\Cells_{l^*} (p) \in \Light$. The existence of such an index shows that the sets $D_l$ cover $D$, and the uniqueness shows that this is in fact a partition.
	        \item By definition, $D_l$ is a subset of $\cup_{C \in \Light_l} C$, which means we can write $D_l = \cup_{C \in \Light_l} D \cap C$. This union can in turn be written as a disjoint union of points in light cells at a distance $\leq t_l/d$ from $S_{\OPT}$ and points in light cells $>t_l/d$ away from $S_{\OPT}$. From \cref{lem:cellProps}, any light cell contains at most $\CH_M$ many points of $D$. Since there are at most $O(kL/\beta)$ cells with side length $t_l$ within a distance of $t_l/d$, it follows that there are at most $\CH_M \cdot O\left(\frac{kL}{\beta}\right)$ many points within a distance of $t_l/d$ from $S_{\OPT}$. Since the total clustering cost for $S_{\OPT}$ must equal $\OPT$, there can be at most $d^2 \OPT_l^2$ many points more than $t_l/d$ away from $S_{\OPT}$. Therefore in sum $\size{D_l} \leq O(d^2 \OPT/t_l^2) + O\left( \frac{kL \CH_M}{\beta} \right)$.
	        \item The second bound follows directly from the first.
	    \end{enumerate}
	\end{proof}

\subsection{Candidate center allocation} \label{subsec:candidates}
	
	\begin{algorithm}
	    \caption{Candidate Center Allocation for $k$-Means in Dimension-Reduced Space given $\OPT$}
        \label{alg:main;alg:CCA}
        \KwData{\mbox{Guess for }$\OPT = k\sqrt{n} \cdot 2^f$, Cell labels $(\Heavy_l, \Medium_l, \Light_l)$, Bucket histogram $\BH_{l,m,r}$, Bucket Sum Oracle $\BSO_{l,m,r}$ for $l\in [L]$, $m \in [M]$, $r \in [R]$, }
        Data drawn from global variables: number of levels $L$, number of LSH scales $M$, number of repetitions $R$\\
        \SetKwProg{DoParallelFor}{Do in parallel for}{:}{end}
        \SetKwBlock{DoParallel}{Do in parallel:}{end}
        $S_{\Heavy} \leftarrow \{ o(C): \exists i C \in {\Heavy}_i \}$\\
        $T_l = \frac{p(1)}{2} \cdot \max\left( \frac{\beta \OPT}{t_l^2 kL^2 d}, \frac{4 \BH_M}{p(1)}, O\left( \frac{c_G \sqrt{n \poly\log n/\beta}}{\epsilon_{\BSO}}  \right)\right)$ \tcc*{Bucket threshold}
        $S_l \leftarrow \emptyset$\\
        \For{$l \in [L]$, $m\in [M]$, $r\in [R]$}{
            \For{$(b,\hat{n}_b) \in \BH_{l,m,r}$ such that $\hat{n}_B \geq T_l $}{
                $\hat{b} \leftarrow \frac{\BSO_{l,m,r}(b)}{\BH_{l,m,r}(b)}$\\
                $\Pi_l(\hat{b}) \leftarrow$ project $\hat{b}$ to $\Lambda_l^f (\cup_{C \in \Heavy_l^f} C)$\\
                $S_l \leftarrow S_l \cup\{\hat{b} \}$
            }
        }
        \KwRet{$S_{\Heavy} \cup \bigcup_{l=1}^L S_l$}
    \end{algorithm}
    
    \begin{algorithm}
        \caption{LDP $k$-means with low additive error}
        \label{alg:main}
        \SetKwProg{DoParallelFor}{Do in parallel for}{:}{end}
        \SetKwBlock{DoParallel}{Do in parallel:}{end}
        Setting: {Distributed dataset $D' \subset \R^{d'}$ over $n$ agents}\\
        \tcc{Step 1: Initialization and first interaction}
        $\gamma \leftarrow$ uniformly random vector in $[-1/2,1/2]^d$ \\
        $T : \R^{d'} \to \R^d$ dimension reduction for $d = O(\log (k/\alpha\beta)/\alpha^2)$ \\
        $S : \R^d \to \R^d$ scaling by a factor $\frac{1}{2(1+\alpha)}$\\
        $P : \R^d \to B(0,1/2)$ projection to $B(0,1/2)$ followed by translation by $\gamma$\\
        $Q = P \circ S \circ T : \R^{d'} \to B(0,1) \subset \R^d$ \\
        $L = \lg n$\\
        \DoParallel{         
        $\CH^l \leftarrow \BTG (\Cells_{l} \circ Q (\cdot), \epsilon_{\CH}, \beta/L)$ \label{alg:main;line:CHdef} \tcc*{Cell-wise Histogram of points}
        }
        $F =  \log_2 \frac{n}{\sqrt{n}k}$ \tcc*{Exponent of $2$ in guess for $\OPT$}
        \For{$f \in [F]$ }{
            $\{ {\Heavy}_i^f, {\Light}_i^f, \Medium_i^f : i \in [L] \} \leftarrow\HCM (\{\CH^l : l\in [L]\}, \mbox{guess for }\OPT = k\sqrt{n} \cdot 2^f)$
        }
        \tcc{Step 2: Candidate center allocation and second interaction}
        $M = 1 + \log_{2} d^{3/2} \sqrt{L} = O(\log \log n)$ \tcc*{Number of LSH scales}
        $r_{l,m} = \frac{2^m t_l}{ d \sqrt{L}}$ for $m \in [M]$ \tcc*{LSH scales}
        $R = O\left(\frac{\log kL^2/\beta}{p(1)}\right)$ \tcc*{Number of repetitions for LSH}
        $\lambda_l = (14c + 5)t_l \sqrt{d}$ \\
        $\Lambda_l^f(\cdot) := p \mapsto \begin{cases} (\lambda_l 1_{\Ancestors^*(C(p))},p - o(C_l(p)))
        \mbox{ if }p \in \cup_{C \in \Ancestors^* (\Heavy^f)} C \\ 0 \mbox{ otherwise } \end{cases}$\tcc*{Mapping to LSH domain}
        \begin{align*}
            H_{l,r,m,f} (p)   = \begin{cases}
                                (p(1),p(c),r_{l,m}, cr_{l,m}) \mbox{-sensitive Hash function on the space  }\Lambda_l^f \mbox{ if }\Cells_l(p) \in \Medium_l \\
                                \perp \mbox{ otherwise }
                                \end{cases}
        \end{align*}
        \DoParallelFor{$f \in [F], l\in [L], m\in [M], r\in [R]$}{	            
            $\BH_{l,m,r,f} \leftarrow \BTG (H_{l,m,r}^f, \beta, \epsilon_{\BH})$ \label{alg:main;line:BHdef} \tcc*{Bucket-wise Histogram of points}
            $\BSO_{l,m,r,f} \leftarrow \HSO (H_{l,m,r,f}, \Lambda_l,\beta,\epsilon_{\BSO})$ \label{alg:main;line:BSOdef} \tcc*{Bucket Sum Oracle}
        }
        $S \leftarrow \emptyset$\\
        \For{$f \in [F]$}{
        $S^f \leftarrow$ \cref{alg:main;alg:CCA} $(\mbox{Guess for }\OPT = k\sqrt{n} \cdot 2^f, \{ {\Heavy}_i^f, {\Light}_i^f, \Medium_i^f : i \in [L] \}, \BH_{l,m,r,f}, \BSO_{l,m,r,f})$
        $S \leftarrow S \cup S^f$\\
        }
        \KwRet(\cref{alg:2RoundRecovery}$(S, Q)$)
    \end{algorithm}
    
    We begin by giving a brief overview of the main steps in \cref{alg:main}.
    
    \paragraph{Step 1 - Initialization and first interaction:}{ We start by setting up the dimension reduction, scaling and projection map $Q$. We then have our first round of interaction with the agents where we make $L$ calls to $\BTG$ in parallel to receive estimates of how many points lie in each cell. We then make geometrically varying guesses for $\OPT$ $k\sqrt{n} 2^f$ for $f \in [F]$ where $F = \log_2 \frac{n}{\sqrt{nk}}$; note that with this definition our guesses vary in powers of two from $k\sqrt{n}$ to $n$. For each guess we generate a marking of cells $\{\Heavy^f_l, \Light^f_l, \Medium^f_l\}$ by calls to \cref{alg:heavylight}, where $\Medium^f_l$ (think \emph{medium}) is notation of convenience to denote light cells with heavy parents. Note that $D_l$ defined earlier is precisely the set of data points which happen to lie in cells in $\Medium^f_l$ in level $l$. }
    
    \paragraph{Step 2 - Candidate center allocation and second interaction}{ We start by defining $M$, the number of LSH scales, and $R$ the number of independent repetitions of the hashing subroutine to boost the success probability. We then define a mapping $\Lambda_l^f : \R^d \to \R^{\Heavy^f_l} \times \R^d$ where the co-domain is a synthetic space mimicking the union of all heavy cells upon which we can define our LSH functions. Concretely, for points $p$ such that $\Ancestors^* (\Cells_l (p))$ is a heavy cell, the image is a 2-tuple of a $\size{\Heavy_{l - 1.5\lg d}}$ length indicator vector indicating which heavy ancestor cell $p$ lies in, and the $p$'s position with respect to the center of its ancestor cell.
    Finally we construct the mapping $H_{l,m,r,f}$ which computes the output of a $(p(1),p(c),r_{l,m},cr_{l,m})$ hash function if the point $p$ lies in $D_l$ (which is true if and only if $\Cells_l (p) \in \Medium_l$) and a null bucket value otherwise (i.e. no participation). The heavy buckets of these hash functions are privately recovered via calls to $\BTG$ and we also recover the sums of all vectors mapping to heavy buckets via a call to $\HSO$, the consequence succinct histogram and sum oracle are $\BH_{l,m,r,f}$ and $\BSO_{l,m,r,f}$ respectively. 
    
    For every guess for $\OPT$ we pass these histograms and oracles to \cref{alg:main;alg:CCA}, which allocates a candidate center $\Pi_l(\hat{b})$ for every heavy bucket $b$ whose count $\hat{n}_b$ crosses a certain threshold $T_l$. The location at which it allocates that actual center is found by querying the oracle for the sum of all points mapping to this bucket to get a value $\BSO^{l,m,r,f} (b)$ and dividing this vector sum by the histogram count $\BH^{l,m,r,f} (b) = \hat{n}_b$; this is an estimate of the average over all points mapping to this bucket. We then project this average to the space $\Lambda_l$ to get the point $\Pi_l (\hat{b})$. \Cref{alg:main;alg:CCA} also allocates a candidate center at the center of every heavy cell. It then returns all centers found to the calling function \cref{alg:main} which stores the centers passed by the call with the guess $k\sqrt{n} \cdot 2^f$ for $\OPT$ in $S_f$. The net bi-criteria solution then is simply $S = \cup_{f \in [F]} S^f$ which it passes to the center recovery algorithm \cref{alg:2RoundRecovery} along with the dimension reduction and random shift mapping $Q$.
    }
    
    The main results of this section are \cref{lem:core1}, which allows us to derive a guarantee and \cref{lem:candidateCount2}, which bounds the total number of candidate centers allocated.
    
    \begin{definition}
        We record some notation that will be convenient to use in the course of our analysis.
        \begin{enumerate}
            \item We denote the data set in the original space $\R^{d'}$ by $D'$ and in the dimension reduced space $\R^d$ (after scaling, projection, and translation) by $D = M(D')$.
            \item We let the optimal clustering cost in the original space be denote $\OPT'$, and the dimension reduced optimal clustering cost be denoted $\OPT$.
            \item We fix an arbitrary optimal solution $S_{\OPT}$ in the dimension reduced space; we will show that our allocation of candidate centers competes well with $S_{\OPT}$. Note that in particular $f_D (S_{\OPT}) = \OPT$.
        \end{enumerate}
    \end{definition}
    
    \begin{lemma}
        \label{lem:costApproximation2}
        With probability $1-\beta$, we have that for every clustering $(D'_1, \dots, D'_k)$ of $D'$,
        \begin{align*}
            \sum_{i \in k} \sum_{p \in D'_i} s\left( p , \frac{\sum_{q\in D'_i} q}{\size{D'_i}} \right) \simeq_{1 + O(\alpha)} \sum_{i \in k} \sum_{p \in M (D'_i) } s\left( Q (p) , \frac{\sum_{q\in D'_i} Q (q)}{\size{D'_i}} \right).
        \end{align*}
        Further, with this notation $\OPT = \Theta(\OPT')$ and $D = M(D') \subset [0,1)^d$.
    \end{lemma}
    
    \begin{proof}
        We write $Q = P \circ S\circ T$, where $T$ is the dimension reduction to $O(\log (k/\alpha\beta)/\alpha^2)$, $S$ is the scaling by a factor of $1/2 (1 + \alpha)$, and $P$ is projection to the unit ball. Given any clustering $(D_1, \dots, D_k)$ of $D$, by \cref{lem:Makarychev} we have that
        \begin{align*}
            \sum_{i \in k} \sum_{p \in D_i} s\left( p , \frac{\sum_{q\in D_i} q}{\size{D_i}} \right) \simeq_{1 + \alpha} \sum_{i \in k} \sum_{p \in D_i} s\left( T(p) , \frac{\sum_{q\in D_i} T(q)}{\size{D_i}} \right).
        \end{align*}
        The scaling map changes all $\ell_2$-distances by precisely the scaling factor, so we also have that
        \begin{align*}
            \sum_{i \in k} \sum_{p \in D_i} s\left( T(p) , \frac{\sum_{q\in D_i} T(q)}{\size{D_i}} \right) = \frac{1}{(1 + \alpha)^2} \sum_{i \in k} \sum_{p \in S \circ T D_i} s\left( S \circ T (p) , \frac{\sum_{q\in D_i} S \circ T(q)}{\size{D_i}} \right).
        \end{align*}
        Since with probability $1-\beta$ all points lie in $B(0,1/2)$ after scaling by a factor of $1/2(1+\alpha)$, the projection map does not move any point and hence the same clustering cost is preserved. Finally, translating all points by the same offset $\gamma$ makes no difference to the clustering cost. The fact that $\OPT = \Theta(\OPT')$ is a direct consequence of the equality between clustering costs (up to small multiplicative approximation) derived above.
    \end{proof}
    
    \begin{definition}
        Fixing any level $l$, we make some definitions to aid our cost analysis.
        \begin{enumerate}
            \item Let $D^\dagger_l := \left\{ p\in D_l: z(p, S_{\OPT}) \leq d t_l^2 \right\}$. We make this definition because for every $p \in D_l$, $o(\Cells_{l-1} (p)) \in S_{\Heavy}$ and $z(p, o(\Cells_{l-1} (p))) \leq d t_l^2$, so $D^{\dagger}_l$ is the set of points that remains to be covered competitively with $S_{\OPT}$ by allocating candidate cluster centers via LSH.
            \item For $s\in S_{\OPT}$, let $D^\dagger_l (s) = \{ p \in D^\dagger_l : \argmin_{s' \in S_{\OPT}} z(p,s') = s \}$.
            \item From \cref{rem:closeCellBound}, we see that there are at most $O(kL/\beta)$ non-empty intersections of cells $C$ with clusters $D^\dagger_l (s)$. For every such cell $C$, we call $D_{l}^* (s) \cap C$ a \emph{cluster section}.
            \item Let $s_A$ be the optimal center for a cluster section $A$. We can partition $A$ depending on what distance any point in it lies from $s_A$ via geometrically increasing thresholds $\left\{ \frac{t_l}{ d \sqrt{L}}, \frac{2t_l}{ d \sqrt{L}}, \frac{2^{2}t_l}{2^l d \sqrt{L}}, \dots, \sqrt{d} t_l  \right\}$. For ease of notation we denote these thresholds $r_{l,1}, \dots, r_{l,M}$ and set $r_{l,0} = 0$. By definition, there are $M = \log_{2} (d^{3/2} \sqrt{L}) = O(\log \log n)$ many such thresholds. With this notation we define the geometrically thresholded partitions of $A$ as $A_j := \left\{ p\in A: \norm{p - s_A} \in \left[ r_{l,j} , r_{l,j+1} \right)  \right\}$ for $j = 0, \dots M$. Further, let $A_0^m$ denote the partial union $\bigcup_{j=0}^{m} A_j$ for $m=0,\dots, M$.
        \end{enumerate}
    \end{definition}
	
	\begin{lemma}
	    With probability $1-\beta$, for all $l\in [L]$ we have the following bound on the number of cluster sections.
	    \begin{align*}
	         \sum_{s \in S_{\OPT}} \size{\{C \in \Ancestors^* (\Heavy_{l}) : C \cap D^\dagger_l(s) \not= \emptyset \}} = O(kL/\beta)
	    \end{align*}
	\end{lemma}
	\begin{proof}
    	From \cref{lem:bflsy}, we know that 
    	$$\Ex \left[\size{\{C \in \Ancestors^* (\Heavy_{l}) : C \cap D^\dagger_l (s) \not= \emptyset \}}\right] = O(1).$$
    	By linearity of expectation and Markov's inequality (in that order), it follows that with probability $1-\beta/L$, 
    	$$ \sum_{s \in S_{\OPT}} \size{\{C \in \Ancestors^* (\Heavy_{l}) : C \cap D^\dagger_l(s) \not= \emptyset \}} = O(kL/\beta).$$
	\end{proof}
	
	To catch cluster sections which have some $O\left(\frac{\beta \OPT}{t_l^2 kL^2 d}\right)$-many points, we use LSH functions applied on a union of heavy cells, where we modify the norm so that the distance between two different cells is always $> 2c \cdot \sqrt{d} \cdot (d^{3/2} t_l)$, i.e.  $2c \cdot \Diam (C)$ units where $C$ is any cell in the ancestor level. The diameter of this entire space is still $O( \Diam (C))$ and since the smallest distance at which we need to allocate cluster centers to serve cluster sections is $\frac{t_l}{d \sqrt{L}}$, the ratio of the distance to the farthest false positive to the ratio of the distance of the closest clustering distance is $O(\poly\log n)$. 
	
	As discussed in the beginning of the section, this allows us to use LSH functions with reasonable parameters and not end up allocating too many candidate centers. The average of heavy buckets corresponding to sufficiently large cluster sections lie within a distance of $c$ times the threshold at which we are competing with $S_{\OPT}$. Then, by the triangle inequality it will follow that since any two cells are at a distance of strictly greater than $2c \Diam(C)$, this average will be closer to the cell in which the cluster section lies than to all other cells. Since a cell is a convex set, we can project to this cell and be assured that the distance to the cluster section the candidate center is meant to cover is not any greater and therefore that we have allocated a candidate center at the desired distance.
	
	\begin{remark}
	    For the rest of this subsection, we analyse the call to \cref{alg:main;alg:CCA} with the ``correct" value of $f$, i.e. the call such that $k\sqrt{n} 2^{f-1} \leq \OPT < k \sqrt{n} 2^f$. Since $k\sqrt{n} 2^f = \Theta (\OPT)$, we will simply refer to $k\sqrt{n} 2^f$ as $\OPT$. We will need to scale all failure probabilities by $F$ to ensure that the guarantees hold simultaenously for all calls to \cref{alg:main;alg:CCA} with probability $1-\beta$.
	\end{remark}
	
	\begin{definition}[Synthetic space for LSH functions]
	    \begin{enumerate}
	        \item Let $\lambda_l$ denote $(14c + 5) r_{l,M} = (14c + 5) t_l \sqrt{d}$.
	        \item Let $\Lambda_l = \R^{\Ancestors^* (\Heavy_{l})} \times \R^d$. We define a mapping $\Lambda_l : [0,1)^d \to \Lambda_l$ as follows;
	        \begin{align*}
	            \Lambda_l(p) = \begin{cases}
	                            (\lambda_l 1_{\Ancestors^* (C_l(p))}, p - o(\Cells_l (p))) &\mbox{ if }p \in D_l \\
	                            0 &\mbox{otherwise}
	                            \end{cases}
	        \end{align*}
	         We note that for $p \in D_l$, $\Lambda_l (p)$ is a two-tuple consisting of scaled indicator vector and a copy of the cell itself translated so that the center of the cell lies at the origin. In words, if $p \in D_l$ then the indicator vector indicates which cell in the ancestor level a point lies in. Since this space as defined lies in $\R^{\Ancestors^*(\Heavy_l)} \times R^d$, it inherits the $\ell_2$ norm in the canonical way which we denote $\norm{\cdot }_\Lambda$. 
	         \item We have the projection maps to the factor spaces $p_1 : \Lambda_l \to \R^{\Ancestors^* (\Heavy_l)}$ and $p_2 : \Lambda_l \to \R^d$. Since $\Lambda_l = \R^{\Ancestors^* (\Heavy_l)} \times \R^d$ is the direct sum of the vector subspaces $\R^{\Ancestors^* (\Heavy_l)}$ and $\R^d$ we have that $\norm{\cdot}^2_{\Lambda_l} = \norm{p_1(\cdot) }^2 + \norm{p_2 (\cdot)}^2$.
	    \end{enumerate}
	\end{definition}
	
	\begin{lemma}
	    The following statements hold.
	    \begin{enumerate}
	        \item For any $p,q \in [0,1)^d$ if it is the case that $\Ancestors^*(\Cells_l(p))\not= \Ancestors^*(\Cells_l (q))$, then $\lVert p - q \rVert_{\Lambda} > \lambda_l$.
	        \item The diameter of the set of points $\Lambda_l (D_l)$ is $O(\lambda_l)$.
        \end{enumerate}
	\end{lemma}
	
	\begin{proof}
	    \begin{enumerate}
	        \item By the properties of $\norm{\cdot}_{\Lambda_l}$, we have that $\norm{p-q}_{\Lambda} \geq \norm{p_1(p-q) }$. Since $\Ancestors^*(\Cells_l(p))\not= \Ancestors^*(\Cells_l (q))$, $\norm{p_1(p - q)} \geq \lambda_l \sqrt{2}$ ($\norm{p_1(p - q)}$ being the difference of two different basis vectors in $\R^{\Ancestors^*(\Cells_l)}$ scaled by $\lambda_l$. The result follows directly.
	        \item The bound follows by appealing to the properties of $\norm{\cdot}_{\Lambda}$.
	        \begin{align*}
	            \norm{p-q}_{\Lambda} &= \sqrt{\norm{p_1(p-q) }^2 + \norm{p_2 (p-1)}^2} \\
    	            &\leq \norm{p_1 (p-q)} + \norm{p_2 (p-q)} \\
    	            &\leq \lambda_l \sqrt{2} + d^{3/2} t_l \\
    	            &= O(\lambda_l).
	        \end{align*}
	    \end{enumerate}
	\end{proof}
	
	\begin{lemma}
	    \label{lem:LSHpars}
	    For $c > \sqrt{2}$, there is a choice of LSH parameters such that
	    \begin{align*}
	        \frac{p^2(1)}{p(c)} &\geq \frac{k \poly\log n}{\epsilon_{\CH} \beta^2}\\
	        p(1) &\geq \tilde{\Omega} \left( (k \poly\log n)^{-1/c'} \right), \\
	        \log N_B &= \tilde{O} ( \poly \log n/(\epsilon_{\CH} \beta) ),
	    \end{align*}
	    where the $\tilde{O}$ and $\tilde{\Omega}$ notation suppress $O(\log\log n)$ terms and $c' = c^2/8 - 1/4$. It will be convenient to write $1/c' = O(1/(2c^2-1))$.
	\end{lemma}
    
    \begin{proof}
        This result is a direct corollary of \cref{lem:LSH}. We bound all occurrences of $\log k$ from above by $\log n$.
    \end{proof}
    
    We state and prove the guarantees of the bucket histogram as derived from the $\BTG$ guarantee.
    
    \begin{lemma}
	    \label{lem:BH}
	    For every $l\in [L]$, $m\in [M]$, $r \in [R]$, and $f\in [F]$ with probability $1-\beta/(LMRF)$ in every call to \cref{alg:main;alg:CCA},
	    \begin{align*}
	        \BH^{l,m,r,f}_E &= O\left(\frac{1}{\epsilon_{\BH}} \sqrt{n \log n/\beta} \right) \\
	        \BH^{l,m,r,f}_{M} &= O\left(\frac{1}{\epsilon_{\BH}} \sqrt{n \poly\log n/\beta} \right)
	    \end{align*}
	    As these bounds are invariant in $l$, $m$, $r$, and $f$ we will find it convenient to drop these indices without loss. Since $\BH_{M} = \Omega(\BH_E)$, we will simply use a uniform bound of $\BH_M$ for the estimation error of any frequency query.
	\end{lemma}
	
	\begin{proof}
	    Fix any $l\in [L]$, $m \in [M]$, $r\in [R]$, and $f\in[F]$. We see that $\BH_{l,m,r,f}$ is derived from a call to $\BTG$ with the mapping $h_{l,m,r} = p \mapsto H_{\Ancestors (C_l(p),m,r)} (p)$ and failure probability $\beta/(LMR)$. If the size of the co-domain is the number of buckets $N_B$ for this LSH function, then from \cref{lem:btg} we have that 
	    $$\BH^{l,m,r,f}_{M} = O\left(\frac{1}{\epsilon_{\BH}} \sqrt{n \log (N_B \cdot LMRF /\beta) \log (LMRF/\beta)} \right).$$ 
	    From \cref{lem:LSHpars} we have that $\log N_B = \tilde{O} ( \poly \log n/(\epsilon_{\CH} \beta) )$. Note that since $L,M,F = O(\log n)$, $R = O\left( \frac{\log(kL^2/\beta)}{p(1)} \right)$ and $p(1) = \tilde{\Omega} \left( (k \poly\log n)^{-O(1/(2c^2-1))} \right)$, it follows that $\log (LMRF/\beta) = O(\log n/\beta)$. Substituting, we get the stated bound. The expression for $\BH_E^{l,m,r,f}$ follows similarly.
	\end{proof}
	
	We see that $\BH^{l,m,r,f}_E = O(\BH^{l,m,r,f}_M)$, which we use throughout the remainder of the proof.
	
	\begin{remark}
	    We observe that by the union bound, \cref{lem:BH} implies that with probability $1-\beta$, for all levels $l\in [L]$, thresholds $m\in [M]$, repetitions $r \in [R]$, and $\OPT$ guess parameter $f\in[F]$, the frequency query estimation error bound $\BH_M$ holds.
	\end{remark}
    
    \begin{lemma}\label{lem:core1}
            Let $A_0^m$ be a partial union for some cluster section $A \subset D^\dagger_l (s^*) \cap C$ such that $|A_0^m| \geq \max\left( \frac{\beta \OPT}{t_l^2 kL^2 d}, \frac{2 \BH_M}{p(1)}, O\left( \frac{c_G \sqrt{n \poly\log n/\beta}}{\epsilon_{\BSO}}  \right)\right)$. With probability $1 - \frac{\beta}{kL^2 MF}$ there is a point $\Pi_l (\hat{p}_m) \in S_l$ such that for every point $p\in A_0^m$, $\norm{p - \hat{p}_m} = O(c r_{l,m})$.
    \end{lemma}
    
    \begin{proof}
        We observe that if $r$ is the diameter of $A_0^m$ in $\Lambda_l$ then $r \leq r_{l,m+1}$ (as all points lie inside the same ancestor cell, the distance between them does not increase in the space $\Lambda_l$). \Cref{lem:LSHGuarantee} gives us that for any fixed arbitrary point $p_m \in A_0^m$, if the average of all points that collide with $p_m$ under a $(p(1),p(c),2r_{l,m+1},2r_{l,m+1} c)$-sensitive hash function is denoted $\bar{p}_m$ then with probability $p(1)/4$,
	    \begin{align*}
	        \norm{ p_m - \bar{p}_m}_{\Lambda_l} &\leq 2c r_{l,m+1} + \frac{8p(c) |D_l|}{p^2(1) |A_0^m|} \Delta'.
	    \end{align*}
	    Since $\size{D_l} = O(\OPT d^2/t_l^2)$, $\Delta'$ is $\Diam(\Lambda_l) = O(c t_l \sqrt{d})$, and $\size{A_0^m} \geq \frac{\beta \OPT}{t_l^2 kL^2 d}$, we get
	    \begin{align*}
	        \norm{ p_m - \bar{p}_m}_{\Lambda_l} &\leq 2c r_{l,m+1} + \frac{p(c)}{p^2(1)} \cdot \frac{(O(\OPT d^2/t_l^2) + O( \frac{1}{\epsilon_{\CH} \beta} k \sqrt{n} \poly\log n) ) t_l^2 k L^2 d}{\beta \OPT} \cdot c t_l \sqrt{d} \\
	        &\leq 2c r_{l,m+1} + \frac{p(c)}{p^2(1)} \left( \frac{O(\OPT d^2/t_l^2) \cdot t_l^2 k L^2 d}{ \beta \OPT} + \frac{(O( \frac{1}{\epsilon_{\CH} \beta} k \sqrt{n} \poly\log n) ) t_l^2 k L^2 d}{ \beta \OPT} \cdot c t_l \sqrt{d} \right) \\
	        &\leq 2 c r_{l,m+1} + \frac{p(c)}{p^2(1)} \left( \frac{O(c t_l k \poly\log n)}{\beta}   + \frac{O(c t_l k \poly\log n)}{\epsilon_{\CH} \beta^2 }\right)
	    \end{align*}
	    where in the above we use that $\OPT \geq k\sqrt{n}$. Since $r_{l,m+1} \geq \frac{t_l}{ d \sqrt{L}}$ if we choose the LSH parameter such that 
	    $$\frac{p(c)}{p^2(1)} \leq \frac{\epsilon_{\CH} \beta^2}{k \poly\log n}$$
	    then $\norm{p_m- \bar{p}_m}_{\Lambda_l} \leq 3c r_{l,m+1}$. Note that the bound on this ratio does not vary with threshold $r_{l,m}$ or level $l$. This explains the uniform choice of LSH parameters used in \cref{alg:main}. \Cref{lem:LSHpars} bounds from below the probability $p(1)$ and from above the number of buckets $N_B$ for this choice of LSH parameter.
	    
	    For any successful run, since $\norm{p_m- \bar{p}_m}_{\Lambda_l} \leq 3c r_{l,m+1}$,  by the triangle inequality, for any $p\in A_{m}$
	    \begin{align*}
	        \norm{p - \bar{p}_m}_{\Lambda_l} &\leq \norm{p - p_m}_{\Lambda_l} + \norm{ p_m - \bar{p}_m}_{\Lambda_l}\\
	        &\leq r_{l,m+1} + 3c r_{l,m+1} \\
	        &\leq (3c + 1) (2 r_{l,m}).
	    \end{align*}
	    
	    Since the success probability $p(1)/4$ does not depend upon the level or threshold, a uniform number of $R = O(\log (k L^2 F/\beta)/p(1)) = O((\frac{k \log n}{\beta})^{O(1/(2c^2-1))} \log n)$ many independent repetitions of this LSH scheme boost the success probability from $\frac{p(1)}{4}$ to $1-\frac{\beta}{kL^2 M F}$. \Cref{lem:LSHGuarantee} also guarantees that in the successful LSH run at least $\frac{p(1)}{2} \cdot \size{A_0^m}$ many points in $A_0^m$ will collide with $p_m$. If $\frac{p(1)}{2} \cdot \size{A_0^m} \geq \BH_M$, then $h_{l,m,r} (p_m) \in \BH_{l,m,r}$ (\cref{lem:BH}), where $
	    \BH_M = O\left(\frac{1}{\epsilon_{\BH}} \sqrt{n \poly \log n/\beta} \right)$. 
	    For every $H_{l,m,r,f} (p_m) \in \BH_{l,m,r,f}$, \cref{alg:main} computes an estimate for $\bar{p}_m$ which we denote 
	    $$\hat{p}_m := \frac{\BSO_{l,m,r} (p_m)}{\BH_{l,m,r,f} (p_m)}.$$ 
	    By \cref{lem:noisyAvgGuarantee}, if $h_{l,m,r} (p_m) \in \BH_{l,m,r,f}$, then the estimation error on querying $\BSO (h_{l,m,r} (p_m))$ obeys the bound
	    \begin{align*}
	        \norm{\frac{\sum_{p:H_{l,m,r,f}(p) = H_{l,m,r,f}(p_m)} p - o(C)}{|\{p:H_{l,m,r,f}(p) = H_{l,m,r,f} (p_m) \}|} - \frac{\BSO_{l,m,r,f} (p_m)}{\BH_{l,m,r,f} (p_m)}}_{\Lambda_l}  &\leq  \frac{1}{\Omega(|A_0^m|) } \cdot O\left( \frac{c_G \Diam(\Lambda_l)}{\epsilon_{\BSO}} \sqrt{n \log^2 n/\beta }\right) \\
	        &\leq O\left( \frac{c_G \lambda_l \sqrt{n \log^2 n/\beta} }{|A_0^m| \epsilon_{\BSO}}  \right).
	    \end{align*}
	    In the above we used that $d= O(\log n)$ to substitute for it in the estimation error. It follows that if $\size{A_0^m} \geq \frac{c_G c_{\BSO} \Diam(\Lambda_l) \sqrt{n \log^2 n/\beta} }{c r_{l,m} \epsilon_{\BSO}} $ for some universal constant $c_{\BSO}$ derived from the $\HSO$ guarantee then the additional estimation error in $\norm{\cdot}_{\Lambda}$ norm incurred is $c r_{l,m}$. Substituting, we get that it would suffice to have
	    \begin{align*}
	        \size{A_0^m} &\geq \Omega\left( \frac{c_G \lambda_l \sqrt{n \log^2 n/\beta} }{c r_{l,m} \epsilon_{\BSO}}\right) \\
	        \Leftrightarrow \size{A_0^m} &\geq \Omega\left( \frac{c_G \cdot c d^2 t_l \cdot \sqrt{n \log^2 n/\beta} }{c r_{l,m} \epsilon_{\BSO}}\right) \\
	        \Leftarrow \size{A_0^m} &\geq \Omega\left( \frac{2 c_G \sqrt{n d^6 L \log^2 n/\beta} }{\epsilon_{\BSO}}\right)
	    \end{align*}
	    where in the above we lower bound $r_{l,m}$ by $r_{l,1}$. 
	    
	    So in sum, we have that for any cluster union $A_0^m$ such that $\size{A_0^m}\geq \frac{ c_G \sqrt{n d^6 L \log^2 n/\beta} }{\epsilon_{\BSO}}$, for some fixed arbitrary point $p_m \in A_0^{m}$, with probability $1-\beta/(kL^2F)$ there exists a hash function $H_{l,m,r,f}$ for some $r\in [R]$ such that the estimate of the average $\hat{p}_m$ over the bucket that $p_m$ maps to lies within a distance of $cr_{l,m}$ units of $\bar{p}_m$, the true average over the heavy bucket, which lies within a distance of $(6c + 2)r_{l,m}$ of the point $p_m$, and by the triangle inequality $\norm{p_m - \hat{p}_m} \leq (7c+2) r_{l,m}$. 
	    
	    Now since the distance between any two different cells in the space $\Lambda_l$ is strictly greater than $(14c+5) r_{l,m} > 2\norm{p_m - \hat{p}_m}$, it follows from the triangle inequality that $\Pi_l (\hat{p}_m)$ the projection of $\hat{p}_m$ onto $\cup_{C \in \Heavy_l} C$ lies in the cell $C$. Indeed, it was to ensure this guarantee that we chose our value of $\lambda_l$. Since $\Pi_l (\hat{p}_m)$ is a projection onto a convex set, $\norm{\Pi_l (\hat{p}_m) - \hat{p}_m} \leq \norm{p_m - \hat{p}_m}$. Now since the diameter of $A_0^m$ is $r_{l,m+1} = 2 r_{l,m}$, it follows that every point in $A_0^m$ lies within a distance of $O(cr_{l,m})$ units of $\Pi_l (\hat{p}_m)$.
    \end{proof}
    
    \begin{lemma}
    \label{lem:candidateCount1}
        We have the following bound on $S_l$, the number of candidate centers allocated per level in \cref{alg:main;alg:CCA}.
        \begin{align*}
            |S_l| = O\left(\frac{k \poly\log n}{\beta }\right)^{1 + O(1/(2c^2-1))}
        \end{align*}
    \end{lemma}
    
    \begin{proof}
        A candidate center is allocated for every $b\in\BH^{l,m,r}$ such that $\BH^{l,m,r}(b) \geq T_l - \BH_M$ where
        \begin{align*}
            T_l = \frac{p(1)}{2} \cdot \max\left( \frac{\beta \OPT}{t_l^2 kL^2 d}, \frac{4 \BH_M}{p(1)}, O\left( \frac{c_G \sqrt{n \poly\log n/\beta}}{\epsilon_{\BSO}}  \right)\right).
        \end{align*}
        The set of buckets which are identified as having these many points is at most the set of buckets which have $T_l - 2 \BH_M$ many points in them. Therefore we want to bound from above the quantity $\size{D_l}/(T_l - 2\BH_M)$. Since $T_l \geq 4 \BH_M$, we can write
        \begin{align*}
            \frac{\size{D_l}}{T_l - 2\BH_M} &\leq \frac{2 \size{D_l}}{T_l}\\
            &\leq \frac{ O(d^2 \OPT/t_l^2) + O\left( \frac{kL \CH_M}{\beta} \right) }{\max \left\{ \frac{ p(1)}{2} \frac{\beta \OPT}{t_l^2 kL^2 d} , 2\BH_M \right\} } \\
            &\leq \frac{O(d^2 \OPT / t_l^2)}{\frac{ p(1)}{2} \frac{\beta \OPT}{t_l^2 kL^2 d}}  + O\left( \frac{kL \CH_M}{\beta} \cdot \frac{1}{\BH_M} \right) \\
            &\leq O\left(\frac{2}{p(1)} \cdot \frac{k L^2 d^3}{\beta} \right) + O\left( \frac{kL \CH_M}{\beta \BH_M} \right).
        \end{align*}
        We use that $\OPT \geq k\sqrt{n}$, $\CH_M = O\left(\frac{1}{\epsilon_{\CH}}\sqrt{n \poly \log n/\beta}\right)$, $\BH_M = \frac{1}{\epsilon_{\BH}} \sqrt{n \poly\log n}$ and that $$\frac{1}{k^{O(1/(2c^2-1))}\poly\log n} \leq \frac{\epsilon_{\CH}}{\epsilon_{\BH}} \leq k^{O(1/(2c^2-1))} \poly\log n$$
        to get
        \begin{align*}
            \frac{\size{D_l}}{T_l - 2\BH_M} &\leq O\left(\frac{2}{p(1)} \cdot \frac{kL^2 d^3}{\beta}\right) + \frac{kL \CH_M}{\beta \BH_M}  \\
            &= O\left(\frac{k \poly\log n}{\beta }\right)^{1 + O(1/(2c^2-1))} + \frac{k \poly\log n}{\beta} \\
            &= O\left(\frac{k \poly\log n}{\beta }\right)^{1 + O(1/(2c^2-1))}.
        \end{align*}
        Taking the union over all possible values of $(l,m,r)$ and absorbing the addition $\log$ factors in the $\poly\log$ term, we get the stated bound. The fact that the bounds on the ratio of $\epsilon_{\CH}$ to $\epsilon_{\BH}$ is adhered to can be checked in the proof of the main theorem at the end of this section.
    \end{proof}
    
    \begin{lemma}
    \label{lem:candidateCount2}
        In \cref{alg:main}, the following bound holds for the total number of candidate centers allocated.
            \begin{align*}
                \size{S} =  \left(\frac{k \poly\log n}{\beta }\right)^{1 + O(1/(2c^2-1))}
            \end{align*}
    \end{lemma}
    
    \begin{proof}
        We observe that $S$ in \cref{alg:main} equals $\cup_{f \in O(\log n)} S_f$ is the union of $F = O(\log n)$ many sets of candidate centers returned by calls to \cref{alg:main;alg:CCA}. It therefore suffices to bound the set of candidate centers 
        $$\size{S_{\Heavy} \cup \bigcup_{l\in[L]} S_l}$$
        returned by \cref{alg:main;alg:CCA}. From \cref{lem:candidateCount1}, by adding the bounds for $S_l$ over $L$ levels, absorbing the factor of $L$ into the $\poly\log n$ term and noting that the  $\size{S_{\Heavy}} = O(kL^2/\beta) \cdot L$ summand is asymptotically dominated by $\size{\cup_{l \in L} S_l}$, we get that
        \begin{align*}
                \size{S^f} =   O\left(\frac{k \poly\log n}{\beta }\right)^{1 + O(1/(2c^2-1))}.
            \end{align*}
        The stated bound now follows simply be absorbing an $O(\log n)$ factor in the $\poly\log$ term.
    \end{proof}
    
    \begin{definition}
	    Let $f^*$ denote the ``correct" call to \cref{alg:main;alg:CCA}, i.e. the unique value of $f \in [F]$ such that $k\sqrt{n} 2^{f-1} \leq \OPT < k\sqrt{n} 2^f$
	\end{definition}
	
	\begin{lemma}
        \label{lem:core2}
        Let $A \subset D^\dagger_l (s) \cap C$ be some cluster section. Then with probability $1-\beta/kL^2$ we have that
        \begin{align*}
            f_{A} (S^{f^*}) = O( f_{A} (S_{\OPT}) + \max\left( \frac{\beta \OPT}{t_l^2 k L^2 d}, \frac{4 \BH_M}{p(1)}, O\left( \frac{c_G \sqrt{n \poly\log n/\beta}}{\epsilon_{\BSO}}  \right)\right) \cdot d t_l^2 + \frac{c^2 t_l^2}{ d^2 L} \cdot \size{A}.
        \end{align*}
    \end{lemma}
    
    \begin{proof}
        Let $m'$ be the largest index such that $|A_0^{m'}| < \max\left( \frac{\beta \OPT}{t_l^2 kL^2 d}, \frac{4 \BH_M}{p(1)}, O\left( \frac{c_G \sqrt{n \poly\log n/\beta}}{\epsilon_{\BSO}}  \right)\right)$. We can write
        \begin{align*}
            f_{A} (S^{f^*}) &\leq f_{A_0^{m'}} (S^{f^*}) + \sum_{m = m'+1}^M f_{A_m} (S^{f^*}).
        \end{align*}
        For every $m>m'$, by \cref{lem:core1} we know there is a candidate center $\Pi_l (\hat{p}_m) \in S^{f^*}$ such that for every point $p\in A_0^m$, $\norm{p - \Pi_l (\hat{p}_m)} = O(c r_{l,m})$. Since for all $p \in A_m \subset A_0^m$, $\norm{ p - s} \geq r_{l,m}$, it follows that with probability $1 - \beta/kL^2M$, $f_{A_m} (S) \leq O( f_{A_m} (S_{\OPT}) + c^2 r_{l,1}^2$. By the union bound, we have that this guarantee holds for all thresholds with probability $1-\beta/kL^2$. Substituting this bound, we get
        \begin{align*}
            f_{A} (S^{f^*}) &\leq \max\left( \frac{\beta \OPT}{t_l^2 kL^2 d }, \frac{4 \BH_m}{p(1)}, O\left( \frac{c_G \sqrt{n \poly\log n/\beta}}{\epsilon_{\BSO}}  \right)\right)\cdot d t_l^2 + \sum_{m = m'+1}^M \left[ O(c f_{A_m} (S_{\OPT}))  + c^2 r_{l,1}^2 \right] \\
            &\leq \max\left( \frac{\beta \OPT}{t_l^2 kL^2 d }, \frac{4 \BH_M}{p(1)}, O\left( \frac{c_G \sqrt{n \poly\log n/\beta}}{\epsilon_{\BSO}}  \right)\right)\cdot dt_l^2 + O( f_{A} (S_{\OPT}) + \frac{c^2 t_l^2 }{d^2 L} \cdot \size{A}.
        \end{align*}
    \end{proof}
	
	\begin{lemma}
        \label{lem:core3}
        The following bound holds.
        \begin{align*}
            f_{D_l} (S^{f^*}) &= O( f_{D_l} (S_{\OPT})) + O( \OPT/L) + O( \CH_M/(d^2 \beta) ) \\
            &+ O\left( \frac{c_G}{\epsilon_{\BSO} \beta} (k \poly\log n)^{1 + O(1/(2c^2-1))} \sqrt{n} \right)
        \end{align*}
    \end{lemma}
    \newcommand{\A}{\mathcal{A}}
    \begin{proof}
        Let $\A$ be the set of all cluster sections. We know that $\size{\A} = O(kL/\beta)$ and that $D^\dagger_l = \sqcup_{A \in \A} A$. We can write
        \begin{align*}
            f_{D^\dagger_l} (S^{f^*}) &= \sum_{A \in \A} f_{\A} (S^{f^*}) \\
            &= \sum_{A \in \A}  \left[ O( f_{A} (S_{\OPT}) + \frac{c^2 t_l^2 }{d^2 L} \cdot \size{A} + \max\left( \frac{\beta \OPT}{t_l^2 kL^2 d}, \frac{4 \BH_M}{p(1)}, O\left( \frac{c_G \sqrt{n \poly\log n/\beta}}{\epsilon_{\BSO}}  \right)\right)\cdot d t_l^2 \right]\\
            &= O( f_{D^\dagger_l} (S_{\OPT})) + \frac{c^2 t_l^2 }{d^2 L} \cdot \size{D^\dagger_l} + \max\left( O\left( \frac{\OPT}{L} \right), \frac{4 \BH_M}{p(1)} O\left(\frac{k L d}{\beta}\right) , O\left( \frac{c_G k \sqrt{n \poly\log n/\beta}}{\epsilon_{\BSO} \beta}  \right)\right)
        \end{align*}
        where in the above we absorb a factor of $d$ in the $\poly\log n$ expression. To bound the second term, we recall that $\size{D_l} \leq O(d^2 \OPT/t_l^2) + O(kL \CH_M/\beta)$. Substituting, we get
        \begin{align*}
            \frac{c^2 t_l^2 }{d^2 L} \cdot \size{D^\dagger_l}  &\leq \frac{c^2 t_l^2 }{d^2 L} \cdot (O(d^2 \OPT/t_l^2) + O(kL \CH_M/\beta)) \\
            &\leq O( \OPT/L) + O( k \CH_M/(d^2 \beta) )
        \end{align*}
        We now simplify the last term in the upper bound for $f_{D^\dagger_l}$ by noting that for any call to \cref{alg:main}, the parameter $c$ is a constant, that $d, L = \log n$, and that since we can let $\epsilon_{\BSO} = \epsilon_{\BH}$ (as they are called exactly the same number of times). In sum
        \begin{align*}
            &\max\left( O\left( \frac{\OPT}{L} \right), \frac{4 \BH_M}{p(1)} O\left(\frac{k L d}{\beta}\right) , O\left( \frac{c_G k \sqrt{n \poly\log n/\beta}}{\epsilon_{\BSO} \beta}  \right)\right) \\
            &\leq  O\left( \frac{c_G}{\epsilon_{\BSO} \beta} (k \poly\log n)^{1 + O(1/(2c^2-1))} \sqrt{n} \right).
        \end{align*}
        We recall that $D^\dagger_l = \{p \in D_l : z(p, S_{\OPT}) < d t_l^2 \}$. It follows that if $p \in D_l \backslash D^\dagger_l$ then $z(p,S_{\OPT}) > d t_l^2$, in which case $z(p, S_{\Heavy}) < z(p,S_{\OPT})$. In sum, $f_{D_l} (S^{f^*}) \leq f_{D^\dagger_l}(S^{f^*}) + f_{D_l} (S_{\OPT})$, from which the stated bound follows directly.
    \end{proof}
    
    \begin{lemma}
            \label{lem:core4}
            The following bound holds.
            \begin{align*}
                f_D(S) = O(\OPT) + O\left( \left(\frac{c_G}{\epsilon_{\BSO} \beta} + \frac{1}{\epsilon_{\CH}} \right) (k \poly\log n)^{1 + O(1/(2c^2-1))} \sqrt{n} \right).
            \end{align*}
    \end{lemma}
    
    \begin{proof}
        \begin{align*}
            f_D (S^{f^*}) &= \sum_{l \in [L]} f_{D_l} (S^{f^*}) \\
            &= \sum_{l \in [L]} \bigg[ O( f_{D_l} (S_{\OPT})) + O( \OPT/L) + O( k \CH_M/(d^2 \beta) )  \\
            &+ O\left( \frac{c_G}{\epsilon_{\BSO} \beta} (k \poly\log n)^{1 + O(1/(2c^2-1))} \sqrt{n} \right) \bigg] \\
            &= O( \OPT) + O( k \CH_M/(d^2 \beta) ) + O\left( \frac{c_G}{\epsilon_{\BSO} \beta} (k \poly\log n)^{1 + O(1/(2c^2-1))} \sqrt{n} \right)
        \end{align*}
        where in the above we use that $L = \log n$ to absorb a factor of $L$ in the $\poly\log n$ expression. Now since $S^{f^*} \subset S$, we can write
        \begin{align*}
                O( \OPT) + O( k \CH_M/(d^2 \beta) ) + O\left( \frac{c_G}{\epsilon_{\BSO} \beta} (k \poly\log n)^{1 + O(1/(2c^2-1))} \sqrt{n} \right)
        \end{align*}
        We now simplify this expression by opening up the expression for $\CH_M$, and by absorbing the $c^2$ term in the big-Oh notation.
        \begin{align*}
            f_D(S) = O(\OPT) + O\left( \left(\frac{c_G}{\epsilon_{\BSO} \beta} + \frac{1}{\epsilon_{\CH}} \right) (k \poly\log n)^{1 + O(1/(2c^2-1))} \sqrt{n} \right).
        \end{align*}
    \end{proof}

    \subsection{Cost analysis}\label{subsec:cost}
	
	\begin{algorithm}
	    \caption{2-Round Center Recovery}
	    \label{alg:2RoundRecovery}
	    \KwData{Bicriteria $k$-means relaxation $S$ for $k$-means clustering under dimension reducing transformation $M$, the tranformation $M : \R^{d'} \to \R^d$}
	    \SetKwProg{DoParallelFor}{Do in parallel for}{:}{end}
        \SetKwBlock{DoParallel}{Do in parallel:}{end}
	    $s(p) := p \mapsto \argmin_{s \in S} \lVert p - s \rVert_2$\\
	    $\CCH = \BTG (s(\cdot),\beta, \epsilon_{\SH})$ \tcc*{Candidate center histogram}
        $D^* \leftarrow \{ s \in S \mbox{ with multiplicity } \SH(s) \}$\\
        $S^* = \{s_1^*, \dots, s_k^* \} \leftarrow \SKM$ \\
        $s^*(p) := p \mapsto \argmin_{s^* \in S^*} \lVert M(p) - s^* \rVert_2$\\
        \DoParallel{
            Agents reveal $\hat{v}(p)$ for $p\in D'$ where
            \begin{align*}
                v(p)_s &= \begin{cases} p \mbox{ if } s = s^* (p) \\ 0 \mbox{ otherwise }\end{cases}\\
                \hat{v} (p) &= v(p) + N\left(0, \frac{c_{G'}^2 \cdot 2}{\epsilon_{G'}^2} \mathbb{I}_{d'k} \right)
            \end{align*}\\
            $\SH = \BTG (s^*(\cdot),\beta, \epsilon_{\SH})$ \tcc*{Cluster centers histogram}
        }
        $\hat{v} = \sum_{p \in D'} \hat{v}(p)$ \\
        $\hat{s}^* = \sum_{p \in D'} \hat{s}^* (p)$\\
        \For{$j = 1,\dots, k$}
        {
            $\hat{\mu}_j = \frac{\hat{v}_j}{\SH(s^*_j)}$\\
        }
        \KwRet{$S' = \{ \hat{\mu}_1,\dots, \hat{\mu}_k \}$ }
	\end{algorithm}
	
	In this subsection we complete the cost analysis of this algorithm. In the previous section we showed that the candidate centers allocated serve as a good bi-criteria solution for the $k$-means problem with respect to the dimension reduced data set $D$. We will be able to use this in turn to show that proxy data set $D^*$ constructed in \cref{alg:2RoundRecovery} has a similar $k$-means clustering function to that of $D$. This result implies that the $k$ cluster centers derived from non-private clustering of $D^*$ work well as cluster centers for $D$. Finally, we conclude our cost analysis by bounding the cost incurred when clustering the original data set $D'$ with the $k$ centers in $S'$ returned after undoing the dimension reduction.
	
	\begin{definition}[Proxy dataset]
	    \begin{enumerate}
	        \item From \cref{alg:2RoundRecovery} we see that 
        $$D^* = \{ s \in S \mbox{ with multiplicity } \hat{n}_s \mbox{ for all }(s,\hat{n}_s) \in \CCH  \}.$$ 
        We call this the \emph{proxy data set} for $D$.
            \item We let $D(s) = \{ p \in D : \argmin_{s_1 \in S} z(p, s_1) = s \} $.
	    \end{enumerate}
    \end{definition}
	
	\begin{lemma}
	    \label{lem:centerCounts}
	    With probability $1-2\beta$ we have that
	    \begin{enumerate}
	        \item For all $s\in S$ we have $\size{\size{\{p\in D : s(p) = s \}} - \CCH(s)} \leq O(\frac{1}{\epsilon_{\CCH}} \log n/\beta)$.
	        \item For all $s^* \in S^*$ we have that $\size{\size{p \in D' : s^*(M(p)) = s^*}} \leq O(\frac{1}{\epsilon_{\SH}} \log n/\beta)$
	    \end{enumerate}
	\end{lemma}
	\begin{proof}
	    The stated bounds follow from the $\BTG$ guarantee. We use the values $\CCH_M$ and $\SH_M$ as uniform error bounds. Note that the size of the co-domain for $s(\cdot)$ is $\size{S}$ and $\log \size{S} = O(\log n)$. Similarly the size of the co-domain for $s^* (\cdot)$ is $\size{S^*} = k$, so the second bound follows directly as well.
	\end{proof}

	\begin{lemma}
	    \label{lem:fourRoundProxyVersusActual}
	    The $k$-means clustering functions of $D$ and $D^*$ are similar. Concretely, for any finite set $S_1$, the following bounds hold.
	    \begin{align*}
	        f_{D^*} ( S_1) &\leq 2 f_D (S) + 2 f_D (S_1) +  O\left(\frac{\size{S}}{\epsilon_{\CCH}} \sqrt{n} \log n/\beta\right), \\
	        f_D (S_1) &\leq 2 f_D (S) + 2 f_{D^*} (S_1) + O\left(\frac{\size{S}}{\epsilon_{\CCH}} \sqrt{n} \log n/\beta\right).
	    \end{align*}
	    As a direct corollary,
	    \begin{align*}
	        f_{D^*} (S_{\OPT}) &\leq O(\OPT) + O\left( \left(\frac{c_G}{\epsilon_{\BSO} \beta} + \frac{1}{\epsilon_{\CH}} \right) (k \poly\log n)^{1 + O(1/(2c^2-1))} \sqrt{n} \right) \\
	        &+ O\left(\frac{\size{S}}{\epsilon_{\CCH}} \sqrt{n} \log n/\beta\right).
	    \end{align*}
	\end{lemma}
	
	\begin{proof}
	    We can enumerate all points in $D^*$ by counting each candidate center in $s\in S$ a total of $\hat{n}_s$ many times.
	    \begin{align*}
	        f_{D^*} ( S_1) &= \sum_{p^* \in D^*} \min_{s \in S_1} z (p^*,s) \\
	        &= \sum_{s \in S} \hat{n}_s \min_{s'\in S_1} z(s,s') \\
	        &= \sum_{s \in S} \size{D(s)} \min_{s'\in S_1} z(s,s') + \hat{n}_s - \size{D(s)}\\
	        &\leq \sum_{p\in D} z(s(p),\argmin_{s'\in S_1} z(s(p),s') ) + O\left(\frac{\size{S}}{\epsilon_{\CCH}} \sqrt{n} \log n/\beta\right) \\
	        &\leq \sum_{p\in D} z(s(p),\argmin_{s'\in S_1} z(p,s') ) + O\left(\frac{\size{S}}{\epsilon_{\CCH}} \sqrt{n} \log n/\beta\right) \\
	        &\leq \sum_{p\in D} 2z(s(p),p) + 2z(p,\argmin_{s'\in S_1} z(p,s') ) + O\left(\frac{\size{S}}{\epsilon_{\CCH}} \sqrt{n} \log n/\beta\right) \\
	        &\leq 2 f_D (S) + 2 f_D (S_1) + O\left(\frac{\size{S}}{\epsilon_{\CCH}} \sqrt{n} \log n/\beta\right)
	   \end{align*}
	   where we apply the weak triangle inequality for $\ell_2^2$ distance. Proceeding similarly,
	   \begin{align*}
	       f_D (S_1) &= \sum_{p \in D} \min_{s \in S_1} z (p,s) \\
	       &= \sum_{s \in S} \sum_{p \in D(s)} \min_{s \in S_1} z(p,s) \\
	       &\leq \sum_{s \in S} \sum_{p \in D(s)} z(p,\argmin_{s_1 \in S_1} z(s,s_1) ) \\
	       &\leq \sum_{s \in S} \sum_{p \in D(s)} \min_{s \in S_1} 2z(p,s) + 2z(s,\argmin_{s_1 \in S_1} z(s,s_1) ) \\
	       &\leq 2 f_D (S) + \sum_{s \in S} \size{D(s)} 2z(s,\argmin_{s_1 \in S_1} z(s,s_1) ) \\
	       &\leq 2 f_D (S) + \sum_{s \in S} (\hat{n}_s ) 2z(s,\argmin_{s_1 \in S_1} z(s,s_1) ) + (\size{D(s)} - \hat{n}_S)\\
	       &\leq 2 f_D (S) + 2 f_{D^*} (S_1) + O\left(\frac{\size{S}}{\epsilon_{\CCH}} \sqrt{n} \log n/\beta\right)
	   \end{align*}
	   The corollary follows by substituting our upper bound for $f_D (S)$ in its place.
	\end{proof}
	
	\begin{lemma}
	    \label{lem:lowDimGuarantee}
	    If the set $S^*$ is such that
	    \begin{align*}
	        f_{D^*} (S^*) \leq \eta \min_{S_1 : |S_1| = k} f_{D^*} (S_1)
	    \end{align*}
	    for some universal constant $\eta$ (for instance the guarantee of the non-private clustering algorithm) then
	    \begin{align*}
	        f_{D^*} (S^*) &= O(\OPT) + O\left( \left(\frac{c_G}{\epsilon_{\BSO} \beta} + \frac{1}{\epsilon_{\CH}} \right) (k \poly\log n)^{1 + O(1/(2c^2-1))} \sqrt{n} \right) + O\left(\frac{\size{S}}{\epsilon_{\CCH}} \sqrt{n} \log n/\beta\right) \\
	        f_D (S^*) &= O(\OPT) + O\left( \left(\frac{c_G}{\epsilon_{\BSO} \beta} + \frac{1}{\epsilon_{\CH}} \right) (k \poly\log n)^{1 + O(1/(2c^2-1))} \sqrt{n} \right) + O\left(\frac{\size{S}}{\epsilon_{\CCH}} \sqrt{n} \log n/\beta\right).
	    \end{align*}
	\end{lemma}
	
	\begin{proof}
	    The first bound follows from the \cref{lem:fourRoundProxyVersusActual} by noting that $f_{D^*} (S_{\OPT})$ is an upper bound for $\min_{S' : |S'| = k} f_{D^*} (S')$, and by absorbing the universal constant $\eta$ in the big-Oh notation. The second bound follows from the first bound and \cref{lem:fourRoundProxyVersusActual}.
	\end{proof}
    
    We have shown that the $k$-means solution found in the dimension reduced space for the proxy dataset works well for the dimension reduced dataset. Now we use the cluster sets hence derived to privately estimate cluster centers in the original space.

    Given a clustering of $D'$ in the original space by identifying points with the clusters derived from $S^*$ in the dimension reduced space, we know that the $k$-means cost of the clustering is of the same order as the $k$-means cost in the dimension reduced space, as proved in \cref{lem:costApproximation2}. We recover the cluster centers in the original space via noisy averaging. In \cref{alg:2RoundRecovery}, each point holds a $k$-tuple of $d'$-dimensional vector $v(p)$ which we can naturally identify as a $kd'$ dimensional vector. If $s_i^*$ is closest to $p$ in the low-dimensional space (breaking ties arbitrarily), then the $i$th tuple value is $p$ and all other tuples are the zero vector. To preserve privacy, agents release this vector via the Gaussian mechanism.
    
    \begin{lemma}
        \label{lem:costGuarantee}
        \begin{align*}
            f_{D'} (S') &= O(\OPT) + O\left( \left(\frac{c_G}{\epsilon_{\BSO} \beta} + \frac{1}{\epsilon_{\CH}} \right) (k \poly\log n)^{1 + O(1/(2c^2-1))} \sqrt{n} \right) + O\left(\frac{\size{S}}{\epsilon_{\CCH}} \sqrt{n} \log n/\beta\right)\\
            &+ O\left( \left( \frac{c_{G'}}{\epsilon_{G'}} + \frac{1}{\epsilon_{\SH}} \right) k \sqrt{d' n }\log n/\beta \right).
        \end{align*}
    \end{lemma}
    
    \begin{proof}
        For $s \in S^*$ let $D'(s) = \{ p \in D' : M(p) \in D(s) \}$, where we recall that $M$ was the composition of the dimension reduction, scaling, projection and translation maps, and $D(s) = \{ p \in D : \argmin_{s_1 \in S^*} z(p, s_1) = s \} $. Let $\mu_j = \sum_{p \in D(s_j)} p/\size{D(s_j)}$. From \cref{lem:costApproximation2} we have that
        \begin{align*}
            f_{D'} ( \{\mu_1, \dots, \mu_k \}) = O(f_D (S^*)).
        \end{align*}
        In \cref{lem:lowDimGuarantee} we have derived the bound
        \begin{align*}
            f_D(S^*) \leq O(\OPT) + O\left( \left(\frac{c_G}{\epsilon_{\BSO} \beta} + \frac{1}{\epsilon_{\CH}} \right) (k \poly\log n)^{1 + O(1/(2c^2-1))} \sqrt{n} \right) + O\left(\frac{\size{S}}{\epsilon_{\CCH}} \sqrt{n} \log n/\beta\right).
        \end{align*}
        In \cref{alg:2RoundRecovery} we construct estimates $\hat{\mu_j} = \hat{v}_{j}/\SH(s^*_j)$ for the $\mu_j$. We now bound the addition error incurred during this approximation step.

        We see that $\hat{v}_j = \sum_{p \in D'(s^*_j)} p + N\left(0, \frac{2 c_G^2}{\epsilon_G^2} \mathbb{I}_{d'} \right)$. If we denote the random noise added by the agent with data $p$ by $\eta_p$, then we have
        \begin{align*}
            P\left( \bigg\|\sum_{p \in D'(s^*_j)} \eta_p \bigg\| \geq t  \right) \leq \exp \left( \frac{- \epsilon_G^2 t^2}{16 d' n c_G^2} \right).
        \end{align*}
        So there is a choice of  
        $$t = O\left(\frac{c_{G'} \sqrt{d' n \log k/\beta}}{\epsilon_{G'}}\right)$$
        such that $\|\hat{v}_j - \sum_{p \in D'(s^*_j)} p \| \leq t$ with probability $1-\beta/k$. From \cref{lem:centerCounts}, we have that
        \begin{align*}
            \lvert \SH(s^*_j) - D(s^*_j)\rvert \leq O\left( \frac{1}{\epsilon_{\SH}} \sqrt{n } \log n/\beta \right).
        \end{align*}
        It follows that by the union bound that all these bounds hold simultaneously with probability $1-2\beta$. For all clusters $D'(s^*_j)$ which have more than $2\SH_M$ data points we have that $\SH(s^*_j) = \Theta (\size{D'(s^*_j)})$, and for all smaller clusters since the diameter of the data domain is $1$ unit, $f_{D'(s^*_j)} \leq \size{D'(s^*_j)} = O\left(\frac{\size{S}}{\epsilon_{\CCH}} \sqrt{n} \log n/\beta\right)$ unconditionally. Assuming that the former case holds, we get that the error bounds for $\hat{v}_s$ and $\SH(s^*_j)$ give us
        \begin{align*}
            \norm{\hat{\mu}_j - \mu_j} &= \norm{ \frac{\hat{v}_{j}}{\SH(s^*_j)} - \frac{\sum_{p \in D'(s^*_j)} p}{\size{D(s^*_j)}} } \\
            &= \norm{ \frac{\hat{v}_{j}}{\SH(s^*_j)} - \frac{\sum_{p \in D'(s^*_j)} p}{\SH(s^*_j)}} + \norm{\frac{\sum_{p \in D'(s^*_j)} p}{\SH(s^*_j)} - \frac{\sum_{p \in D'(s^*_j)} p}{\size{D(s^*_j)}}} \\
            &\leq O\left(\frac{c_{G'} \sqrt{d' n \log k/\beta}}{\epsilon_{G'} \size{D'(s^*_j)}} \right) + O\left( \frac{1}{\epsilon_{\SH}\size{D'(s^*_j)}} \sqrt{n } \log n/\beta \right) \norm{\mu_j}
        \end{align*}
        We can bound $\norm{\mu_j}$ from above by $O(1)$ since the domain is of unit diameter. We can then state a simplified bound of
        \begin{align*}
            \norm{\hat{\mu}_j - \mu_j} &= O\left( \left( \frac{c_{G'}}{\epsilon_{G'}} + \frac{1}{\epsilon_{\SH}} \right) \frac{\sqrt{d' n }\log n/\beta}{\size{D' (s^*_j)}} \right).
        \end{align*}
        
        From \cref{lem:approxClustering}, we can bound the cost of cluster $D'(s^*_j)$ via $S' = \{\hat{\mu}_j : j = 1, \dots, k \}$ by the following relation
        \begin{align*}
            f_{D'(s^*_j)} (S') &\leq  f_{D'(s_j)} (\{\mu_1, \dots, \mu_k \}) + \size{D'(s_j)} \norm{\mu_j - \hat{\mu}_j}^2 \\
            &\leq O(f_{D(s^*_j)} (S^*)) + \size{D'(s_j)}O\left( \left( \frac{c_{G'}}{\epsilon_{G'}} + \frac{1}{\epsilon_{\SH}} \right)^2 \frac{d' n \log^2 n/\beta}{\size{D' (s^*_j)}^2} \right)
        \end{align*}
        For each cluster $D'(s^*_j)$, we see that if $D'(s^*_j) \geq \left( \frac{c_{G'}}{\epsilon_{G'}} + \frac{1}{\epsilon_{\SH}} \right) \sqrt{d' n }\log n/\beta $, then 
        \begin{align*}
            f_{D'(s^*_j)} (S') &\leq  O(f_{D(s^*_j)} (S^*)) + O\left( \left( \frac{c_{G'}}{\epsilon_{G'}} + \frac{1}{\epsilon_{\SH}} \right) \frac{\sqrt{d' n }\log n/\beta}{\size{D' (s^*_j)}} \right)
        \end{align*}
        On the other hand, if $D'(s) < \left( \frac{c_{G'}}{\epsilon_{G'}} + \frac{1}{\epsilon_{\SH}} \right) \sqrt{d' n }\log n/\beta$, then we have the same bound unconditionally since the diameter of the data domain is $O(1)$. Summing up over cluster over all size ranges, we get
        \begin{align*}
            f_{D'} (S') &= O(f_{D} (S^*)) + O\left(\frac{\size{S}}{\epsilon_{\CCH}} k\sqrt{n} \log n/\beta\right) + O\left( \left( \frac{c_{G'}}{\epsilon_{G'}} + \frac{1}{\epsilon_{\SH}} \right) k \sqrt{d' n }\log n/\beta \right)  \\
            &= O(\OPT) + O\left( \left(\frac{c_G}{\epsilon_{\BSO} \beta} + \frac{1}{\epsilon_{\CH}} \right) (k \poly\log n)^{1 + O(1/(2c^2-1))} \sqrt{n} \right) + O\left(\frac{\size{S}}{\epsilon_{\CCH}} \sqrt{n} \log n/\beta\right)\\
            &+ O\left( \left( \frac{c_{G'}}{\epsilon_{G'}} + \frac{1}{\epsilon_{\SH}} \right) k \sqrt{d' n }\log n/\beta \right).
        \end{align*}
    \end{proof}
    
    We can now derive the main result of this section.
    
    \fourRoundGuarantee*
    
    \begin{proof}
        To prove this theorem, we will account for all privacy loss and then scale the privacy parameters used in each data access subroutine to ensure a net $(\epsilon,\delta)$ privacy loss guarantee. We will then substitute these parameters into \cref{lem:costGuarantee} to derive the bound on the cost incurred with this choice of parameters.
        
        We see that data access occurs in 4 rounds through the following mechanisms:
        \begin{enumerate}
            \item $L$ calls in parallel to $\BTG$ to construct $\CH^l$ for $l\in [L]$ with privacy parameter $\epsilon_{\CH}$.
            \item $FLMR$ calls in parallel to $\BTG$ and $\HSO$ to construct $\BH_{l,m,r,f}$ and $\BSO_{l,m,r,f}$ for $l\in [L], m\in [M], r \in [R]$ and $f \in [F]$. The two types of calls have respective privacy parameters $\epsilon_{\BH}$ and $(\epsilon_{\BSO},\delta_{\BSO})$ (note that $\delta_{\BSO}$ occurs in our cost guarantee inside the Gaussian mechanism parameter $c_G$). Recall that during the course of our analysis we required that $\epsilon_{\BH} = \epsilon_{\BSO}$ with the observation that they were called an equal number of times.
            \item One call to $\BTG$ to construct $\CCH$ with privacy parameter $\epsilon_{\CCH}$
            \item Gaussian mechanism and one call to $\BTG$ to construct $\SH$ in parallel when computing the noisy averages over cluster sets derived from low-dimensional clustering. The respective privacy parameters are $(\epsilon_{G'},\delta_{G'})$ and $\epsilon_{\SH}$ (note that $\delta_{G'}$ occurs in our cost guarantee inside the Gaussian mechanism parameter $c_{G'}^2$).
        \end{enumerate}
        We allocate private parameters of $(\epsilon/4,0)$, $(\epsilon/4,\delta/2)$, $(\epsilon/4,0)$ and $(\epsilon/4,\delta/2)$ to each of these four steps, and sub-divide the privacy parameters within. Since 
        \begin{align*}
            FLMR &= O(\log n) \cdot O(\log n) \cdot O(\log\log n) \cdot O \left( {k \poly\log n} \right)^{O(1/(2c^2-1))} \\
            &= k^{{O(1/(2c^2-1))}} \poly\log n
        \end{align*} we can write
        \begin{align*}
            \epsilon_{\CH} &= \frac{\epsilon}{4 \log n} &\\
            \epsilon_{\BH} &= \epsilon_{\BSO} = \frac{\epsilon}{8 k^{{O(1/(2c^2-1))}} \poly\log n} & \\
            \delta_{\BSO} &= \frac{\delta}{2 k^{O(1/(2c^2-1))} \poly\log n} & \\
            \Rightarrow c_G &< {O(1/(2c^2-1))} \sqrt{\ln (n/\delta)} &\\
            \epsilon_{\CCH} &= \frac{\epsilon}{4} &\\
            \epsilon_{G'} &= \epsilon_{\SH} = \frac{\epsilon}{8}& \\
            \delta_{G'} &= \frac{\delta}{2}. & \\
            \Rightarrow c_{G'} &= O(\sqrt{\ln (1/\delta)}).
        \end{align*}
        
        Substituting these terms along with the bound 
        $$\size{S} \leq O\left(\frac{k \poly\log n}{\beta }\right)^{1 + O(1/(2c^2-1))}$$
        in the cost guarantee of \cref{lem:costGuarantee}, we get
        \begin{align*}
            f_{D'} (S') &= O(\OPT) + O\left( \left(\frac{\sqrt{\ln (n/\delta)}}{\epsilon \beta} + \frac{\log n}{\epsilon} \right) (k \poly\log n)^{1 + O(1/(2c^2-1))} \sqrt{n} \right)\\
            &+ O\left(\frac{1}{\epsilon} \sqrt{n}\left(\frac{k \poly\log n}{\beta }\right)^{1 + O(1/(2c^2-1))}\right) + O\left( \left( \frac{\sqrt{\ln (1/\delta)}}{\epsilon} \right) k \sqrt{d' n }\log n/\beta \right) \\
            &\leq O(\OPT) + O \left(\frac{1}{\epsilon} \sqrt{d' n \ln(n/\delta)} \right) \left(\frac{k \poly\log n}{\beta }\right)^{1 + O(1/(2c^2-1))}.
        \end{align*}
        
    \end{proof}

%% file: 5Appendix.tex
\section{Concentration bounds}

We recall some basic concentrations bounds that we draw upon for our proofs.

\begin{lemma}[Hoeffding's inequality]
    \label{lem:Hoeffding}
    Given $n$ i.i.d. Bernoulli random variables $X_i$ that take values in $\{0,1\}$ with mean $p$,
    \begin{align*}
        P\left(\size{\sum_{i \in [n]} X_i - np} > t\right) \leq 2\exp(-2t^2/n).
    \end{align*}
\end{lemma}

\begin{lemma}[Chernoff bound for Gaussian random variables]
    \label{lem:NormalChernoff}
    Given $n$ i.i.d. Gaussian random variables $\eta_i \sim N(0,\sigma^2)$,
    \begin{align*}
        P\left( \size{\sum_{i \in [n]} \eta_i} > t  \right) \leq 2\exp \left( \frac{-t^2}{2n\sigma^2} \right).
    \end{align*}
\end{lemma}

We will also need the following more involved concentration bound to bound the estimation error of the $\HSO$ developed later as a tool which allows us to reduce the round complexity of our protocols. We follow the formulation in \S 1.6.2 of \cite{DBLP:journals/ftml/Tropp15}, who attributes it to \cite{oliveira2009concentration} and \cite{DBLP:journals/focm/Tropp12}.

\begin{lemma}[Matrix Bernstein's inequality]\label{lem:matrixBernstein}
    Let $S_1, \dots , S_n$ be independence centered random matrices with common dimension $d_1\times d_2$ and assume that each one is uniformly bounded, i.e.
    \begin{align*}
        \Ex [S_k] = 0, \\
        \norm{S_k} \leq L \; \forall k \in [n].
    \end{align*}
    Let $Z = \sum_{k=1}^n S_k$ and $v(Z)$ denote the matrix variance statistic of $Z$ i,e,
    \begin{align*}
        v(Z) &= \max \{ \norm{\Ex (Z Z^*)}, \norm{ \Ex{Z^* Z}} \} \\
        &= \max \left\{ \norm{  \sum_{k=1}^n\Ex {S_k S_k^*} }, \norm{ \sum_{k=1}^n\Ex {S_k^* S_k}} \right\}.
    \end{align*}
    Then
    \begin{align*}
        P(\norm{Z} \geq t) \leq (d_1 + d_2) \cdot \exp \left( \frac{-t^2/2}{v(Z) + Lt/3} \right) \; \forall t\geq 0.
    \end{align*}
    Further,
    \begin{align*}
        \Ex[\norm{Z}] \leq \sqrt{2 v(Z) \log(d_1 + d_2)} + \frac{1}{3} L \log (d_1 + d_2).
    \end{align*}
\end{lemma}

\section{Bitstogram and the Heavy Sums Oracle}\label{subsec:BTGandHSO}

The contents of this subsection are used in the cost analysis for both clustering algorithms. In the sequel we make extensive use of locally private frequency estimation. For private frequency estimation a lower bound of $\Omega_{\epsilon}(\sqrt{n})$ is known \citep{DBLP:conf/esa/ChanSS12}. A state of the art construction for this problem is the $\BTG$ algorithm \cite{DBLP:journals/jmlr/BassilyNST20}, which is an $\epsilon$-LDP algorithm for the heavy-hitters problem that achieves low error.

\lembtg*

We introduce an extension of the $\BTG$ algorithm called $\HSO$ that allows us to query the sums of some vector valued function over the set of elements that map to a queried heavy-hitter value. For a given value-mapping function $f:\X \to \mathcal{V}$ and a vector-valued function $g:\X \to \R^d$ the sum estimation oracle privately returns for every heavy hitter $v \in \mathcal{V}$ the sum of all agents that map to $x$, i.e. $\sum_{p : f(p) = x} p$. We recall that $\BTG$ is a modular algorithm with two subroutines; a frequency oracle that privately estimates the frequency of any value in the data universe, and a succinct histogram construction that constructs the heavy hitters in a bit-wise manner by making relatively few calls to the frequency oracle. The construction of $\HSO$ essentially mimics the frequency oracle construction called $\mathsf{Hashtogram}$ from \cite{DBLP:journals/jmlr/BassilyNST20} and can be run in parallel with $\BTG$, allowing us to reduce the round complexity of our protocols.

\begin{algorithm}
    \caption{HeavySumsOracle}
    \label{alg:HS}
    Public randomness: Uniformly random matrix $Z\in \{\pm 1\}^{\size{\V} \times n}$\\
    Setting: Agent $j \in [n]$ holds $x_j \in \mathcal{X}$, public  functions $f:\X \to \V$, $g:\X \to [0,b]^d$, $g$ has known bounded sensitivity $\Delta_{g,2}$.\\
    For $j\in [n]$ let $y_j \leftarrow Z[f(x_j),j]\cdot g(x_j) + \eta_j $ for $\eta_j \sim N\left(0,\frac{4c^2 \Delta_{g,2}^2}{\epsilon^2}\right)$ where $c^2$ is according to \cref{lem:gauss} \\
    On input $v \in \mathcal{V}$ return $S(v) = \sum_{j \in [n]} y_j \cdot Z[v,j]$  and wait for next query
\end{algorithm}

\lemHSO*

\begin{proof}
    Let the data of the $j$th agent be denoted $x_j$, and let $y_j$ denote the value sent by the $j$th agent, i.e. $Z[f(x_j),j]\cdot (g(x_j) + \eta_j)$ where $\eta_j \sim N\left(0,\frac{c_G^2 \Delta^2}{\epsilon^2} \right)$.

    \begin{align*}
        S(v) &= \sum_{j \in [n]} y_j \cdot Z[v,j] \\
        &= \sum_{j \in [n]} (Z[f(x_j),j]\cdot g(x_j) + \eta_j) \cdot Z[v,j] \\
        \Rightarrow \Ex[S(v)] &= \sum_{j \in [n]} \Ex[ Z[f(x_j),j] \cdot Z[v,j] \cdot  (g(x_j))] +  \Ex[\eta_j \cdot Z[v,j]]. \\
        &= \sum_{j:f(x_j) = v} \Ex[g(x_j)] + \sum_{j:f(x_j) \not= v} \Ex[Z[f(v_j),j]] \Ex[Z[f(x_j),j] \cdot g(x_j)]  + 0\\
        &= \sum_{j:f(x_j) = v} g(x_j) + 0.
    \end{align*}
    This gives us that in expectation, $S(v) = \sum_{y \in D_{f(x)}} g(y)$. Now we derive high probability bounds on the estimation error. Let $S$ denote the quantity of interest, i.e. $ \sum_{j:f(x_j) = v} g(x_j)$. We have
    \begin{align*}
        \lVert S(v) - S \rVert &= \bigg\lVert \sum_{j \in [n]} (Z[f(x_j),j]\cdot g(x_j) + \eta_j) \cdot Z[v,j] - S \bigg\rVert \nonumber \\
        &\leq \bigg\lVert \sum_{j : f(x_j) \not= v} b_j g(x_j) + \sum_{j \in [n]} b_j \eta_j \bigg\rVert \nonumber \\
        &\leq \bigg\lVert \sum_{j : f(x_j) \not= v} b_j g(x_j) \bigg\rVert + \bigg\lVert  \sum_{j \in [n]} b_j \eta_j \label{eqn:normSum} \bigg\rVert,
    \end{align*}
    where we let $b_j$ denote uniformly random $\{\pm 1\}$ bits. Note that the cancellations of the summands $g(x_j)$ in $S(v)$ (where $j$ was such that $f(x_j) = v$) with $S$ were deterministic, but the Gaussian noise introduced to retain privacy remains for all agents. To bound the first summand, we observe that $b_j g(x_j)$ are at most $n$ independent vectors with $\ell_2$ norm at most $\Delta$ such that $\Ex [b_j g(x_j)] = \Ex[b_j] \Ex [g(x_j)] = 0$ and matrix variance $\max (\Ex [ \norm{\langle b_j g(x_j). b_j g(x_j)\rangle }, \Ex [ \norm{ b_j g(x_j)\otimes b_j g(x_j) }) = \norm{g(x_j)}^2 \leq \Delta^2$. Then, identifying $b_j g(x_j)$ with $S_j$ in \cref{lem:matrixBernstein} and bounding the size of the set $\{j : f(x_j) \not= v\}$ by $n$, we have that
    \begin{align*}
        P\bigg( \norm{\sum_{j=1}^n S_j} \geq t \bigg) \leq (d'+1) \exp\bigg( \frac{-t^2/2}{\Delta^2 n + \Delta t/3} \bigg) \; \forall t\geq 0.
    \end{align*}
    For an error probability of at most $\beta$, we see that it would suffice to set $t$ such that
    \begin{align*}
        (d'+1) \exp\left( \frac{-t^2/2}{\Delta^2 n + \Delta t/3} \right) &\leq \beta \\
        \Leftarrow \frac{\Delta^2 n + \Delta t/3}{t^2/2} &\leq \frac{1}{\log (d'+1)/\beta} \\
        \Leftrightarrow \frac{2\Delta^2 n}{ t^2} + \frac{2\Delta}{3t} &\leq \frac{1}{\log (d'+1)/\beta} 
    \end{align*}
    We see that it suffices to let $t > 2 \sqrt{2} \Delta \sqrt{n \log (d'+1)/\beta}$. Next, we would like to bound the second summand $\left\lVert  \sum_{j \in [n]} b_j \eta_j \right\rVert$. We have that
    \begin{align*}
        \bigg\lVert  \sum_{j \in [n]} b_j \eta_j \bigg\rVert^2 &= \bigg\langle \sum_{j \in [n]} b_j \eta_j, \sum_{j \in [n]} b_j \eta_j \bigg\rangle \\
        &= \sum_{j \in [n]} \norm{\eta_j}^2 + \sum_{j,k \in [n]} b_j b_k \left\langle \eta_j, \eta_k \right\rangle \\
        &\leq \sum_{j \in [n]} \sum_{d \in [d']} \eta_{j,d}^2 + \sum_{j,k \in [n]} b_j b_k \norm{\eta_j} \norm{\eta_k} \\
        &\leq \sum_{j \in [n]} \sum_{d \in [d']} \eta_{j,d}^2 + \sum_{j,k \in [n]} (\norm{\eta_j}^2 + \norm{\eta_k}^2) \\
        &\leq 2 \sum_{j \in [n]} \sum_{d \in [d']} \eta_{j,d}^2
    \end{align*}
    Since the upper bound is a sum of $d' n$ i.i.d. normal random variables with variance $\sigma^2 = \frac{4 c_G^2 \Delta_{g,2}^2}{\epsilon^2}$. We can now apply \cref{lem:NormalChernoff} which gives us
	\begin{align*}
	    P\left(\bigg\|\sum_{j \in [n]} b_j \eta_j\bigg\| > 2 t_2\right) \leq 2 \exp \left(\frac{ - t_2^2}{2d' n \sigma^2} \right),
	\end{align*}
	where $\sigma^2 = \frac{4 c_G^2 \Delta_{g,2}^2}{\epsilon^2}$. We again set the error probability to be $\beta/2$ to get
	\begin{align*}
	    2 \exp \left( \frac{- t_2^2}{8\sigma^2 d' n} \right) &\leq \frac{\beta}{2} \\
	    \Leftrightarrow  t_2 &\geq \sigma \sqrt{8d' n\log\frac{4}{\beta}}.
	\end{align*}
	Substituting for $\sigma$ we get the stated error bound. To see why this routine is $(\epsilon,\delta)$-differentially private, we see that the sensitivity of the response $Z[f(x_j),j]\cdot g(x_j)$ is $2\Delta_{g,2}$. The privacy guarantee is hence a direct consequence of \cref{lem:gauss}.

\end{proof}

The objects returned by $\BTG$ and $\HSO$ are often used in conjunction to estimate the average vector value for collections of data points that accumulate under some value-mapping. The consequent error bound in all these applications is formalized in the following lemma.

\newcommand{\HG}{\mathsf{HG}}
\newcommand{\SO}{\mathsf{SO}}
\begin{lemma}\label{lem:noisyAvgGuarantee}

	Given a function $f : \mathcal{X} \to \mathcal{V}$, and $g:\X \to B(0,\Delta/2) \subset \R^{d}$, if a succinct histogram $\HG : \mathcal{X} \to \mathbb{R}$ is returned by $\BTG (f,\beta,\epsilon)$ and a sum oracle $\SO : \mathcal{V} \to \mathbb{R}^d$ is returned by $\HSO(f,g,\beta,\epsilon)$, then with total probability $1-2\beta$, for every heavy hitter if $v\in \HG$ $S_v$ denotes the sum $\sum_{x \in X : f(x) = v} g(x)$ and $n_v$ denotes its frequency $|\{x \in X : f(x) = v\}|$, the following bound holds
	\begin{align*}
	    \left\lVert \frac{\SO(v)}{\HG(v)} - \frac{S_v}{n_v} \right\rVert \leq \frac{1}{n_v - \HG_E } \cdot \left( \SO_E + \HG_E \norm{\frac{S_v}{n_v}} \right).
	\end{align*}
	In the above, as per our convention, $\HG_E$ refers to the error term in the estimation of $\HG$, and $\SO_E$ refers to the error term in the estimation error of $\SO$. Note that for every heavy hitter $n_v$ we can assume without loss that $n_v \geq 2\HG_E$, that $\norm{\frac{S_v}{n_v}} = O(\Delta)$, and that $\HG_E = O((c_G \Delta/\epsilon) \sqrt{dn \log n/\beta})$, from which it follows that
	\begin{align*}
	    \left\lVert \frac{\SO(v)}{\HG(v)} - \frac{S_v}{n_v} \right\rVert \leq O\left( \frac{c_G \Delta \sqrt{d n \log n/\beta }}{\epsilon n_v} \right).
	\end{align*}
\end{lemma}

\begin{proof}
The proof is a direct consequence of the triangle inequality and some algebra.
\begin{align*}
    \left \lVert \frac{\SO(v)}{\HG(v)} - \frac{S_v}{n_v} \right\rVert &= \left\lVert \frac{\SO(v)}{\HG(v)} - \frac{S_v}{\HG(v)} + \frac{S_v}{\HG(v)} - \frac{S_v}{n_v} \right\rVert \\
    &\leq \frac{\lVert \SO(v) - S_v \rVert}{\HG(v)} + \frac{|n_v - \HG(v)|}{\HG(v)} \left\lVert \frac{S_v}{n_v} \right\rVert \\
    &\leq \frac{1}{n_v - \HG_E } \cdot \left( \SO_E + \HG_E \norm{\frac{S_v}{n_v}} \right).
\end{align*}
\end{proof}

\section{Locality Sensitive Hashing}
	
	The contents of this subsection are used only for the construction and analysis of the multi-round $k$-means algorithm with low additive error. We start by recalling the definition of an LSH family.
	
	\defLSH*
   	
   	In this work we use an LSH-family construction construction from \cite{AI06}.
   	
   	\thmLSH*
	
	Note that by scaling the input to the LSH function this gives us constructions for $(p,q,r,cr)$-sensitive LSH families for arbitrary values of $r>0$. Due to the occurrence of terms like $t^{O(t)}$ in the collision probabilities and the number of buckets, the performance of an LSH family is very sensitive to the choice of $t$. In the following lemma we show how to choose a value of $t$ for a desired ratio of $p^2(1)$ to $p(c)$.
	
	\lemLSH*
	
	\begin{proof}
	    We have that
	    \begin{align*}
	        \frac{p^2(1)}{p(c)} &\geq \frac{A^2}{8t} \frac{(1+c^2/4\sqrt{t})^{t/2}}{(1+(1/4\sqrt{t}) + 1/2t)^t} \\
	        &\geq \frac{A^2}{8t} \frac{\exp \left( c^2 / 4\sqrt{t} - c^4/32 t \right)^{t/2}}{\exp(1/4\sqrt{t} + 1/2t)^{t}} \\
	        &\geq \frac{A^2}{8t} \exp ( (c^2/8 - 1/4 )\sqrt{t} - c^4/64 + 1/2) \\
	        &= \Omega \left(  \exp ( c' \sqrt{t})/t \right).
	    \end{align*}
	    where $c' = (c^2/8 - 1/4 )$. It follows that for $t = (\log B +\log\log B)^2/(c')^2$, \begin{align*}
	        \frac{p^2 (1) }{p(c)} &\geq (c')^2 \frac{B + \log B}{(\log B +\log\log B)^2} \\
	        &= \Omega(B) \\
	        p(1) &\geq \frac{A}{2\sqrt{t}} {\exp(- \sqrt{t}/8 - 1/4)}\\
	        &= \Omega( B^{-1/c'}/\log B  )\\
	        \log N_B &= O(t \log t + \log\log n) \\
	        &= O( \log^2 B \log\log B + \log \log n).
	    \end{align*}
	\end{proof}

	In the construction of the multi-round $k$-means algorithm with low additive error, we will need to estimate the average of all points that map to a given heavy bucket. Due to the pair-wise nature of the LSH guarantee, the analysis of this requires us to use an arbitrary point from the bucket as a filter to ensure that sufficiently many points close to it and not too many points far from it map to that bucket. This result and its proof follow the lines of a similar result by \cite{nissim2018clustering}, but are modified to allow for the possibility of false positives and have been phrased differently.
	
	\lemLSHGuarantee*
    \begin{proof}
        Let $x_0$ be an arbitrary fixed point in $C$. Let $N \subset C$ be the set of points that lie near $x_0$ and collide with it under the LSH function, i.e. $N = \{ y\in D : H(y) = H(x_0), d(y,x_0) \leq r\}$. Since $\forall y\in C$, $d(x_0,y)<r$, $\Ex [|N|] \geq p(1) |C|$. We note that $|N|$ is supported on $\{0,\dots, |C|\}$ and let $p:= P(|N| \geq \frac{p(1)}{2} |C|)$. Then we can write
        \begin{align*}
            \Ex [|N|] &= \sum_{i=0}^{|C|} P(|N| = i) \cdot i \\
            &= \sum_{i=0}^{\left\lceil\frac{p(1)}{2} |C| \right\rceil-1} P(|N| = i) \cdot i + \sum_{i= \left\lceil\frac{p(1)}{2} |C| \right\rceil}^{|C|} P(|N| = i) \cdot i \\
            &\leq (1-p) \cdot \frac{p(1) |C|}{2} + p \cdot |C| \\
            \Rightarrow p(1) |C| &\leq \frac{p(1) |C|}{2} + p |C| (1 - \frac{p(1)}{2}) \\
            \Rightarrow p &\geq \frac{p(1)}{2 - p(1)}.
        \end{align*}
        Let $F \subset C$ be the set of all points which lie far from $x_0$ and collide with it under $H$, i.e. $\{ y\in D : H(y) = H(x_0), d(y,x_0) \geq cr\}$. It again follows from the LSH guarantee that $\Ex [|F|] \leq p(c) |D|\Leftrightarrow \Ex[|D \backslash F| ] \geq (1 - p(c)) |D|$. If $q:= P(|D\backslash F| \geq (1 - \frac{ 4p(c)}{p(1)}) |D|)$, then we can write
        \begin{align*}
            \Ex [|D \backslash F|] &= \sum_{i=0}^{|D|} P(|D \backslash F| = i) ] \cdot i \\
            &\leq (1-q)\cdot \left(1 - \frac{4 p(c)}{p(1)}\right) |D| + q \cdot |D| \\
            \Rightarrow (1 - p(c)) |D| &\leq (1-q)\left(1 - \frac{ 4p(c)}{p(1)}\right) |D| + q |D|\\
            \Rightarrow (1 - p(c)) |D| &\leq \left(1 - \frac{ 4p(c)}{p(1)}\right) |D| +  q \frac{4p(c)}{p(1)} \size{D} \\
            \Rightarrow \left(\frac{4}{p(1)} - 1\right) p(c)  &\leq q \frac{4p(c)}{p(1)} \\
            \Rightarrow 1 - \frac{p(1)}{4} &\leq q.
        \end{align*}
        With probability at least $1- p(1)/4$, $|F| < 4 p(c) |D| / p(1)$, and with probability at least $p(1)/2$, $\size{N} \geq p(1)\size{C}/2$. Applying the union bound on the negation of these events it follows that with probability at least $p(1)/4$ both these events hold. Conditioning on their intersection, we want to bound the distance between the average $\hat{x}_0$ of all points that collide with $x_0$ and $x_0$ itself. Let $N' = \{y\in D: H(y) = H(x_0), d(y,x_0) < cr \}$, with which definition we have $|N'| \geq |N| \geq p(1) |C|/2$. Further, the set of all points that collide with $x_0$ are partitioned by $N'$ and $F$. It follows that
        \begin{align*}
                \lVert x_0 - \hat{x}_0 \rvert &= \left\lVert x_0 - \frac{ \sum_{y \in N'} y + \sum_{y \in F} y}{|N'| + |F|} \right\rvert \\
                &\leq \frac{\sum_{y \in N'} \| x_0 - y \| + \sum_{y \in F}  \| x_0 - y \|}{|N'| + |F|} \\
                &\leq \frac{cr |N'| + \Delta |F|}{|N'| + |F|} \\
                &\leq cr + \frac{|F|}{|N'|} \Delta \\
                &\leq cr + \frac{8 p(c) |D|}{p(1)^2 |C|} \Delta.
         \end{align*}
    \end{proof}

\section{Miscellaneous tools}

In the course of our analysis, we will make extensive use of two weak triangle inequalities which hold for the distance function $d$.

\begin{lemma}[Weak triangle inequalities]
    \label{lem:triangle}
    \begin{enumerate}
        \item Given points $p,q$ and $r \in \R^d$ such that $z(q,r) \leq c \alpha^2 z(p,q)$, where $c$ is some constant and $\alpha<1$,
        \begin{align*}
            z(p,r) &\leq (1 + O(\alpha)) z(p,q).
        \end{align*}
        \item Given arbitrary points $p, q$ and $r\in \R^d$,
        \begin{align*}
            z(p,r) &\leq 2z(p,q) + 2z(q,r).
        \end{align*}
    \end{enumerate}
    
\end{lemma}

\begin{proof}
    \begin{enumerate}
        \item The bound follows from an application of the triangle inequality for the $\ell_2$ norm.
        \begin{align*}
            \sqrt{z(p,r)} &\leq \sqrt{z(p,q)} + \sqrt{z(q,r)} \\
            \Rightarrow z(p,r) &\leq z(p,q) + 2 \sqrt{z(p,q) z(q,r)} + z(q,r) \\
            &\leq z(p,q) + 2 \alpha \sqrt{c} z(p,q) + c \alpha^2 z(p,q) \\
            &\leq (1 + O(\alpha)) z(p,q).
        \end{align*}
        \item The bound follows from an application of the triangle inequality for the $\ell_2$ norm and the A.M.-G.M. inequality.
        \begin{align*}
            \sqrt{z(p,r)} &\leq \sqrt{z(p,q)} + \sqrt{z(q,r)} \\
            \Rightarrow z(p,r) &\leq z(p,q) + 2 \sqrt{z(p,q) z(q,r)} + z(q,r) \\
            &\leq 2z(p,q) + 2 z(q,r).
        \end{align*}
    \end{enumerate}
    
\end{proof}

\begin{lemma}
\label{lem:greedycover}
    Let there be a set $U$ and a family of subsets $\mathcal{S} \subset 2^{U}$ such that some subfamily $\mathcal{Z} \subset \mathcal{S}$ covers $U$, that is \[\bigcup_{z \in \mathcal{Z}} z = U.\] If we pick a collection of sets $\mathcal{C} = \{c_1,\dots,c_Y\} \subset \mathcal{S}$ where $Y = \lceil2\size{\mathcal{Z}}\log(1/\alpha)\rceil$ such that for each $i \in [Y]$
    \begin{align*}
        U_i &= U \setminus \bigcup_{j=1}^{i-1}c_j\\
        \size{c_i \cap U_i} &\ge \frac{\max_{c \in \mathcal{S}}\size{c \cap U_i}}{2}
    \end{align*}
    then $\size{\bigcup_i \in [Y] c_i} \ge (1-\alpha)\size{U}$.
\end{lemma}

We refer the reader to Lemma 2.7, \cite{chaturvedi2020differentially} for a proof of \cref{lem:greedycover}.

\begin{lemma}
    \label{lem:approxClustering}
    Given a $k$-means clustering $D'_1, \dots, D'_k$ of a data set $D'$ where 
    $$\mu_j = \frac{\sum_{p \in D'_j} p}{\size{D'_k}},$$ 
    if $\hat{\mu}_j$ is an estimate for $\mu_j$ then the $k$-means clustering cost with respect to $\{\hat{\mu}_j : j = 1,\dots, k\}$ for any cluster $D'_j$ can be bounded by
    \begin{align*}
        f_{D'_j} (\{\hat{\mu}_j : j = 1,\dots, k\}) &\leq f_{D'_j} (\mu_j) + \size{D'_j} \norm{\mu_j - \hat{\mu}_j}^2.
    \end{align*}
\end{lemma}

\begin{proof}
    First we observe that $f_{D'_j} (\{\hat{\mu}_j : j = 1,\dots, k\}) \leq f_{D'_j} ( \hat{\mu}_j)$ by definition. To get the stated bound, we perform a couple of algebraic manipulations.
    \begin{align*}
        f_{D'_j} ( \hat{\mu}_j) &= \sum_{p' \in D'_j} z(p',\hat{\mu}_j) \\
        &=  \sum_{p' \in D'_j} \norm{p' - \hat{\mu}_j}^2 \\
        &=  \sum_{p' \in D'_j} \norm{p' - {\mu}_j + \mu_j - \hat{\mu}_j }^2 \\
        &=  \sum_{p' \in D'_j} \langle p' - {\mu}_j + \mu_j - \hat{\mu}_j, p' - {\mu}_j + \mu_j - \hat{\mu}_j \rangle \\
        &=  \sum_{p' \in D'_j} \norm{p' - \mu_j}^2 + 2\langle p' - \mu_j, \mu_j - \hat{\mu}_j \rangle + \norm{\mu_j - \hat{\mu}_j}^2 \\
        &= f_{D'} (\{ \mu_j : j \in [k] \}) +  \size{D'_j} \norm{\mu_j - \hat{\mu}_j}^2.
    \end{align*}
\end{proof}